\documentclass[submission, Phys]{SciPost}

\usepackage{amsmath}
\usepackage{amsfonts}

\usepackage{amsthm}
\usepackage{multirow}
\usepackage{tikz}
\usepackage{color}
\usepackage{float}

\usepackage{graphicx}

\begin{document}

\begin{center}{\Large \textbf{
Feynman diagrams and $\Omega-$deformed M-theory
}}\end{center}

\begin{center}
J. Oh\textsuperscript{1*},
Y. Zhou\textsuperscript{2},
\end{center}

\begin{center}
{\bf 1} Department of Physics,
  University of California, Berkeley, CA 94720, U.S.A.
\\
{\bf 2} Perimeter Institute for Theoretical Physics, 31 Caroline St. N., Waterloo, ON N2L 2Y5, Canada
\\
* jihwan.oh@maths.ox.ac.uk
\end{center}

\begin{center}
\today
\end{center}


\section*{Abstract}
{\bf
We derive the simplest commutation relations of operator algebras associated to M2 branes and an M5 brane in the $\Omega$-deformed M-theory, which is a natural set-up for Twisted holography. Feynman diagram 1-loop computations in the twisted-holographic dual side reproduce the same algebraic relations.
}

\vspace{10pt}
\noindent\rule{\textwidth}{1pt}
\tableofcontents\thispagestyle{fancy}
\noindent\rule{\textwidth}{1pt}
\vspace{10pt}

\newcommand{\secref}[1]{\S\ref{#1}}
\newcommand{\figref}[1]{Figure~\ref{#1}}
\newcommand{\appref}[1]{Appendix~\ref{#1}}
\newcommand{\tabref}[1]{Table~\ref{#1}}
\def\ie{\begin{equation}\begin{aligned}}
\def\fe{\end{aligned}\end{equation}}
\def\Id{{\mathbf{I}}} 
\def\lra{\leftrightarrow}
\def\lea{\leftarrow}
\def\ria{\rightarrow}
\newcommand{\vev}[1]{{\left\langle {#1} \right\rangle}}
\newcommand{\bR}{\mathbb{R}}
\newcommand{\bZ}{\mathbb{Z}}
\newcommand{\bC}{\mathbb{C}}
\newcommand{\cG}{\mathcal{G}}
\newcommand{\cM}{\mathcal{M}}
\newcommand{\cN}{\mathcal{N}}
\newcommand{\cL}{\mathcal{L}}
\newcommand{\A}{{\alpha}}
\newcommand{\B}{{\beta}}
\newcommand{\D}{{\delta}}
\newcommand{\E}{{\epsilon}}

  \newtheorem{theorem}{Theorem}
  \newtheorem{proposition}{Proposition}
  \newtheorem{lemma}{Lemma}
  \newtheorem{corollary}{Corollary}
  \newtheorem{conjecture}{Conjecture}

\newcommand{\pa}{\partial}
\newcommand{\la}{\langle}
\newcommand{\ld}{\lambda}
\newcommand{\ra}{\rangle}
\newcommand{\cA}{\mathcal{A}}
\newcommand{\cB}{\mathcal{B}}
\newcommand{\cP}{\mathcal{P}}
\newcommand{\C}{\mathbb{C}}
\newcommand{\R}{\mathbb{R}}
\newcommand{\Z}{\mathbb{Z}}
\newcommand{\CP}{\mathbb{CP}}
\newcommand{\cC}{{\mathcal C}}
\newcommand{\cD}{{\mathcal D}}
\newcommand{\cF}{{\mathcal F}}
\newcommand{\cI}{{\mathcal I}}
\newcommand{\cJ}{{\mathcal J}}
\newcommand{\cK}{{\mathcal K}}
\newcommand{\cO}{{\mathcal O}}
\newcommand{\cR}{{\mathcal R}}
\newcommand{\cS}{{\mathcal S}}
\newcommand{\cZ}{{\mathcal Z}}
\newcommand{\cW}{{\mathcal W}}
\newcommand{\cV}{{\mathcal V}}

\def\xTB{$\times$} 
\section{Introduction and Conclusions}\label{sec:Intro}
In \cite{Costello:2016mgj}, Costello and Li developed a beautiful formalism, which prescribes a way to topologically twist supergravity. Combining with the classical notion of topological twist of supersymmetric quantum field theory \cite{Witten:1988ze,Witten:1988xj}, we are now able to explore a topological sector for both sides of AdS/CFT correspondence. It was further suggested in \cite{Costello:2016nkh} a systematic method of turning an $\Omega$-background, which plays an important roles \cite{Nekrasov:2002qd,Alday:2009aq,Nekrasov:2009rc,Nekrasov:2010ka,Yagi:2014toa,Nekrasov:2015wsu} in studying supersymmetric field theories, in the twisted supergravity.

A topological twist along with $\Omega$-deformation enables us to study a particular protected sub-sector of a given supersymmetric field theory \cite{Oh:2019bgz,Jeong:2019pzg,Beem:2018fng,Oh:2019mcg}, which is localized not only in the field configuration space but also in the spacetime. Interesting dynamics usually disappear {along} the way, but as a payoff, we can make a more rigorous statement on the operator algebra. 

{The topological holography is an exact isomorphism between the operator algebras of gravity and field theory. In this paper, we will focus on a particular example of topological holography: the correspondence of the operator algebra of M-theory on a certain background parametrized by $\E_1$, $\E_2$, which localizes to 5d non-commutative $U(K)$ Chern-Simons theory with non-commutativity parameter $\E_2$ \footnote{{The 5d CS theory that appears in this paper is always meant to be a certain variant of the usual 5d CS theory with a topological-holomorphic twist and with non-commutativity turned on in the holomorphic directions.}}, and the operator algebra of the worldvolume theory of M2-brane, which localizes to 1d topological quantum mechanics(TQM). In particular, \cite{Costello:2017fbo} proved that two operator algebras are Koszul dual \cite{Costello:2017fbo} to each other.}

The important first step of the proof was to impose the BRST-invariance of the 5d $U(K)$ CS theory coupled with the 1d TQM. The 5d CS theory is a renormalizable, and self-consistent theory \cite{Costello:2015xsa}. However, in the presence of the topological defect that couples the 1d TQM and the 5d CS theory, certain Feynman diagrams turn out to have non-zero BRST variations. For the combined, interacting theory to be quantum mechanically consistent, the BRST variations of the Feynman diagrams should combine to give zero. This procedure magically reproduces the algebra commutation relations that define 1d TQM operator algebra, $\cA_{\E_1,\E_2}$. Intriguingly, one can extract non-perturbative information in the protected operator algebra from a perturbative calculation.

In fact, both the algebra of local operators in 5d CS theory and the 1d TQM operator algebra $\cA_{\E_1,\E_2}$ are deformations of the universal enveloping algebra of the Lie algebra $\mathrm{Diff}_{\E_2}(\bC)\otimes \mathfrak{gl}_K$ over the ring $\bC[\![\E_1]\!]$. Deformation theory tells us that the space of deformations of $U(\mathrm{Diff}_{\E_2}(\bC)\otimes \mathfrak{gl}_K)$ is the second Hochschild cohomology $\mathrm{HH}^2(U(\mathrm{Diff}_{\E_2}(\bC)\otimes \mathfrak{gl}_K))$. Although this Hochschild cohomology is known to be hard to compute, there is still a clever way of comparing these two deformations \cite{Costello:2017fbo}: notice that both of the algebras are defined compatibly for super groups $\mathrm{GL}_{K+R|R}$, {and their deformations are compatible with transition maps $\mathrm{GL}_{K+R|R}\hookrightarrow\mathrm{GL}_{K+R+1|R+1}$}, so there are induced transition maps between Hochschild cohomologies $\mathrm{HH}^2(U(\mathrm{Diff}_{\E_2}(\bC)\otimes \mathfrak{gl}_{K+R+1|R+1}))\to \mathrm{HH}^2(U(\mathrm{Diff}_{\E_2}(\bC)\otimes \mathfrak{gl}_{K+R|R}))$, hence the equivalence class of deformations are actually elements in the limit
\ie
\underset{R\ria\infty}{\lim}\mathrm{HH}^2(U(\mathrm{Diff}_{\E_2}(\bC)\otimes \mathfrak{gl}_{K+R|R}))
\fe
and the limit is well-understood \footnote{The actual computation in \cite{Costello:2017fbo} is more subtle, and will not be used in this work.}, it turns out that the space of all deformations is essentially one-dimensional: a free module over $\bC[\kappa]$ where $\kappa$ is the central element $1\otimes \mathrm{Id}_K$. Hence the algebra of local operators in 5d CS theory and the 1d TQM operator algebra are isomorphic up to a $\kappa$-dependent reparametrization
\ie
\hbar\mapsto \sum_{i=1}^{\infty}f_i(\kappa)\hbar^i
\fe
where $f_i(\kappa)$ are polynomials in $\kappa$.

Later, in \cite{Gaiotto:2019wcc} the same algebra with $K=1$ was defined using the gauge theory approach, and a combined system of M2-branes and M5-branes were studied. Especially, \cite{Gaiotto:2019wcc} interpreted the degrees of freedom living on M5-branes as forming a bi-module $\cM_{\E_1,\E_2}$ of the M2-brane operator algebra, and suggested the evidence by going to the mirror Coulomb branch algebra \cite{Bullimore:2015lsa,Braverman:2016wma} and using the known Verma module structure of massive supersymmetric vacua \cite{Bullimore:2016nji,Bullimore:2016hdc}. Appealing to the brane configuration in type IIB frame, they argued a triality in the M2-brane algebra, which can also be deduced from its embedding in the larger algebra, affine $gl(1)$ Yangian \cite{Tsymbaliuk:2014fvq,Prochazka:2015deb,Kodera:2016faj,Gaberdiel:2017dbk}.

Crucially, \cite{Gaiotto:2019wcc} noticed $U(1)$ CS should be treated separately from $U(K)$ CS theory with $K>1$, since the algebras differ drastically and the ingredients of the Feynman diagram are different in $U(1)$ CS, due to the non-commutativity. As a result, operator algebra isomorphism should be re-assessed.

Our work was motivated by the observation, and we will solve the following problem in a part of this paper.
\begin{itemize}
\item{}The simplest algebra $\cA_{\E_1,\E_2}$ commutator, which has $\E_1$ correction.
\end{itemize}
The problem will be solved by two complementary methods:
\begin{itemize}
    \item[(1)] Direct calculation by definition.
    \item[(2)] Using Feynman diagrams whose non-trivial BRST variation lead to the commutator.
\end{itemize}

Next, we will make the first attempt to derive the bi-module structure from the 5d $U(1)$ CS theory, where the combined system of the M2-branes and the M5-brane is realized as the 1d TQM and the $\beta-\gamma$ system\footnote{{One way to understand the appearance of $\beta-\gamma$ system is to go to type IIA frame, where the M5-brane maps to a D4 brane and the 11d supergravity background maps to a D6-brane. D4-D6 strings form 4d $\cN=2$ hypermultiplet. Under the $\Omega$-background, the 4d $\cN=2$ hypermultiplet localizes to $\B\gamma$ system \cite{Oh:2019bgz}.}}. Especially, we will answer the following problems.
\begin{itemize}
\item{}The commutator of the simplest bi-module $\cM_{\E_1,\E_2}$ of $\cA_{\E_1,\E_2}$, which has $\E_1$ correction.
\end{itemize}
Again, this will be solved by two complementary ways:
\begin{itemize}
    \item[(1)] Direct calculation by definition.
    \item[(2)] Using Feynman diagrams whose non-trivial BRST variation leads to the commutator of bi-module $\cM_{\E_1,\E_2}$.
\end{itemize}

Our work is only a part of a bigger picture. The algebra $\cA_{\E_1,\E_2}$ is a sub-algebra of affine $gl(1)$ Yangian \cite{Gaiotto:2019wcc}, and there exists a closed-form formula for the most general commutators, which can be derived from affine $gl(1)$ Yangian. One can try to derive the commutators from 5d $U(1)$ CS theory Feynman diagram computation.

Going to type IIB frame, the brane configurations map to Y-algebra configuration \cite{Gaiotto:2017euk}. Here, the general M2-brane algebra is formed by the co-product of three different M2-brane algebras related by the triality. {The local operators supported on M5-branes form a generalized $\cW_{1+\infty}$  algebra\cite{Gaiotto:2017euk}. The $\beta-\gamma$ Vertex Algebra that our M5-brane supports is the simplest example of this generalized $\cW_{1+\infty}$.} Hence, we are curious if our story can be further generalized to the coupled system of the 5d $U(1)$ CS theory and the generalized $\cW_{1+\infty}$ algebra.
\subsection{Structure of the paper}
After reviewing the general concepts in section \secref{sec:review}, we show the following algebra commutator in \secref{subsec:M2}.
\ie
\left[t[2,1],t[1,2]\right]_{\E_1}=\E_1\E_2t[0,0]+\E_1\E_2^2t[0,0]t[0,0]
\fe
where $[\bullet]_{\E_1}$ is the $\cO(\E_1)$ part of $[\bullet]$, $t[m,n]\in\cA_{\E_1,\E_2}$. The detail of the proof is shown in \appref{app:alg}. The commutation relation was successfully checked by 1-loop Feynman diagram associated to 5d CS theory and 1d TQM. This is the content of section \secref{sec:pert}. We collected some intermediate integral computations used in the Feynman diagram in \appref{app:2-1}.

Next, we show the following algebra-bi-module commutator in \secref{subsec:M5}.
\ie
\left[t[2,1],b[z^1]c[z^0]\right]_{\E_1}={\E_1\E_2t[0,0]b[z^0]c[z^0]}+\E_1\E_2b[z^0]c[z^0]
\fe
where $b[z^m]$, $c[z^m]\in\cM_{\E_1,\E_2}$. The detail of the proof can be found in \appref{app:bimod}. We reproduced the commutation relation using the 1-loop Feynman diagram computation in the 5d CS theory, 1d TQM, and 2d $\B\gamma$ coupled system. This is the content of section \secref{sec:5d2dpert}. We collected some intermediate integral computations used in the Feynman diagram in \appref{app:2-2} and \appref{app:2-3}.

{\it{Note added: recently, complete commutation relations for the algebra $\cA_{\E_1,\E_2}$ was proposed in \cite{Gaiotto:2020vqj}.}}
\section{Twisted holography via Koszul duality}\label{sec:review}
Twisted holography is the duality between the protected sub-sectors of full supersymmetric AdS/CFT \cite{Maldacena:1997re,Gubser:1998bc,Witten:1998qj}, obtained by a topological twist and $\Omega$-background both turned on in the field theory side and supergravity side. The most glaring aspect of twisted holography\footnote{A similar line of development was made in \cite{Bonetti:2016nma,Mezei:2017kmw}, using twisted $\mathbb{Q}$-cohomology, where  $\mathbb{Q}$ is a particular combination of a supercharge Q and a conformal supercharge S \cite{Beem:2013sza}. In the sense of \cite{Oh:2019bgz}, $\mathbb{Q}$-cohomology is equivalent to $Q_V$-cohomology, where $Q_V$ is the modified scalar supercharge in $\Omega-$deformed theories.} is an correspondence between operator algebra in both sides, which is manifested by a rigorous Koszul duality.  Moreover, the information of physical observables such as Witten diagrams in the bulk side that match with correlation functions in the boundary side is fully captured by OPE algebra in the twisted sector \cite{Gaiotto:2019mmf}. 

This section is prepared for a quick review of twisted holography for non-experts. The idea was introduced in \cite{Costello:2016mgj} and studied in various examples \cite{Costello:2016nkh,Costello:2017fbo,Ishtiaque:2018str,Costello:2018zrm,Gaiotto:2019wcc,Costello:2020jbh} with or without $\Omega$-deformation. The reader who is familiar with \cite{Costello:2016nkh} can skip most of this section, except for \secref{subsec:omegad}, \secref{subsec:opalgcomp}, and \secref{subsec:M5O}, where we set up the necessary conventions for the rest of this paper. These subsections can be skipped as well, if the reader is familiar with \cite{Gaiotto:2019wcc}. Also, see a complementary review of the formalism in the section 2 of \cite{Gaiotto:2019wcc}.

After defining the notion of twisted supergravity in \secref{subsec:twisted}, we will focus on a particular (twisted and $\Omega-$deformed) M-theory background on $\bR_t\times\bC_{NC}^2\times\bC_{\E_1}\times\bC_{\E_2}\times\bC_{\E_3}$, where $NC$ means non-commutative, and $\E_i$ stands for $\Omega-$background related to $U(1)$ isometry with a deformation parameter $\E_i$ in \secref{subsec:omegad}. N $M2$ branes extending $\bR_t\times\bC_{\E_1}$ leads to the field theory side. As we will explain in \secref{subsec:opalgcomp}, a bare operator algebra isomorphism between twisted supergravity and twisted M2-brane worldvolume theory is given by an interaction Lagrangian between two systems. Due to this interaction, a perturbative gauge anomaly appears in various Feynman diagrams, and a careful cancellation of the anomaly will give a consistent quantum mechanical coupling between two systems. Strikingly, the anomaly cancellation condition itself leads to a complete operator algebra isomorphism, by fixing algebra commutators. This will be described in \secref{subsec:anomaly}. To discuss holography, it is necessary to include the effect of taking a large N limit and the subsequent deformation in the spacetime geometry. We will illustrate the concepts in \secref{subsec:backreac}. In \secref{subsec:M5O}, we will explain how to introduce M5-brane in the system and describe the role of M5-brane in the gravity and field theory side. In short, the degree of freedom on M5-brane will form a module of the operator algebra of M2-brane. Similar to the M2-brane case, the anomaly cancellation condition for M5-brane uniquely fixes the structure of the module. 
\subsection{Twisted supergravity}\label{subsec:twisted}
Before discussing the topological twist of supergravity, it would be instructive to recall the same idea in the context of supersymmetric field theory and make an analog from the field theory example.

Given a supersymmetric field theory, we can make it topological by redefining {the generator of the rotation symmetry $M$  using the generator of the R-symmetry $R$.}
\ie
M\quad\ria\quad M'=M+R
\fe
As a part of Lorentz symmetry is redefined, supercharges, which were previously spinor(s), split into a scalar $Q$, which is nilpotent 
\ie\label{nilpotencycond}
Q^2=0,
\fe
and a 1-form $Q_\mu$. Because of the nilpotency of $Q$, one can define the notion of Q-cohomology. 

Following anti-commutator explains the topological nature of the operators in Q-cohomology-- a translation is Q-exact.
\ie\label{trivialtr}
\{Q,Q_\mu\}=P_\mu
\fe 
To go to the particular Q-cohomology, one needs to turn off all the infinitesimal super-translation $\E_Q$ except for the one that parametrizes the particular transformation $\D_Q$ generated by $Q$.

More precisely, if we were to start with a gauge theory, which is quantized with BRST formalism, the physical observables are defined as BRST cohomology, with respect to some $Q_{BRST}$. The topological twist modifies $Q_{BRST}$, and the physical observables in the resulting theory are given by $Q'_{BRST}$-cohomology.
\ie
Q_{BRST}\quad\ria\quad Q'_{BRST}=Q_{BRST}+Q
\fe

As an example, consider $3d$ $\cN=4$ supersymmetric field theory. The Lorentz symmetry is $SU(2)_{Lor}$ and R-symmetry is $SU(2)_H\times SU(2)_C$, where H stands for Higgs and C stands for Coulomb. There are two ways to re-define the Lorentz symmetry algebra, and we choose to twist with $SU(2)_C$, as this will be used in the later discussion. In other words, one redefines
\ie
M\quad\ria\quad M'=M+R_C
\fe
The resulting scalar supercharge is obtained by identifying two spinor indices, one of Lorentz symmetry $\A$ and one of $SU(2)_C$ R-symmetry $a$
\ie
Q^\A_{a\dot{a}}\quad\ria\quad Q^a_{a\dot{a}}\fe
and taking a linear combination.
\ie
Q=Q^+_{1\bar{1}}+Q^-_{1\bar{2}}
\fe
This twist is called Rozansky-Witten twist \cite{Rozansky:1996bq} and will be used in twisting our M2-brane theory.

One way to start thinking about the topological twist of supergravity is to consider a brane in the background of the ``twisted'' supergravity. If one places a brane in a twisted supergravity background, it is natural to guess that the worldvolume theory of the brane should also be topologically twisted coherently with the prescribed twisted supergravity background.

Given the intuition, let us define twisted supergravity, following \cite{Costello:2016mgj}. In supergravity, the supersymmetry is a local(gauge) symmetry, a fermionic part of super-diffeomorphism. {As usual in gauge theories, one needs to take a quotient by the gauge symmetry, and this is done by introducing a ghost field. As supersymmetry is a fermionic symmetry, the corresponding ghost field is a bosonic spinor, $q$. Twisted supergravity is defined as supergravity in a background where the bosonic ghost $q$ takes a non-zero value.}

{It is helpful to recall how we twist a field theory to have a better picture for presumably unfamiliar non-zero bosonic ghost.} One can think the infinitesimal super-translation parameter $\E$ that appears in the global supersymmetry transformation as a rigid limit of the bosonic ghost $q$. For instance, in 4d $\cN=1$ holomorphically twisted field theory \cite{Nekrasov:1996,Johansen:1995,Costello:2011np,Saberi:2019ghy}, with Q paired with $\E_+$, the supersymmetry transformation of the bottom component $\phi$ of anti-chiral superfield $\bar{\Psi}=(\bar{\phi},\bar{\psi},\bar{F})$ transforms as
\ie
\D\phi=\bar{\E}\bar{\psi},\quad \D\bar{\psi}=i\E_+\bar{\pa}\bar\phi+i\E_-\pa\bar{\phi}+\bar\E\bar F
\fe
As we focus on Q-cohomology, we set $\E_+=1$, $\E_-=\bar{\E}=0$, then the equations reduce into
\ie
\D\bar\phi=0,\quad\D\bar\psi=i\bar\pa\bar\phi
\fe
In the similar spirit, in the twisted supergravity, we control the twist by giving non-zero VEV to components of the bosonic ghost $q$.

 Indeed, \cite{Costello:2016mgj} proved that by turning on non-zero bosonic spinor vacuum expectation value $\la q\ra\neq0$ with $q_\A\Gamma^{\A\B}_\mu q_\B=0$ for a vector gamma matrix, one can obtain the effect of topological twisting. We can now compare with the field theory case above \eqref{nilpotencycond}: $Q^2=0$ with $Q\neq0$. One can think of $\E_Q$ as a rigid limit of q. 
 
The operator algebra of twisted type IIB supergravity is isomorphic to that of Kodaira-Spencer theory \cite{Bershadsky:1993cx}. The following diagram gives a pictorial definition of the two theories, which turned out to be isomorphic to each other.
\begin{figure}[H]
\centering
\includegraphics[width=11cm]{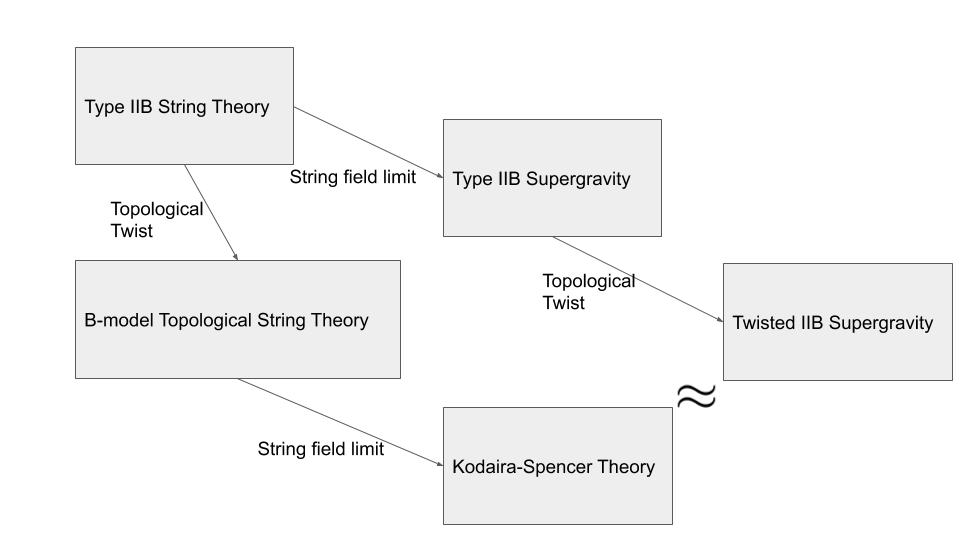}
\caption{Starting from type IIB string theory, one can obtain the same theory by taking two operations-- 1. String field limit, 2. Topological twist-- in any order.}
\label{Relations}
 \centering
\end{figure}
\noindent
Notice that the topological twist in the first column of the picture is the twist applied on the worldsheet string theory\footnote{We thank Kevin Costello, who pointed out that the arrow from Type IIb string theory to B-model topological string theory is still mysterious in the following sense. In {Ramond-Ramond} formalism, as the super-ghost is in the Ramond sector and it is hard to give it a VEV.  In the Green-Schwarz picture surely it should work better, but there are still problems there, as the world-sheet is necessarily embedded in space-time whereas in the B model that is not allowed.}, whereas in the second column is the twist on the target space theory.

Lastly,  there are two types of twists available: a topological twist and a holomorphic twist, and it is possible to turn on the two different types of twists in the two different directions of the spacetime. The mixed type of twists is called a topological-holomorphic twist, for example, \cite{Kapustin:2006hi}. Different from a topological twist, a holomorphic twist makes only the (anti)holomorphic translation to be Q-exact; after the twist we have $Q$ and $Q_z$ such that
\ie
{\{Q,Q_z\}=P_{\bar z}}
\fe
Hence, {the holomorphic translation} is physical(not Q-exact), and there exist non-trivial dynamics arising from this. \cite{Costello:2016mgj,Costello:2016nkh} showed that it is possible to discuss a holomorphic twist in the supergravity. It is important to have a holomorphic direction to keep the non-trivial dynamics, as we will later see.

\subsection{$\Omega$-deformed M-theory}\label{subsec:omegad}
Similar to the previous subsection, we will start reviewing the notion of $\Omega$-deformation of topologically twisted field theory. To define $\Omega$-background, one first needs an isometry, typically $U(1)$, generated by some vector field $V$ on a plane where one wants to turn on the $\Omega$-background. $\Omega$-deformation is a deformation of topologically twisted field theory. Physical observables are in the modified $Q_V$ cohomology, which satisfies
\ie\label{qvcoh}
Q_V^2=L_V,\quad\text{where }Q_V=Q+i_{V^\mu}Q_\mu
\fe
where $L_V$ is a conserved charge associated with $V$, and $i_{V^\mu}$ is a contraction with the vector field $V^\mu$, reducing the form degree by 1.

As the RHS of \eqref{qvcoh} is non-trivial, $Q_V$ cohomology only consists of operators, which are fixed by the action of $L_V$ such that $L_V\cO=0$. Hence, effectively, the theory is defined in two fewer dimensions, if the isometry group is $U(1)$. More generally, one can turn on $\Omega$-background in the $n$ planes, and the dynamics of the original theory defined on $D$-dimensions localizes on $D-2n$ dimensions.

\cite{Costello:2016nkh} proposed a prescription for turning $\Omega$-background in twisted 11d supergravity; we need a 3-form field $\E C$, along with $U(1)$ isometry generated by a vector field $\E V$, where $\E$ is a constant,  measuring the deformation. Similar to the field theory description, in this background$(\la q\ra,C\neq0)$, the bosonic ghost $q$ squares into the vector field, $\E V$ to satisfy the 11d supergravity equation of motion.
\ie
q^2=q_\A(\Gamma^{\A\B})_\mu q_\B=\E V_\mu
\fe
The $\Omega$-background localizes the supergravity field configuration into the fixed point of the $U(1)$ isometry. {From now on, we will call $\Omega$-background with parametrized by $\E_i$ as $\Omega_{\E_i}$ background. }More generally, one can turn on multiple $\Omega_{\E_i}$-backgrounds in the separate 2-planes, which we will denote as $\bC_{\E_i}$.

{The topologically twisted and $\Omega-$deformed 11d background} that we will focus in this paper is
\ie
\text{11d SUGRA: }\bR_t\times\bC^2_{NC}\times\bC_{\E_1}\times TN_{1;\E_2,\E_3}
\fe
where $TN_{1;\E_2,\E_3}$ is Taub-NUT space, which can be thought of as $S^1_{\E_2}\times\bR\times\bC_{\E_3}$. The twist is implemented with the bosonic ghost chosen such that {holomorphic} twist in $\bC^2_{NC}$ directions \footnote{NC stands for Non-Commutative. This will become clear in the type IIa frame.} and {topological} twist in $\bR_t\times\bC_{\E_1}\times TN_{1;\E_2,\E_3}$ directions\footnote{As remarked, if one introduces branes, the worldvolume theory inherits the particular twist that is turned on in the particular direction that the branes extend.}. The 3-form is
\ie
C=V^d\wedge d\bar z_1\wedge d\bar z_2
\fe
where $V^d$ is 1-form, which is a Poincare dual of the vector field $V$ on $\bC_{\E_2}$ plane, and $z_1,z_2$ are holomorphic coordinates on $\bC^2_{NC}$. 

The twisted holography is the duality between the protected subsector of M2-brane and the localized supergravity, due to the $\Omega$-background. We first want to introduce $M2$ branes and establish the explicit isomorphism at the level of operator algebras. Place $N$ M2-branes on 
\ie
\text{M2-brane: }\bR_t\times\{\cdot\}\times\bC_{\E_1}\times\{\cdot\}
\fe
For the concrete computation, it is convenient to go to type IIa frame by reducing along an M-theory circle. We pick the M-theory circle as $S^1_{\E_2}$, which is in the direction of the vector field $V$.\footnote{For a different purpose, to make contact with Y-algebra system, type IIb frame is better, but we will not pursue this direction in this paper.}

After reducing on $S^1_{\E_2}$, the Taub-NUT geometry maps into one D6-brane and N M2-branes map to N D2-branes.
\ie
\text{type IIa SUGRA : }&\bR_t\times\bC^2_{NC}\times\bC_{\E_1}\times\bR\times\bC_{\E_3}\\
\text{D6-brane : }&\bR_t\times\bC^2_{NC}\times\bC_{\E_1}\\
\text{D2-branes : }&\bR_t\times\quad~\quad\times\bC_{\E_1}
\fe
and 3-form C-field reduces into a B-field, which induces a non-commutativity $[z_1,z_2]=\E_2$ on $\bC^2_{NC}$.
\ie
B=\E_2 d\bar z_1\wedge d\bar z_2
\fe
There are two types of contributions to gravity side: 1. closed strings in type IIa string theory and 2. open strings on the D6-brane. It was shown in \cite{Costello:2016nkh} that we can completely forget about the closed strings. The reason is in the presence of the non-commutativity,the holomorphically twisted supergravity background(B-model) is the same as the topologically twisted background(A-model) equipped with a B-field. As we are working in the supergravity limit, where there is no instanton effect, we can also ignore the effect from a B-field. Hence, for closed string, the background becomes topological A-model, which is trivial. Therefore, the open strings from the D6-brane entirely capture gravity side. 

D6-brane worldvolume theory is 7d SYM, and it localizes on 5d non-commutative $U(1)$ Chern-Simons on $\bR_t\times\bC^2_{NC}$ due to $\Omega_{\E_1}${-background} on $\bC_{\E_1}$ \cite{Costello:2018txb}. The 5d Chern-Simons theory is not the typical Chern-Simons theory, as it inherits a topological twist in $\bR_t$ direction and a holomorphic twist in $\bC_{NC}^2$ direction, in addition to the non-commutativity. As a result, a gauge field only has 3 components 
\ie
A=A_tdt+A_{\bar z_1}d\bar z_1+A_{\bar z_2}d\bar z_2
\fe
and the action takes the following form.
\ie
S=\frac{1}{\E_1}\int_{\bR_t\times\bC^2_{NC}}dz_1dz_2\bigg(A\star dA+\frac{2}{3}A\star A\star A\bigg)
\fe
The star product $\star_{\E_2}$ is the standard Moyal product induced from the non-commutativity of $\bC^2_{NC}$: $[z_1,z_2]=\E_2$. The Moyal product between two holomorphic functions\footnote{{The Moyal product is extended to a product on the Dolbeault complex $\Omega^{0,*}(\bC^2)$ by the same formula, except that the product between two functions becomes a wedge product between two forms.}} $f$ and $g$ is defined as
\ie
f\star_{\E} g=fg+\E\frac{1}{2}\E_{ij}\frac{\pa}{\pa z_i}f\frac{\pa}{\pa z_j}g+\E^2\frac{1}{2^22!}\E_{i_1j_1}\E_{i_2j_2}\left(\frac{\pa}{\pa z_{i_1}}\frac{\pa}{\pa z_{i_2}}f\right)\left(\frac{\pa}{\pa z_{j_1}}\frac{\pa}{\pa z_{j_2}}g\right)+\cO(\epsilon^3)
\fe

The gauge transformation $\Lambda\in\Omega^0(\bR\times\bC^2_{NC})\otimes gl_1$\footnote{{$gl_1$ Lie algebra factor comes from the simple fact that the theory is $U(1)$ gauge theory. For now, there is no essential difference between $\Omega^0(\bR\times\bC)^2_{NC}$ and $\Omega^0(\bR\times\bC)^2_{NC}\otimes gl_1$; however, having $gl_1$ rather than $gl_K$ makes a huge difference in the Feynman diagram computation, which will be discussed in \secref{sec:pert}.}} acting on the gauge field $A$ is
\ie
A\mapsto A+d\Lambda+[\Lambda,A],~\text{where}~[\Lambda,A]=\Lambda\star_{\E_2}A-A\star_{\E_2}\Lambda
\fe 

The field theory side is defined on $N$ D2-branes, which extend on $\bR_t\times\bC_{\E_1}$. This is the 3d $\cN=4$ gauge theory with 1 fundamental hypermultiplet and 1 adjoint hypermultiplet. {Since the D2-branes are placed on a topologically twisted supergravity background,} the theory inherits the topological twist, which is the Rozansky-Witten twist. We will work on $\cN=2$ notation, then each of $\cN=4$ hypermultiplet splits into a chiral and an anti-chiral $\cN=2$ multiplet. We denote the scalar bottom component of the fundamental chiral and anti-chiral multiplet as $I_a$ and $J^a$, and that of adjoint multiplets as $X^a_b$ and $Y^a_b$, where $a$ and $b$ are $U(N)$ gauge indices. {Those scalars parametrize the hyper-Kahler target manifold $\cM$, which has a non-degenerate holomorphic symplectic structure. This structure turns the ring of holomorphic functions on $\cM$ into a Poisson algebra with the} following basic Poisson brackets:
\ie
\{I_a,J^b\}=\D^b_a,\quad\{X^a_b,Y^c_d\}=\D^a_d\D^c_b
\fe
It is known that {the gauge-invariant combinations of }Q-cohomology of Rozansky-Witten twisted $\cN=4$ theory {is equivalent to the Higgs branch chiral ring.} The elements of Higgs branch chiral ring are gauge invariant polynomials of $I$, $J$, $X$, and $Y${:}
\ie\label{words}
IS(X^mY^n)J,\quad \text{Tr}S(X^mY^n)
\fe
where ${S(\bullet)}$ means fully symmetrized polynomial of the monomial $\bullet$.

Upon imposing the F-term relation\footnote{{Physically, one needs to impose the F-term relation, as it is a part of defining condition for the supersymmetric vacua, as a critical locus of our specific 3d $\cN=4$ superpotential. Algebraically, the F-term relation forms an ideal of the ring of holomorphic functions on $\cM$.}}
\ie
\left[X,Y\right]+IJ=\E_2\D,
\fe
one can show two words in \eqref{words} are equivalent\footnote{They are related by following relation:
\ie
IS{(X^mY^n)}J=\E_2\text{TrS}{(X^mY^n)}
\fe
} up to a factor of $\E_2$\footnote{{Note that the $\E_2$ factor, which was previously introduced as a measure for the non-commutativity in the 5d CS theory, acts as an FI parameter in the 3d $\cN=4$ gauge theory.}}, and the physical operators purely consist of one of them. Let us call them as
\ie\label{algebraelem}
t[m,n]=\frac{1}{\E_1}\text{Tr}SX^mY^n
\fe
{In the $\Omega_{\E_1}$-background, the Higgs branch chiral ring is quantized to an algebra and the support of the operator algebra in 3d $\cN=4$ theory also localizes to the fixed point of the $\Omega_{\E_1}$-background.} Therefore, the theory effectively becomes 1d TQM(Topological Quantum Mechanics) \cite{Dedushenko:2016jxl,Beem:2016cbd,Bullimore:2016hdc}.

In summary, two sides of twisted holography are 5d non-commutative Chern-Simons theory and 1d TQM. Until now, we have not quite taken a large $N$ limit and resulting back-reaction that will deform the geometry. The large $N$ limit will be crucial for the operator algebra isomorphism to work and we will illustrate this point in section \secref{subsec:backreac}.
\subsection{Comparing elements of operator algebra}\label{subsec:opalgcomp}
As 5d CS theory has a trivial equation of motion: $F=0$, all the operators have positive ghost numbers. Also, since $\bR_t$ direction is topological, the fields do not depend on $t$. As a result, operator algebra consist of ghosts $c(z_1,z_2)$ with holomorphic dependence on coordinates of $\bC_{NC}^2$, $z_1$, $z_2$. The elements are then Fourier modes of the ghosts.
\ie\label{elements}
{c[m,n]=\pa_{z_1}^m\pa_{z_2}^nc(0,0)}
\fe
The non-commutativity in $\bC_{NC}^2$ planes induces an algebraic structure in the holomorphic functions on $\bC^2_{NC}$ defined by the Moyal product.
\ie\label{commut}
\left[z_1^az_2^b,z_1^cz_2^d\right]=(z_1^az_2^b)\star_{\E_2}(z_1^cz_2^d)-(z_1^cz_2^d)\star_{\E_2}(z_1^az_2^b)=\sum_{m,n}f^{m,n}_{a,b;c,d}z_1^mz_2^n
\fe
{At the classical level ($\E_1=0$), the operator algebra $A_{\E_1=0,\E_2}$ of 5d CS theory is generated by ghost fields $c[m,n]$ with anti-commutativity relations, together with BRST differential $\delta$. As a graded associative algebra, $A_{0,\E_2}$ is isomorphic to $\wedge^* \left(\mathbb C[z_1,z_2]_{\E_2}\right)\cong\wedge^* \left(\mathrm{Diff}_{\E_2}\mathbb C\right)$, note that here we identify $z_1$ as $\partial_{z_2}$ using the Moyal product. The BRST differential $\delta$ is the dual of the Lie bracket, thus $A_{\E_1=0,\E_2}$ is the Chevalley-Eilenberg algebra of cochains on the Lie algebra $\mathfrak{g}=\mathrm{Diff}_{\E_2}\mathbb C\otimes \mathfrak{gl}_1$, denote by $C^*(\mathfrak{g})$. Note that here we treat the algebra $A_{\E_1=0,\E_2}$ as an algebra over the base ring $\mathbb{C}[\E_1,\E_2]$, so $\E_1,\E_2$ are algebraic parameters. At the quantum level, the operator algebra $A_{\E_1=0,\E_2}$ receives deformations, we will denoted it by $A_{\E_1,\E_2}$.}

On the other hand, the elements of the algebra of operators in 1d TQM {in the large $N$ limit} consist of $t[m,n]$. The defining commutation relations come from the quantization of the Poisson brackets deformed by $\Omega_{\E_1}${-background}:
\ie\label{commadhm}
\left[I_a,J^b\right]=\E_1\D^b_a,\quad\left[X^a_b,Y^c_d\right]=\E_1\D^a_d\D^c_b
\fe
We will write the F-term relation with {explicit gauge indices} as follows.
\ie\label{Ftermadhm}
X^i_kY^k_j-X^k_jY^i_k+I_jJ^i=\E_2\D^i_j
\fe
We will call the algebra {generated} by $t[m,n]$ {with relations} \eqref{commadhm}, \eqref{Ftermadhm} as ADHM algebra or $\cA_{\E_1,\E_2}$. {Note that here we treat the algebra $\cA_{\E_1,\E_2}$ as an algebra over the base ring $\mathbb{C}[\E_1,\E_2]$. This may seems to be strange at the first glance since $t[m,n]$ is \textit{a priori} defined over $\mathbb{C}[\E_1^{\pm},\E_2]$, nevertheless the commutators between those $t[m,n]$'s only involve polynomials in $\E_1$, so the algebra $\cA_{\E_1,\E_2}$ is well-defined over $\mathbb{C}[\E_1,\E_2]$.}

There is a one-to-one correspondence between $c[m,n]$ and $t[m,n]$, and \cite{Costello:2017fbo} proved an isomorphism between ${}^!A_{\E_1,\E_2}=U_{\E_1}(g)$ and $\cA_{\E_1,\E_2}$ {as $\mathbb{C}[\E_1,\E_2]$-algebras }for 5d $U(K)$ Chern-Simons theory coupled with 1d TQM with {$K>1$}, where ${}^!A_{\E_1,\E_2}$ is a Koszul dual of an algebra $A_{\E_1,\E_2}$\footnote{It is known that for $A_{\E_1{=0},\E_2}=C^*(g)$, the Koszul dual ${}^!A_{\E_1{=0},\E_2}$ is $U(g)$ {\cite{Costello:2013zra}}.}. The proof consists of two parts. First, one checks two algebras' commutation relations match in the $\cO(\E_1)$ order. Next, one proves the uniqueness of the deformation of the universal enveloping algebra $U(g)$ by $\E_1$ that ensures all order matching. {It worth mentioning that in the classical limit $\E_1\to0$ the algebra $\cA_{\E_1,\E_2}$ does not agree with the classical operator algebra of the ADHM mechanics, since the definition of the algebra $\cA_{\E_1,\E_2}$ involves $1/\E_1$, in other word, the isomorphism holds only at the quantum level.}

One of our goals is to extend the $\cO(\E_1)$ order matching to $K=1$. It may seem trivial compared to higher K, but it turns out that it is more complicated. We will give the proof in \secref{sec:pert}, \secref{sec:5d2dpert}. The uniqueness of the deformation applies for all $K$ including $K=1$, so we will not try to spell out the details in this work.
\subsection{Koszul duality}\label{subsec:koszul}
Let us explain why in the first place we can expect the Koszul duality between 5d and {the large $N$ limit} of 1d operator algebra. Further details on Koszul duality can be found in \cite{Gaiotto:2015zna,Gaiotto:2015aoa,Gaiotto:2019wcc,Costello:2020jbh}

The 5d theory is defined on $\bR_t\times\bC_{NC}^2$, where $\bR_t$ is topological and $\bC_{NC}^2$, and 1d TQM couples to the 5d theory along $\bR_t$. As explained in \eqref{trivialtr}, there is a scalar supercharge $Q$ and 1-form supercharge $\D$ that anti-commute to give a translation operator $P_t$. We can build a topological line defect action using topological descent.
\ie\label{lined}
Pexp\int^\infty_{-\infty}\left[\D,x(t)\right]
\fe
where 
\ie
x(t)=\sum_{m,n}c[m,n]t[m,n]
\fe
The BRST variation of \eqref{lined} vanishes if $x(t)$ {satisfies} a Maurer-Cartan equation:
\ie
\left[Q,x\right]+x^2=0
\fe
and if ${x\in A\otimes {}^!A}$ for some $A$, the Maurer-Cartan equation is always satisfied.  Hence, it is natural to expect the Koszul duality between $A_{\E_1,\E_2}$ and $\cA_{\E_1,\E_2}$. So, the coupling between the 5d ghosts and gauge invariant polynomials of 1d TQM is given by 
\ie\label{1dtqmline}
S_{int}=\int_{\bR_t}t[m,n]c[m,n]dt.
\fe

Now that we have three types of Lagrangians:
\ie
S_{5d~{CS}}+S_{1d~TQM}+S_{int}
\fe
We need to make sure if the quantum gauge invariance of 5d Chern-Simons theory remains to be true in the presence of the interaction with 1d TQM. Namely, we need to investigate if there is a non-vanishing gauge anomaly in Feynman diagrams. Along the way, we will derive the isomorphism between the operator algebras, as a consistency condition for the gauge anomaly cancellation.
\subsection{Anomaly cancellation}\label{subsec:anomaly}
To give an idea that the cancellation of the gauge anomaly of the 5d CS Feynman diagrams fixes the algebra of operators in 1d TQM that couples to the 5d CS, let us review 5d $U(K)$ Chern-Simons example shown in \cite{Costello:2017fbo}. Consider following Feynman diagram.
\begin{figure}[H]
\centering
\includegraphics[width=11cm]{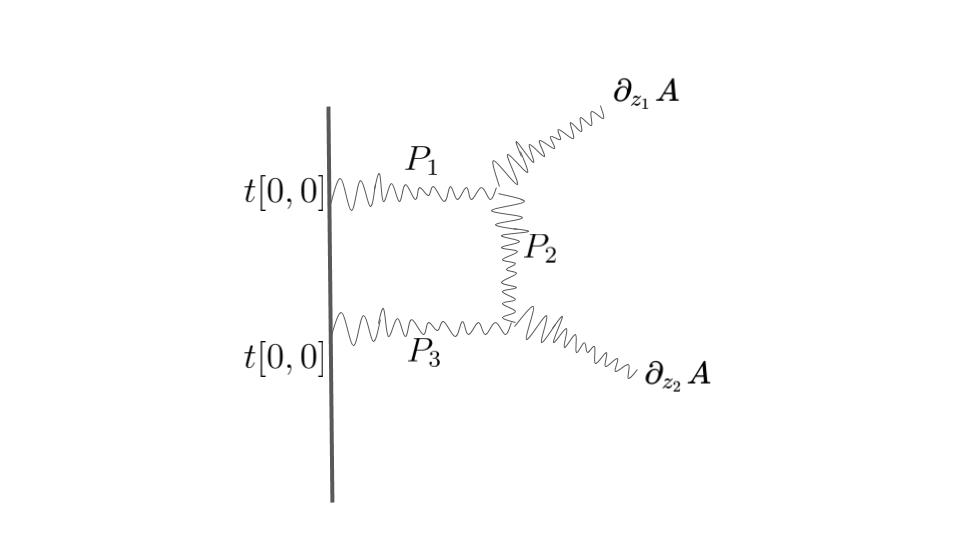}
\caption{The vertical solid line represents the time axis. Internal wiggly lines stand for 5d gauge field propagators $P_i$, and the external wiggly lines stand for Fourier components 5d gauge field.}
\label{fig1}
 \centering
\end{figure}
\noindent
The BRST variation$(\D A=\pa c)$ of the amplitude of the above Feynman diagram is non-zero.
\ie\label{sampevari1}
\E_1 \E_{ij}(\pa_{z_i}A^a)(\pa_{z_j}c^b)K^{fe}f^c_{ae}f^d_{bf}t[0,0]t[0,0]
\fe
where $K^{ab},f^a_{bc}$ are a Killing form and a structure constant of $u(K)$, and $t[m,n]$ is an element of {ADHM algebra with gauge group $G=U(N)$, and flavor group $\hat{G}=U(K)$}.

To have a gauge invariance, we need to cancel the anomaly, and the gauge variation of the following diagram has exactly factors like $\E_{ij}(\pa_{z_i}A^a)(\pa_{z_j}c^b)$:
\begin{figure}[H]
\centering
\includegraphics[width=10cm]{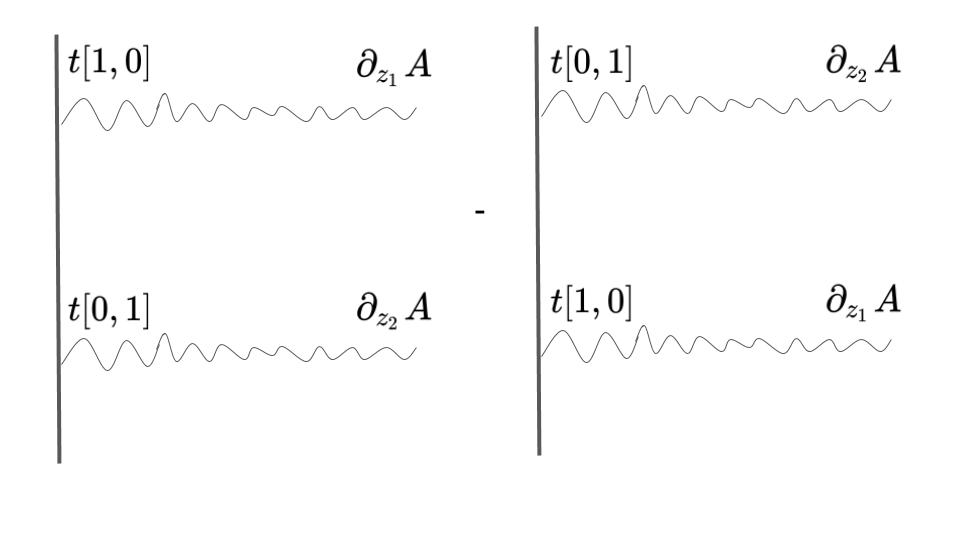}
  \vspace{-30pt}
\caption{}
\label{fig3}
 \centering
\end{figure}
\noindent
The BRST variation of the amplitude is
\ie\label{sampevari2}
\E_1 \E_{ij}(\pa_{z_i}A^a)(\pa_{z_j}c^b)K^{fe}f^c_{ae}f^d_{bf}[t[1,0],t[0,1]]
\fe
Imposing the cancellation of the BRST variation between \eqref{sampevari1} and \eqref{sampevari2}, we obtain
\ie
\left[t[1,0],t[0,1]\right]=\E_1t[0,0]t[0,0]
\fe
This is very impressive, since we obtain the ADHM algebra from 5d Chern-Simons theory Feynman diagrams!

We will see that if $K=1$, some ingredients of Feynman diagram change, but we can still reproduce ADHM algebra with $G=U(N)$, $\hat{G}=U(1)$.
\subsection{Large $N$ limit and a back-reaction of $N$ M2-branes}\label{subsec:backreac}
Although we have not discussed explicitly about taking large $N$ limit, but it was being used implicitly in the construction of the algebra $\cA_{\E_1,\E_2}$ which makes it a crucial step towards the holography. However, it is important to notice that large N is not necessary for Koszul duality, but it is important for holography. \footnote{We thank an anonymous referee who made this point.}

Here we explain some detail of taking large $N$ limit. First notice that there are homomorphisms $\iota_{N}^{N'}:\mathcal O(T^*V_{K,N'})\to \mathcal O(T^*V_{K,N})$ for all $N'>N$ induced by natural embedding $\mathbb C^N\hookrightarrow \mathbb C^{N'}$, where 
\ie
V_{K,N}=\mathfrak{gl}_N\oplus \mathrm{Hom}(\mathbb C^K,\mathbb C^N),\fe
so that $T^*V_{K,N}$ is the linear span of single operators $I,J,X,Y$, and the algebra $ \mathcal O(T^*V_{K,N})$ is the commutative (classical) algebra generated by these operators (with no relations imposed). Then we define the \textit{admissible} sequence of weight $0$ as 
\ie
\{f_N\in \mathcal O(T^*V_{K,N})^{\mathrm{GL}_N}|\iota_{N}^{N'}(f_{N'})=f_N\},
\fe
and for integer $r\ge 0$, a sequence $\{f_N\}$ is called admissible of weight $r$ if $\{N^{-r}f_N\}$ is admissible sequence of weight $0$ (e.g. the sequence $\{N\}$ is admissible of weight 1), and define $\mathcal O(T^*V_{K,\bullet})^{\mathrm{GL}_{\bullet}}$ to be the linear span of admissible sequences of all possible weights. It's easy to see that $\mathcal O(T^*V_{K,\bullet})^{\mathrm{GL}_{\bullet}}$ is an algebra. Next we turn on the quantum deformation which turn the ordinary commutative product to the Moyal product $\star _{\E_1}$, and it's easy to see that for admissible sequences $\{f_N\}$ and $\{g_N\}$, $\{f_N\star _{\E_1} g_N\}$ is also admissible. In this way we obtained the quantized algebra $\mathcal O_{\E_1}(T^*V_{K,\bullet})^{\mathrm{GL}_{\bullet}}$. 

{The action of $\mathfrak{gl}_N$ on $V_{K,N}$ induces a moment map \ie
\mu:\mathfrak{gl}_N\to \mathcal O_{\E_1}(T^*V_{K,N}),\quad E_i^j\mapsto X_i^kY_k^j-X_k^jY_i^k+I_iJ^j,
\fe
We want to set the moment map to $\E_2$ times the identity, so we consider the shifted moment map:}
\ie
\mu_{\E_2}:\mathfrak{gl}_N\to \mathcal O_{\E_1}(T^*V_{K,N}),\quad E_i^j\mapsto X_i^kY_k^j-X_k^jY_i^k+I_iJ^j-\E_2\delta _i^j,
\fe
which is $\mathrm{GL}_N$-equivaraint. Together with the Moyal product on $\mathcal O_{\E_1}(T^*V_{K,N})$, $\mu_{\E_2}$ gives rise to a $\mathrm{GL}_N$-equivaraint map of left $\mathcal O_{\E_1}(T^*V_{K,N})$-modules
\ie
\mu_{\E_2}:\mathcal O_{\E_1}(T^*V_{K,N})\otimes\mathfrak{gl}_N\to \mathcal O_{\E_1}(T^*V_{K,N}).
\fe
Taking $\mathrm{GL}_N$-invariance, we obtain the quantum moment map
\ie
\mu_{\E_2}:(\mathcal O_{\E_1}(T^*V_{K,N})\otimes\mathfrak{gl}_N)^{\mathrm{GL}_N}\to \mathcal O_{\E_1}(T^*V_{K,N})^{\mathrm{GL}_N}.
\fe
It is easy to verify that the image of $\mu_{\E_2}$ is a two-sided ideal. Similar to $\mathcal O_{\E_1}(T^*V_{K,\bullet})^{\mathrm{GL}_{\bullet}}$, we can define admissible sequences in $(\mathcal O_{\E_1}(T^*V_{K,N})\otimes\mathfrak{gl}_N)^{\mathrm{GL}_N}$ and call this space $(\mathcal O_{\E_1}(T^*V_{K,\bullet})\otimes\mathfrak{gl}_{\bullet})^{\mathrm{GL}_{\bullet}}$. Quantum moment maps for all $N$ give rise to 
\ie
\mu_{\E_2}:(\mathcal O_{\E_1}(T^*V_{K,\bullet})\otimes\mathfrak{gl}_{\bullet})^{\mathrm{GL}_{\bullet}}\to \mathcal O_{\E_1}(T^*V_{K,\bullet})^{\mathrm{GL}_{\bullet}},
\fe
and the image is a two-sided ideal, so we can take the quotient of $\mathcal O_{\E_1}(T^*V_{K,\bullet})^{\mathrm{GL}_{\bullet}}$ by this ideal, this is by definition the large-$N$ limit denoted by $\mathcal O_{\E_1}(\mathcal M^{\E_2}_{K,\bullet})$.

From the definition above, we can write down a set of generators of $\mathcal O_{\E_1}(\mathcal M^{\E_2}_{K,\bullet})$:
\ie
\{N\}\text{ and }\{I_{\alpha}S(X^nY^m)J^{\beta}\}\text{ for all integers }n,m\ge 0.
\fe
Note that Costello also defined a combinatorial algebra $\mathcal A^{\text{comb}}_{\E_1,\E_2}$ in section 10 of \cite{Costello:2017fbo}, which depends on $K$ but not on $N$. This is related to $\mathcal O_{\E_1}(\mathcal M^{\E_2}_{K,\bullet})$ in the sense that generators of $\mathcal A^{\text{comb}}_{\E_1,\E_2}$ are 
\ie
\{N\}\text{ and }\{\frac{1}{\E _1}I_{\alpha}S(X^nY^m)J^{\beta}\}\text{ for all integers }n,m\ge 0,
\fe
when $\E _1\neq 0$ \footnote{In the notation of \cite{Costello:2017fbo} they correspond to
$
D(\emptyset)\text{ and }\mathrm{Sym}(D(\alpha\Downarrow,\uparrow ^n,\downarrow _m,\beta \Uparrow))\text{ for all integers }n,m\ge 0,
$
respectively.}. {Under the aforementioned correspondence between generators, $\mathcal A^{\text{comb}}_{\E_1,\E_2}$ is isomorphic to $\mathcal O_{\E_1}(\mathcal M^{\E_2}_{K,\bullet})$ (Proposition 13.4.3 of \cite{Costello:2017fbo}) when $\E_1$ is invertible.}

The general philosophy of AdS/CFT \cite{Maldacena:1997re} teaches us that the back-reaction of $N$ M2-branes will deform the spacetime geometry. In our case, since the closed strings completely decouple from the analysis, the back-reaction is only encoded in the interaction related to the open strings. More precisely, the back-reaction is already encoded in the 5d-1d interaction Lagrangian \eqref{1dtqmline}, a part of which we reproduce below.
\ie
S_{back}=\int_{\bR_t}t[0,0]c[0,0]dt.
\fe
Here, we can explicitly see $N$ in $t[0,0]$, as 
\ie
t[0,0]=IJ/\E_1=\E_2Tr\D^i_j/\E_1=N\frac{\E_2}{\E_1}
\fe
where in the second equality, we used the F-term relation. 

After taking large $N$ limit, $N$ becomes an element of the algebra $\cA_{\E_1,\E_2}$, which is coupled to the zeroth Fourier mode of the 5d ghost, $c[0,0]$.
\subsection{M5-brane in $\Omega-$deformed M-theory}\label{subsec:M5O}
We want to include one M5(D4)-brane in the story, and review the role played by the new element {coming from the bi-module on M5(D4)-brane} in the boundary and the bulk. 
\begin{table}[H]\centering
\begin{tabular}{|c|c|cc|cccc|cccc|} 
\hline
   &  0 &  1 &  2 &  3 &  4 &  5 &  6 &  7 &  8 &  9 & 10 \\
\hline\hline
 Geometry& $\bR_t$& $\bC_{\E_1}$ &   &  $\bC^2_{NC}$ &  &  &   & $\bC_{\E_3}$ & &$\bR$ & $S^1_{\E_2}$ \\
\hline
$M2(D2)$ &$\times$ &$\times$&$\times$&&&&&& &&\\
\hline
$M5$ & &$\times$&$\times$&$\times$&$\times$&&&& &$\times$&$\times$\\
\hline
$D4$ & &$\times$&$\times$&$\times$&$\times$&&&& &$\times$&\\
\hline
\end{tabular}
\caption{M2, M5-brane}
\label{table:M2, M5-brane}
\end{table}

In the boundary perspective, it intersects with the M2(D2)-brane with two directions and supports 2d $\cN=(2,2)$ supersymmetric field theory with two chiral superfields, whose bottom components are $\varphi$, $\tilde\varphi$, arising from $D2-D4$ strings. This 2d theory interacts with the 3d $\cN=4$ ADHM theory with a superpotential 
\ie
\cW=\tilde\varphi X\varphi
\fe
where $X$ is a scalar component of the adjoint hypermultiplet of the 3d theory. 
\begin{figure}[H]
\centering
\includegraphics[width=10cm]{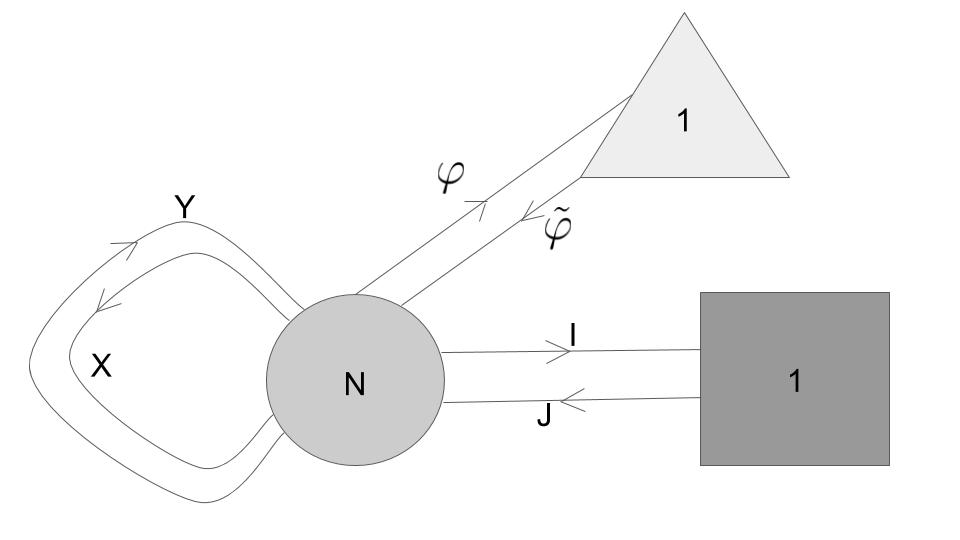}
  \vspace{-20pt}
\caption{3d $\cN=4$ ADHM quiver gauge theory with $G=U(N)$, $F=U(1)$, decorated with $2d$ $\cN=(2,2)$ field theory. $X$, $Y$ are scalars of adjoint hypermultipet, and $I$, $J$ are scalars of (anti)fundamental hypermultiplet. The triangle node encodes the 2d theory. $\varphi$ and $\tilde{\varphi}$ are 2d scalars. In type IIA language, the circle, square, and triangle node correspond to D2, D6, D4 branes, respectively.}
\label{fig5}
 \centering
\end{figure}
A naive set of gauge invariant operators living on the 2d intersection are
\ie
IX^mY^n\tilde\varphi,\quad \varphi X^mY^nJ,\quad  \varphi X^mY^n\tilde\varphi
\fe
The superpotential reduces \cite{Bullimore:2016nji,Gaiotto:2019wcc} the above set into
\ie\label{moduleelem}
\cM_{\E_1,\E_2}=\{b[z^n]=IY^n\tilde\varphi,\quad c[z^n]=\varphi Y^nJ\}
\fe
The set of 2d observables $\cM_{\E_1,\E_2}$ forms a bi-module of the ADHM algebra $\cA_{\E_1,\E_2}$. 

The difference between left and right actions of the algebra $\cA$ on $\cM_{\E_1,\E_2}$ is encoded in the form of a commutator:
\ie\label{modulerels}
\left[a,m\right]=m',\quad\text{where }a\in\cA,\quad m,m'\in\cM_{\E_1,\E_2}
\fe
To verify \eqref{modulerels}, we need to establish the commutation relations between the set of letters $\{\varphi,\tilde\varphi\}$ and $\{X,Y,I,J\}$. Those are given by
\ie\label{superreps}
IP(\varphi,\tilde\varphi)&=P(\varphi,\tilde\varphi)I\\
JP(\varphi,\tilde\varphi)&=P(\varphi,\tilde\varphi)J\\
X^i_jP(\varphi,\tilde\varphi)&=P(\varphi,\tilde\varphi)X^i_j\\
Y^i_jP(\varphi,\tilde\varphi)&=P(\varphi,\tilde\varphi)(Y^i_j+\tilde\varphi^i\varphi_j)\\
X^i_j\varphi_iP(\varphi,\tilde\varphi)&=-\E_1\pa_{\tilde\varphi_j}P(\varphi,\tilde\varphi)\\
X^i_j\tilde\varphi^jP(\varphi,\tilde\varphi)&=-\E_1\pa_{\varphi_i}P(\varphi,\tilde\varphi)
\fe
Again, the non-trivial commutation relations in the last three lines originates from the effect of the particular superpotential $\cW$. {For the derivation, we refer the reader to \cite{Bullimore:2016nji,Gaiotto:2019wcc}.}

{In the $\Omega_{\E_1}$-background, 2d $\cN=(2,2)$ theory localizes to a point,} which is the origin of $\bR_t$. 
\begin{figure}[H]
\centering
\includegraphics[width=10cm]{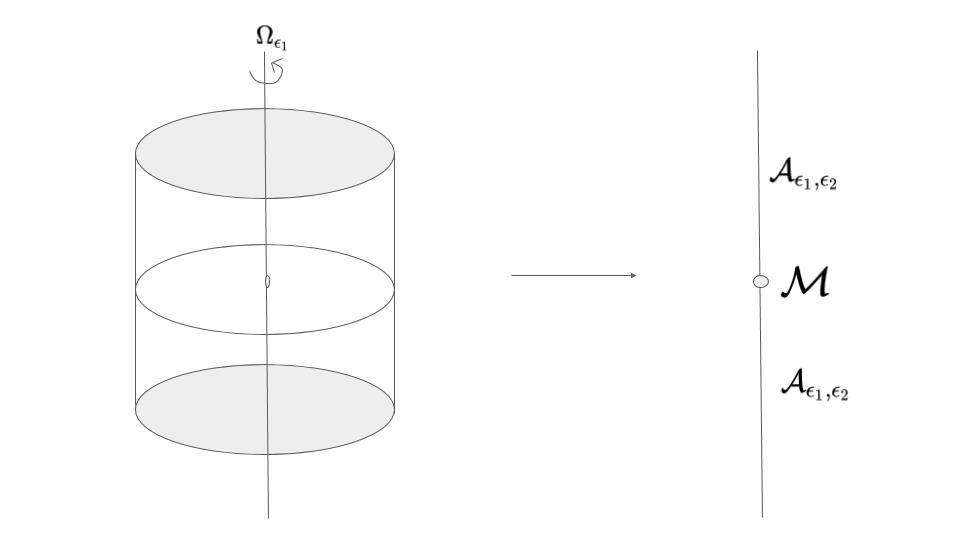}
  \vspace{-20pt}
\caption{Left figure represents a coupled system of 3d $\cN=4$ ADHM theory(the cylinder) and 2d $\cN=(2,2)$ theory(the middle disk in the cylinder) from D2 branes and a D4 brane. {In the $\Omega_{\E_1}$-background, the system localizes to $1d+0d$ system.}}
\label{Omega.jpg}
 \centering
\end{figure}
\noindent 
Hence, the resulting system is {ADHM algebra $\cA_{\E_1,\E_2}$ and bi-module $\cM_{\E_1,\E_2}$ of the algebra.}

To study the bulk perspective, we need to study what degree of freedoms that M5-brane support in the 5d spacetime $\bR_t\times\bC_{NC}^2$ and how the M5-brane interacts with 5d Chern-Simons theory. 5d CS theory is defined in the context of type IIa, and M5-brane is mapped to a D4-brane. The local degree of freedom comes from D4-D6 strings, which are placed on $\{\cdot\}\times\bC\in\bR_t\times\bC^2_{NC}$. These 2d degrees of freedom are actually coming from 4d $\cN=2$ hypermultiplet, as the true intersection between D4 and D6 is $\bC\times\bC_{\E_1}$. {In the $\Omega_{\E_1}$-background, the 4d $\cN=2$ hypermultiplet localizes to a $\beta-\gamma$ system \cite{Oh:2019bgz}.} Hence, we arrive at $\beta-\gamma$ Vertex Algebra on $\bC\subset\bC_{NC}^2$.
\begin{table}[H]\centering
\begin{tabular}{|c|c|cc|cccc|ccc|} 
\hline
   &  0 &  1 &  2 &  3 &  4 &  5 &  6 &  7 &  8 &  9 \\
\hline\hline
 Geometry& $\bR_t$& $\bC_{\E_1}$ &   &  $\bC^2_{NC}$ &  &  &   & $\bC_{\E_3}$ & &$\bR_{\E_2}$ \\
\hline
1d TQM &$\times$ &&&&&&&& &\\
\hline
2d $\beta\gamma$ & &&&$\times$&$\times$&&&& &\\
\hline
5d~CS&$\times$&   &   & $\times$ & $\times$ & $\times$&  $\times$&  &  &\\
\hline
\end{tabular}
\caption{Bulk perspective}
\label{table:M2Directions}
\end{table}
The $\beta-\gamma$ system minimally couples to 5d Chern-Simons theory via
\ie
\int_{\bC}\beta(\bar\pa+A\star)\gamma
\fe
The observables to be compared with those of field theory side: $b[z^n]$ and $c[z^n]$ can be naturally compared with the modes of $\beta$ and $\gamma$: $\pa_z^n\beta$, $\pa_z^n\gamma$, and the Koszul duality manifests itself by the coupling between two types of observables:
\ie\label{couplingGM}
\int_{\{0\}}\pa^{k_1}_{z_2}\B\cdot b[z^{k_1}]+\int_{\{0\}}\pa^{k_2}_{z_2}\gamma\cdot c[z^{k_2}]
\fe
where $z=z_2$, and the integral on a point is merely for a formal presentation.  

The following figure depicts the entire bulk and boundary system including the line and the surface defect, and describes how all the ingredients are coupled.
\begin{figure}[H]
\centering
\includegraphics[width=10cm]{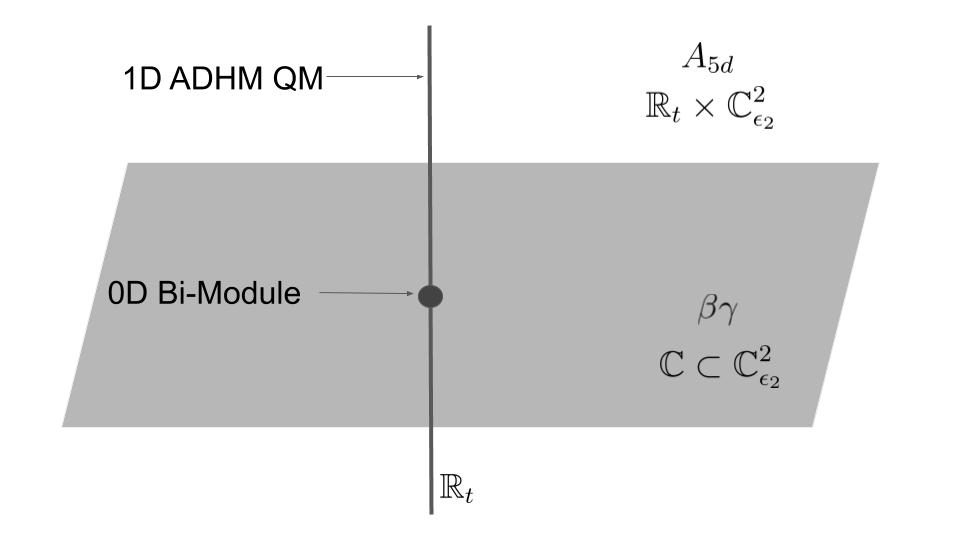}
  \vspace{-20pt}
\caption{5d Chern-Simons$(\bR_t\times\bC^2_{NC})$, 1d generalized Wilson line defect$(\bR_t)$, and 2d surface defect$(\bC\subset\bC^2_{NC})$.}
\label{fig6}
 \centering
\end{figure}

As explained in section \secref{subsec:anomaly}, we need to make sure if the introduction of the 2d system is quantum mechanically consistent, or anomaly free. Imposing the anomaly cancellation condition of the 5d, 2d, 1d coupled system, we should be able to derive the bi-module commutation relations defined in the field theory side. This is the content of \secref{sec:5d2dpert}.
\section{M2-brane algebra and M5-brane module}\label{sec:M2M5}
In this section, we will provide a representative commutation relation for the algebra $\cA_{\E_1,\E_2}$
\ie
\left[a,a'\right]=a_0+\E_1a_1+\E_1^2a_2+\ldots,\quad\text{where }a,a',a_i\in\cA_{\E_1,\E_2}
\fe
and a representative commutation relation for the algebra $\cA_{\E_1,\E_2}$ and the bi-module $\cM_{\E_1,\E_2}$ for $\cA_{\E_1,\E_2}$.
\ie
\left[a,m\right]=m_0+\E_1m_1+\E_1^2m_2+\ldots,\quad\text{where }a\in\cA_{\E_1,\E_2},~m,m_i\in\cM_{\E_1,\E_2}
\fe
We first recall the notation for a typical element of $\cA_{\E_1,\E_2}$ and $\cM_{\E_1,\E_2}$:
\ie
t[m,n]&=\frac{1}{\E_1}TrS(X^mY^n)=\frac{1}{\E_1\E_2}IS(X^mY^n)J\in\cA_{\E_1,\E_2}\\
b[z^m]&=\frac{1}{\E_1}IY^m\tilde\varphi\in\cM_{\E_1,\E_2}\\
c[z^n]&=\frac{1}{\E_1}\varphi Y^nJ\in\cM_{\E_1,\E_2}
\fe
For the convenience of later discussions, we also introduce the notation:
\ie
T[m,n]=\frac{\E_2}{\E_1}TrS(X^mY^n)=\frac{1}{\E_1}IS(X^mY^n)J\in\cA_{\E_1,\E_2}
\fe
Our final goal is to reproduce the $\cA_{\E_1,\E_2}$ algebra from the anomaly cancellation of 1-loop Feynman diagrams in 5d Chern-Simons theory. So, it is important to have commutation relations that yield $\cO(\E_1)$ term in the right hand side, where $\E_1$ is a loop counting parameter in 5d CS theory.
\subsection{M2-brane algebra}\label{subsec:M2}
Since we have not provided a concrete calculation until now, let us give a simple computation to give an idea of ADHM algebra and its bi-module. It is useful to recall $G=U(N)$, $\hat{G}=U(K)$ ADHM algebra, which serves as a practice example, and at the same time as an example that explains the non-triviality of $G=U(N)$, $\hat{G}=U(1)$ ADHM algebra, compared to $K>1$ cases.

It was shown in \cite{Costello:2017fbo} that following commutation holds for $G=U(N)$, $\hat{G}=U(K)$ ADHM algebra. 
\ie\label{unex}
\left[t[1,0],t[0,1]\right]=\E_1t[0,0]t[0,0]\quad\text{or}\quad\left[IXJ,IYJ\right]=\E_1(IJ)(IJ)
\fe
This does not work for $\hat{G}=U(1)$. It is instructive to see why. 
\ie
\left[TrX,TrY\right]&=[X^i_i,Y^j_j]=\D^i_j\D^j_i\E_1=\D^i_j\E_1\\
&=N\E_1
\fe
Multiplying both sides by $\E_2^2/\E_1^2$, we can convert it into $T[m,n]$ basis:
\ie\label{u1ex}
[T[1,0],T[0,1]]=\E_2 T[0,0]
\fe
The RHS of \eqref{u1ex} is different from \eqref{unex} crucially in its dependence on $\E_1$. The RHS of \eqref{u1ex} is $\cO(\E_1^0)$, but that of \eqref{unex} is $\cO(\E_1)$. While it was sufficient to consider this simple commutator to see the $\E_1$ deformation of the algebra for $\hat{G}=U(K)$ with $K>1$, we need to consider a more complicated commutator to see $\cO(\E_1)$ correction in the RHS.

In Appendix \secref{app:alg}, we will derive a set of relations that will determine all other relations, of which the simplest ones are:
\ie\label{ADHMgeneral}
\left[t[3,0],t[0,3]\right]&=9t[2,2]+\frac{3}{2}\big(\sigma_2t[0,0]-\sigma_3t[0,0]t[0,0]\big)\\
\left[t[2,1],t[1,2]\right]&=3t[2,2]-\frac{1}{2}\big(\sigma_2t[0,0]-\sigma_3t[0,0]t[0,0]\big)
\fe
where
\ie
\sigma_2=\E_1^2+\E_2^2+\E_1\E_2,\quad\sigma_3=-\E_1\E_2(\E_1+\E_2)
\fe

To compare the commutation relation to that from 5d Chern-Simons calculation, we need to make sure if the parameters of ADHM algebra $\cA_{\E_1,\E_2}$ are the same as those in 5d CS theory. From \cite{Costello:2017fbo}, the correct parameter dictionary\footnote{We thank Davide Gaiotto, who pointed out this subtlety.} is
\ie\label{shiftingscheme}
(\E_1)_{ADHM}=(\E_1)_{CS},\quad\bigg(\E_2+\frac{1}{2}\E_1\bigg)_{ADHM}=(\E_2)_{CS}
\fe 
Hence, the commutation relation that we are supposed to match from the 5d computation is
\ie\label{ADHMgeneral1}
\left[t[2,1],t[1,2]\right]&=3t[2,2]-\frac{1}{2}\bigg(\big(\E_2^2+\frac{3}{4}\E_1^2\big)t[0,0]+\big(\E_1\E_2^2-\frac{\E_1^3}{4}\big)t[0,0]t[0,0]\bigg)
\fe
There is one term in the RHS of \eqref{ADHMgeneral1} that is in $\cO(\E_1)$ order:
\ie\label{targetalge}
\left[t[2,1],t[1,2]\right]&=\cO(\E_1^0)-\frac{1}{2}\E_1\E_2^2t[0,0]t[0,0]+\cO(\E_1^2)
\fe
We will try to recover the $\cO(\E_1)$ term from 5d Feynman diagram calculation\footnote{The basis used in the Feynman diagram computation is $T[m,n]$, not $t[m,n]$. However, the change of basis does not affect any computation because the $\mathcal O(\E_1)$ term in \eqref{targetalge} is quadratic in $t$. } in section \secref{sec:pert}; the general argument that gauge anomaly cancellation leads to the Koszul dual algebra commutation relation is given in \secref{subsec:anomaly}.
\subsection{M5-brane module}\label{subsec:M5}
We will use the commutation relations \eqref{commadhm}, \eqref{Ftermadhm}, \eqref{superreps} to compute the commutators between $a\in\cA_{\E_1,\E_2}$ and $m\in\cM_{\E_1,\E_2}$, which are defined in \eqref{algebraelem}, \eqref{moduleelem}. When one tries to compute some commutators, one immediately notices some normal ordering ambiguity in a general module element $m$, which can be seen in following example.
\ie
\left[IXJ,(I\tilde\varphi)(\varphi J)\right]=\left[I_iX^i_jJ^j,I_a\tilde\varphi^a\varphi_b J^b\right]
\fe
Assuming that the order of letters is consistent with the order of fields in the real line $\bR_t$, it is obvious that we need to place $\tilde\varphi^a\varphi_b$ together, as they are defined at a point $\{0\}\in\bR_t$\footnote{Recall that $\varphi$, $\tilde\varphi$ are chiral multiplet scalars that are localized at the interface(between the line and the surface). {In the $\Omega_{\E_1}$-background}, the interface localizes to a point. Hence, $\varphi$, $\tilde\varphi$ are localized to be at a point on the line. }. However, it is ambiguous whether we put $I_a$, $J^b$ in the right or left of $\tilde\varphi^a\varphi_b$, as $I_a$, $J^b$ are living on $\bR_t$. We will try to fix this ambiguity to prepare a concrete calculation.

Considering following normal ordering when writing a module element $(IY\varphi)(\varphi J)$ will be enough to fix the ambiguity.
\ie
\rvert\tilde\varphi^j\varphi_k\rvert I_iJ^kY^i{}_j
\fe
We simply choose other letters like $X,Y,I,J$ to be placed on the right side of $\varphi$ and $\tilde\varphi$.

Still, there is an ordering ambiguity. For instance between two words:
\ie
\rvert\tilde\varphi\varphi\rvert IJY\quad vs\quad \rvert\tilde\varphi\varphi\rvert JIY
\fe
We simply choose an alphabetical order to arrange letters. In other words, we use the commutation relations until the letters in the word have an alphabetical order. When the word has an alphabetical order, we contract the gauge indices to form a single-trace word and omit the gauge indices. For instance,
\ie
(\tilde\varphi\varphi):=&\rvert\tilde\varphi^j\varphi_j\rvert \\
(IY\tilde\varphi)(\varphi J):=&\rvert\tilde\varphi^j\varphi_l\rvert I_kJ^lY^k{}_j \\
(I\tilde\varphi)(\varphi J)(IJ):=&\rvert\tilde\varphi^j\varphi_k\rvert I_jJ^kI_iJ^i
\fe
As a consequence, some more steps are needed for the following:
\ie
\rvert\tilde\varphi^j\varphi_k\rvert I_iI_jJ^kJ^i
\fe
That is, we need to commute $I_i$ through $J^k$ to contract with $J^i$. While doing this, we necessarily use $[I_i,J^k]=\E_1\D^k_i+J^kI_i$, which produces two terms.

Having fixed the ordering ambiguity, there is a few things to keep in mind additionally: 
\begin{itemize}
\item{}We use F-term relation and the basic commutation relation between $X$ and $Y$ in maximum times to get rid of X's in the word, since the module only consists of $\varphi$, $\tilde\varphi$, $I$, $J$, $Y$. \\
\item{}To use F-term relation, we first need to pull the target XY(or YX) pair to the right end, not to ruin the gauge invariance, and pull it back to the original position in the word.\\
\item{}To use the superpotential relations$(X\varphi=\E_1\pa_{\tilde\varphi}$ or $X\tilde\varphi=\E_1\pa_\varphi$), we need to bring $X$ right next to $\varphi$ or $\tilde\varphi$.\\
\end{itemize}

Given the prescription, we would like to find $a\in\cA_{\E_1,\E_2}$ and $m\in\cM_{\E_1,\E_2}$ such that the value of $[a,m]$ contains $\cO(\E_1)$ terms. To illustrate the prescription, let us consider following simple example, which will not produce $\cO(\E_1)$ term.\\
\newline
{\bf{Example: }{$\left[IXJ,(IY\tilde\varphi)(\varphi J)\right]$}}\\
\newline
It is much clear and convenient to use closed word version for the algebra element. We will recover the open word at the end by simply multiplying $\E_2$ on the closed words.
\ie
\left[TrX,(IY\tilde\varphi)(\varphi J)\right]=(X)\cdot(IY\tilde\varphi)(\varphi J)-(IY\tilde\varphi)(\varphi J)\cdot(X)
\fe
Compute the first term:
\ie
X^0_0\rvert\tilde\varphi^b\varphi_c\rvert I_aY^a_bJ^c&=\rvert\tilde\varphi^b\varphi\rvert I_a(\E_1\D^a_b+Y^a_bX^0_0)J^c\\
&=\E_1\rvert\tilde\varphi^b\varphi_c\rvert I_bJ^c+(IY\tilde\varphi)(\varphi J)\cdot(X)
\fe
So,
\ie
\left[TrX,(IY\tilde\varphi)(\varphi J)\right]&=\E_1\rvert\tilde\varphi^b\varphi_c\rvert I_bJ^c\\
&=\E_1(I\tilde\varphi)(\varphi J)
\fe
After normalization, by multiplying $\frac{\E_2}{\E_1^3}$ both sides, we get
\ie
\left[T[1,0],b[z]c[1]\right]=\E_2b[1]c[1]
\fe
There is no $\cO(\E_1)$ correction. So, we need to work harder.

The first bi-module commutator that has an $\E_1$ correction with some non-trivial dependence on $\E_2$ is $\left[TrS(X^2Y),(IY\tilde\varphi)(\varphi J)\right]$. After properly normalizing it, we have
\ie\label{targetmodule}
\left[T[2,1],b[z]c[1]\right]=&\bigg(-\frac{5}{3}\E_2T[0,1]+\E_2^2b[1]c[1]\bigg)\\
&+\boxed{\E_1\bigg({-\E_2b[1]c[1]T[0,0]}+{\frac{4}{3}\E_2b[1]c[1]}\bigg)}\\
&+\E_1^2\bigg(-\frac{4}{3}b[1]c[1]T[0,0]\bigg)\\
&+\E_1^3\bigg(-\frac{1}{3}b[1]c[1]b[1]c[1]\bigg)
\fe
Here, we used the re-scaled basis $T[m,n]$ for $\cA_{\E_1,\E_2}$.  This is a better choice to be coherent with the form of the bi-module elements, since $b[z^n]=IY^n\tilde\varphi$ and $c[z^n]=\varphi Y^nJ$ explicitly depend on $I$ and $J$. \footnote{Similar to the algebra case, there might be a shift in parameters $\E_1$ and $\E_2$ in 5d CS side; here, we simply assumed that there is no shift: $(\E_1)_{5d}=(\E_1)_{1d-2d}$, $(\E_2)_{5d}=(\E_2)_{1d-2d}$. If there were a shift in the $\E_2$ dictionary, the tree level term may be a potential problem.}We have shown the proof in Appendix \secref{app:bimod}.

\section{Perturbative calculations in 5d $U(1)$ CS theory coupled to 1d QM}\label{sec:pert}
In this section, we will provide a derivation of the $G=U(N)$, $\hat{G}=U(1)$ ADHM algebra $\cA_{\E_1,\E_2}$ using the perturbative calculation in 5d $U(1)$ CS. We will see the result from the perturbative calculation matches with the expectation \eqref{targetalge}. The strategy, which we will spell out in this section, is to compute the $\cO({\E_1}^1)$ order gauge anomaly of various Feynman diagrams in the presence of the line defect from $M2$ brane$(\bR^1\times\{0\}\subset\bR^1\times\bC^2_{NC})$. Imposing a cancellation of the anomaly for the 5d CS theory uniquely fixes the algebra commutation relations. 

{Purely working in the weakly coupled 5d CS theory}, we will derive the representative commutation relations of the ADHM algebra \eqref{targetalge}:
\begin{itemize}
\item{} {Algebra commutation relation}
\ie\label{sec7relsI}
&\left[t[2,1],t[1,2]\right]=\ldots+\E_1\E_2^2t[0,0]t[0,0]+\ldots\\
\fe
where $t[n,m]$ is a basis element of $\cA_{\E_1,\E_2}$. 
\end{itemize}
As we commented in \secref{subsec:M2}, the algebra basis used in the Feynman diagram computation is $T[m,n]$, which is related to $t[m,n]$ by rescaling with $\E_2$. The effect of the change of basis is trivial in \eqref{sec7relsI}, so we will interchangeably use $t[m,n]$ and $T[m,n]$ without loss of generality.
\subsection{Ingredients of Feynman diagrams}\label{sec:5dcs}
To set-up the Feynman diagram computations, we recall the $5d$ $U(1)$ Chern-Simons theory action on $\bR_t\times\bC^2_{NC}$.
\ie\label{sec7csaction}
S=\frac{{1}}{\E_1}\int_{\bR_t\times\bC^2_{NC}}dz_1dz_2\left(A\star_{\E_2} dA+\frac{2}{3}A\star_{\E_2} A\star_{\E_2} A\right)
\fe 
with $\rvert\E_1\rvert\ll\rvert\E_2\rvert\ll1$. In components, the 5d gauge field $A$ can be written as
\ie
A=A_tdt+A_{\bar{z}_1}d\bar{z}_1+A_{\bar{z}_2}d\bar{z}_2
\fe
with all the components are smooth holomorphic functions on $\bR^1\times\bC^2_{NC}$. 

Now, we want to collect all the ingredients of the Feynman diagram computation. It is convenient to rewrite \eqref{sec7csaction} as
\ie\label{sec7csaction2}
S=\frac{{1}}{\E_1}\int_{\bR^1\times\bC^2_{NC}}dz_1dz_2\left(AdA+\frac{2}{3}A(A\star_{\E_2} A)\right)
\fe
\eqref{sec7csaction2} is equivalent to \eqref{sec7csaction} up to {a} total derivative. From the kinetic term of the Lagrangian, we can read off the following information:
\begin{itemize}
\item{} 5d gauge field propagator $P$ is a solution of
\ie
dz_1\wedge dz_2\wedge dP=\D_{t=z_1=z_2=0}.
\fe
That is,
\ie\label{5dpropa}
P(v_1,v_2)=\la A(v_1)A(v_2)\ra=\frac{\bar{z}_{12}d\bar{w}_{12}dt_{12}-\bar{w}_{12}d\bar{z}_{12}dt_{12}+t_{12}d\bar{z}_{12}d\bar{w}_{12}}{d_{12}^5}
\fe
where
\ie
v_i=(t_i,z_i,w_i),\quad d_{ij}=\sqrt{t_{ij}^2+\rvert z_{ij}\rvert^2+\rvert w_{ij}\rvert^2},\quad t_{ij}=t_i-t_j
\fe
\end{itemize}

From the three-point coupling in the Lagrangian, we can extract 3-point vertex. This is not immediate, as the theory is defined on non-commutative background. Different from $U(N)$ CS, where the leading contribution of the 3-point vertex was $AAA$, the leading contribution of the 3-point coupling of the $U(1)$ gauge bosons starts from $\cO(\E_2)A\pa_{z_1}A\pa_{z_2}A$. The reason is following:
\ie
&\int  dz\wedge dw\wedge A\wedge(A\star_{\E_2}A)\\
=&\int A\wedge\left((A_tdt+A_{\bar z}d\bar z+A_{\bar w}d\bar w)\star(A_tdt+A_{\bar z}d\bar z+A_{\bar w}d\bar w)\right)\\
=&\int  dz\wedge dw\wedge A\wedge\left[dt\wedge d\bar z\left(A_t\star A_{\bar z}-A_{\bar z}\star A_t\right)+\ldots\right]\\
=&\int  dz\wedge dw\wedge A\wedge\left[dt\wedge d\bar z\left(0+2\E_2\left(\pa_{z}A_t\pa_{w}A_{\bar z}-\pa_{w}A_t\pa_{z}A_{\bar z}\right)\right)+\ldots\right]\\
=&~2\E_2\int  dz\wedge dw\wedge A\wedge\left[dt\wedge d\bar z(\pa_{z}A_t\pa_{w}A_{\bar z}-\pa_{w}A_t\pa_{z}A_{\bar z})\right]+\cO(\E_2^2)
\fe
Note that for $U(N)$ case, $SU(N)$ Lie algebra factors attached to each $A$ prevents the $\cO(\E_2^0)$ term to vanish. Still, $U(1)\subset U(N)$ part of $A$ contributes as $\cO(\E_2)$, but it can be ignored, since we take $\E_2\ll1$. 

Hence, in $U(1)$ CS, the 3-point $A\pa_z A\pa_w A$ coupling contributes as
\begin{itemize}
\item{} Three-point vertex $\cI_{3pt}$:
\ie\label{AAAint}
\cI_{3pt}=\E_2dz\wedge dw
\fe
\end{itemize}

Now, we are ready to introduce the line defect into the theory and study how it couples to 5d gauge fields. Classically, $t[n_1,n_2]$ couples to the mode of 5d gauge field by
\ie\label{sec71ptcoup}
\int_{\bR}t[n_1,n_2]\pa_{z_1}^{n_1}\pa_{z_2}^{n_2}Adt
\fe
The last ingredient of the bulk Feynman diagram computation comes from the interaction \eqref{sec71ptcoup}. 
\begin{itemize}
\item{} One-point vertex $\cI_{1pt}^A$:
\ie\label{1ptA}
\cI_{1pt}^A=\begin{cases}t[n_1,n_2]\D_{t,z_1,z_2}\quad&\text{if }\pa_{z_1}^{n_1}\pa_{z_2}^{n_2}A\text{ is a part of an internal propagator} \\ t[n_1,n_2]\pa_{z_1}^{n_1}\pa_{z_2}^{n_2}A\quad&\text{if }\pa_{z_1}^{n_1}\pa_{z_2}^{n_2}A\text{ is an external leg}\end{cases}
\fe
\end{itemize}

Lastly, the loop counting parameter is $\E_1$. Each of the propagator is proportional to ${\E_1}$ and the internal vertex is proportional to $\E_1^{-1}$. Hence, 0-loop order$(\cO({\E_1}^0))$ Feynman diagrams may contain the same number of internal propagators and internal vertices and 1-loop order$(\cO({\E_1}))$ diagrams may contain one more internal propagators than internal vertices. 

Until now, we have collected all the components of the 5d perturbative computation \eqref{5dpropa}, \eqref{AAAint}, \eqref{sec71ptcoup}, and \eqref{1ptA}. With these, let us see what Feynman diagrams have non-zero BRST variations and how the cancelation of BRST variations of different diagrams leads to the ADHM algebra $\cA_{\E_1,\E_2}$.

\subsection{Feynman diagram}\label{subsec:algefeyn1}
The goal of this section is derive the $\mathcal O(\E_1)$-term of $[t[2,1],t[1,2]]$ by Feynman diagrams. We interpret the commutator $[t[2,1],t[1,2]]$ as the following difference between two tree level diagrams
\begin{figure}[H]
\centering
\includegraphics[width=10cm]{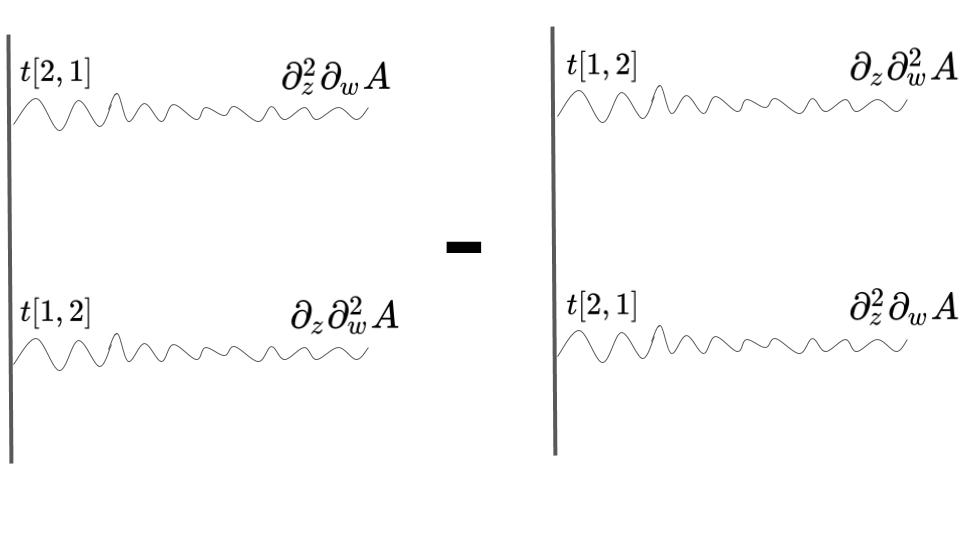}
  \vspace{-30pt}
\caption{There is no internal propagators, but just external ghosts for 5d gauge fields, which directly interact with 1d QM.  {The minus sign in the middle literally means that we take a difference between two amplitudes. In the left diagram $t[1,2]$ vertex is located at $t=0$ and $t[2,1]$ is at $t=\E$. In the right diagram, $t[1,2]$ is at $t=-\E$ and $t[2,1]$ at $t=0$.}}
\label{fig2}
 \centering
\end{figure}
\noindent
The amplitude of the diagram is
\ie
\left[t[2,1],t[1,2]\right]\pa_{z_1}^2\pa_{z_2}A_1\pa_{z_1}\pa_{z_2}^2A_2
\fe
so the BRST variation of the amplitude is proportional to
{
\ie\label{firstbrstv2}
\left[t[2,1],t[1,2]\right]\pa_{z_1}^2\pa_{z_2}A_1\pa_{z_1}\pa_{z_2}^2c_2+\left[t[2,1],t[1,2]\right]\pa_{z_1}^2\pa_{z_2}c_1\pa_{z_1}\pa_{z_2}^2A_2
\fe
Note that the BRST variation on $A$ fields is $Q_{BRST}A=\pa c$. At $\mathcal O(\E_1)$ level, this diagram will cancel all anomalies coming from one-loop diagrams with two external legs coupled to $\pa_{z_1}^2\pa_{z_2}A$ and $\pa_{z_1}\pa_{z_2}^2A$ respectively. Let's enumerate those diagrams, there are two types of diagrams:}
\begin{itemize}
    \item [(1)] {See figure \ref{extrafig}}.
    \begin{figure}[H]
\centering
\includegraphics[width=11cm]{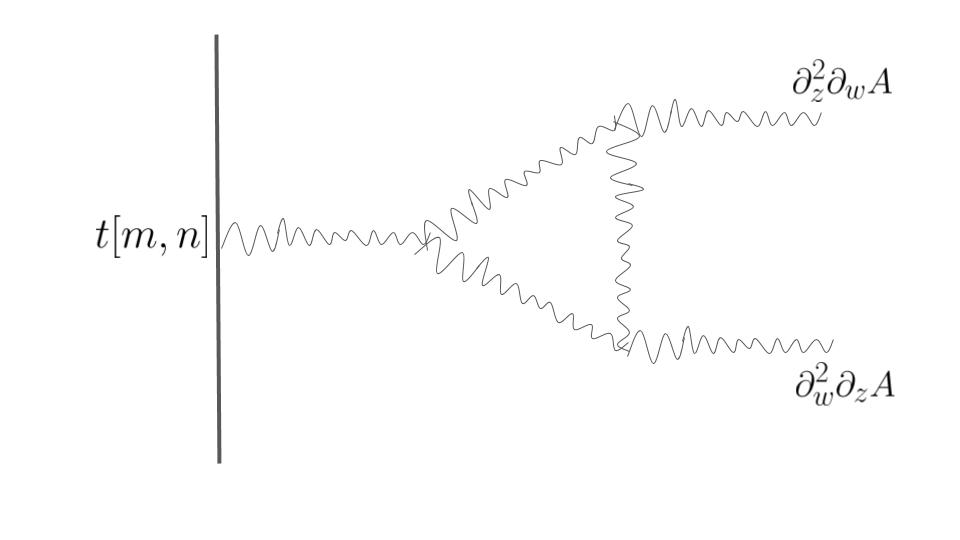}
\caption{A diagram, which has a vanishing amplitude.}
\label{extrafig}
 \centering
\end{figure}
    \item [(2)] See figure \ref{t[0,0]t[0,0]X}.
\begin{figure}[H]
\centering
\includegraphics[width=11cm]{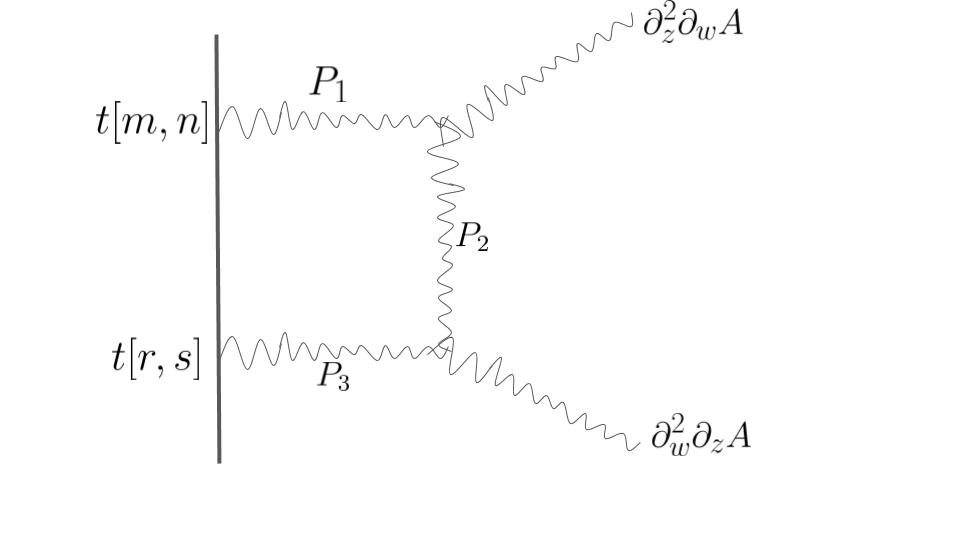}
\caption{The vertical solid line represents the time axis, where 1d topological defect is supported. Internal wiggly lines stand for 5d gauge field propagators $P_i$, and the external wiggly lines stand for 5d gauge field $A$.}
\label{t[0,0]t[0,0]X}
 \centering
\end{figure}
\end{itemize}
{
For the first diagram, we claim that the amplitude is always zero. This can be seen as follows. Let $U(1)$ act on $z$ and $w$ by rotation with weight $1$, then propagators has weight $-2$. For the interaction vertex, it contains the integration measure $dz\wedge dw$ together with $\pa_z$ and $\pa_w$ in the interaction term, so the total weight of the interaction vertex is zero. Each external leg is of weight $3$. Hence, the total weight of the amplitude is $-2-m-n<0$, i.e. it's not invariant under the $U(1)$-rotation symmetry, so the amplitude must be zero.

For the second diagram, we will follow the approach shown in \cite{Costello:2017dso} and show that the diagram has a nonvanishing amplitude if and only if $m=n=r=s=0$. And in the case that it is nonzero, it has a nonvanishing gauge anomaly consequently, under the BRST variation $Q_{BRST}A=\pa c$. 

Let's do the same analysis on the second diagram as the first one, i.e., let $U(1)$ act on $z$ and $w$ by rotation with weight $1$, then the total weight of the amplitude is $-n-m-r-s$. Hence, the diagram is nonzero only if $m=n=r=s=0$. In the following discussion, we will focus on he case $m=n=r=s=0$.
}

We first integrate over the first vertex $(P_1~\pa_{z}^2\pa_{w}A~P_2)$ and then integrate over the second vertex$(P_2~\pa_{z}\pa_{w}^2A~P_3)$.
\newline

\noindent{\bf{First vertex$(P_1~\pa_{z}^2\pa_{w}A~P_2)$}}\\
\newline
First, we focus on computing the integral over the first vertex:
\ie\label{v1integralalg}
\E_1\E_2^2\int_{v_1} dw_1\wedge dz_1\wedge\pa_{z_1}P_1(v_0,v_1)\wedge\pa_{z_2}\pa_{w_1} P_2(v_1,v_2)(z_1^2w_1\pa_{z_1}^2\pa_{w_1}A)
\fe
Note that $\pa_{z_1}$ and $\pa_{w_1}$ comes from the three-point coupling at $v_1$:
\ie
\E_2A\wedge \pa_{z_1}A\wedge\pa_{w_1}A
\fe
And $\pa_{z_2}$ comes from the 3-pt coupling at $v_2$:
\ie
\E_2A\wedge\pa_{z_2}A\wedge\pa_{w_2}A
\fe
We will consider $\pa_{w_2}$ later when we treat the second vertex.

The factor $z_1^2w_1\pa_{z_1}^2\pa_{w_1}A$ is for the external leg attached to $v_1$, which is $c[2,1]$. In short, this is an ansatz, and we can start without fixing $m,n$ in $c[m,n]$. However, we will see that the integral converges to a finite value only with this particular choice of $(m,n)$. For a simple presentation, we will drop $\pa_{z_1}^2\pa_{w_1}A$, and recover it later.

After some manipulation, which we refer to {\bf{Lemma 1}} in \appref{app:2-1}, \eqref{v1integralalg} becomes
\ie\label{crucialsttepo}
-\int_{v_1}dt_1dz_1d\bar z_1dw_1d\bar w_1\frac{\rvert z_1\rvert^2\rvert w_1\rvert^2\bar z_2(\bar w_{12}dt_2-t_{12}d\bar w_2)}{d_{01}^5d_{12}^9}
\fe
The integral \ref{crucialsttepo} can be further simplified by using the typical Feynman integral technique, which can be found in {\bf{Lemma 2}} in \appref{app:2-1}. We are left with
\ie
\bar z_2(\bar w_2dt_2-t_2d\bar w_2)\left(\frac{c_1}{d_{02}^5}+\frac{c_2w_2^2}{d_{02}^7}+\frac{c_3z_2^2}{d_{02}^7}+\frac{c_4z_2^2w_2^2}{d_{02}^9}\right)
\fe
with $c_i$ being a constant. Note that all terms in the parenthesis have a same order of divergence. Therefore, it suffices to focus on the first term to check the convergence of the full integral(we still need to do $v_2$ integral.) 

We will explicitly show the calculation for the first term, and just present the result for the second, third, and fourth term in \eqref{secthirfour}. They are all non-zero and finite. We will denote the first term as $\cP$, which is 1-form. 
\newline

\noindent{\bf{Second vertex$(\cP~\pa_{z_1}^2\pa_{z_2}A~P_3)$}}\\
\newline
Now, let us do the integral over the second vertex$(v_2)$. The remaining things are organized into
\ie\label{v2integralb2}
&\int_{v_2}\cP\wedge\pa_{w_2}P_{3}(v_2,v_3)\wedge dz_2\wedge dw_2 (z_2w_2^2\pa_{z_2}\pa_{w_2}^2A)\\
\fe
where we dropped forms related to $v_3$, as we do not integrate over it. $\pa_{w_2}$ comes from the 3-pt coupling at $v_2$:
\ie
\E_2A\wedge\pa_{z_2}A\wedge\pa_{w_2}A
\fe
The factor $z_2w_2^2\pa_{z_2}\pa_{w_2}^2A$ is for the external leg attached to $v_2$, which corresponds to $c[1,2]$. Again, this is an ansatz. We will see that only this integral converges and does not vanish. We will drop $\pa_{z_2}\pa_{w_2}^2A$ and recover it later.

The integral \eqref{v2integralb2} is simplified to 
\ie\label{crucialsttep2}
\int_{v_2}-\frac{\rvert z_2\rvert^2\rvert w_2\rvert^4}{d^5_{02}d^7_{23}}dt_2d\bar z_2d\bar w_2 dw_2 dz_2
\fe
The intermediate steps can be found in {\bf{Lemma 3}} in \appref{app:2-1}.


Now, it remains to evaluate the delta function at the third vertex and use Feynman technique to evaluate the integral. By {\bf{Lemma 4}} in \appref{app:2-1}, we are left with 
\ie\label{fir}
(const)\E_1\E_2^2t[0,0]t[0,0]\pa_{z_1}^2\pa_{z_2}A_1\pa_{z_1}^1\pa_{z_2}^2A_2
\fe
The BRST variation of the amplitude is
\ie\label{firstbrstv1}
(const)\E_1\E_2^2t[0,0]t[0,0]\pa_{z_1}^2\pa_{z_2}A_1\pa_{z_1}^1\pa_{z_2}^2c_2
\fe
This indicates that the theory is quantum mechanically inconsistent, as it has a Feynman diagram that has nonzero BRST variation. However, as long as there is another diagram whose BRST variation is proportional to the same factors
We can cancel the anomaly.

Hence, imposing BRST invariance of the sum of Feynman diagrams, we bootstrap the possible 1d TQM that can couple to 5d $U(1)$ CS.

An obvious choice is the tree-level diagram where $(\pa_{z_1}A)(\pa_{z_2}A)$ appears explicitly:

By equating \eqref{firstbrstv1} and \eqref{firstbrstv2}, we get
\ie
\left[t[2,1],t[1,2]\right]=\E_1\E_2^2t[0,0]t[0,0]+\ldots
\fe
Therefore, we have reproduced the $\cO(\E_1)$ part of the ADHM algebra $\cA_{\E_1,\E_2}$ commutation relation from the Feynman diagram computation:
\ie\label{notfull}
\left[t[2,1],t[1,2]\right]_{\E_1}=\E_1\E_2^2t[0,0]t[0,0]
\fe
{where $[-,-]_{\E_1}$ is the $\cO(\E_1)$-part of the commutator.}

\section{Perturbative calculations in 5d $U(1)$ CS theory coupled to 2d $\B\gamma$}
\label{sec:5d2dpert}
In this section, we will provide a bulk derivation of the ADHM algebra $\cA_{\E_1,\E_2}$ action on the bi-module $\cM_{\E_1,\E_2}$ of the ADHM algebra $\cA_{\E_1,\E_2}$ using 5d Chern-Simons theory. The strategy is similar to that of the previous section. We will compute the $\cO({\E_1}^1)$ order gauge anomaly of various Feynman diagrams in the presence of the line defect from $M2$ brane$(\bR^1\times\{0\}\subset\bR^1\times\bC^2_{NC})$, and at the same time the surface defect from $M5$ brane on $(\{0\}\times\bC\subset\bR^1\times\bC^2_{NC})$. Imposing a cancellation of the anomaly for the 5d gauge theory uniquely fixes the algebra action on the bi-module. 

We will confirm the representative commutation relation between ADHM algebra and its bimodule {\eqref{sec7relsII}} using the Feynman diagram calculation in 5d Chern-Simons, 1d topological line defect, and 2d $\B\gamma$ coupled system. 
\begin{itemize}
\item{} {The algebra and bimodule commutation relation}
\ie\label{sec7relsII}
\left[t[2,1],b[z^1]c[z^0]\right]_{\E_1}&=\E_1\E_2~t[0,0]c[z^0]b[z^0]+\E_1\E_2~c[z^0]b[z^0]
\fe
\end{itemize}
{where $c[z^n]$ and $b[z^m]$ are elements of the {bi-module.}

\subsection{Ingredients of Feynman diagram}\label{subsec:FeynMod}
Generators of the 0d bimodule $b[z^n]$, $c[z^m]$ couple to the mode of $\B$, $\gamma$ through 
\ie\label{couplinggm}
\int_{\{0\}}\pa^{k_1}_{z_2}\B\cdot b[z^{k_1}]+\int_{\{0\}}\pa^{k_2}_{z_2}\gamma\cdot c[z^{k_2}]
\fe
where $z=z_2$. The coupling is defined at a point, so the integral is only used for a formal presentation.

From the coupling, we learn another ingredient of the 5d-2d Feynman diagram computation:
\begin{itemize}
\item{} One-point vertices from \eqref{couplinggm}:
\ie\label{bc1ptcoup}
\cI^\B_{1pt}&=\begin{cases}b[z^k]\D_{z_2}\quad&\text{if }\pa_{z_2}^k\B\text{ is a part of an internal propagator} \\b[z^k]\pa_{z_2}^{k}\B\quad&\text{if }\pa_{z_2}^{k}\B\text{ is an external leg}\end{cases},\\
\cI^\gamma_{1pt}&=\begin{cases}c[z^k]\D_{z_2}\quad&\text{if }\pa_{z_2}^k\gamma\text{ is a part of an internal propagator} \\c[z^k]\pa_{z_2}^{k}\gamma\quad&\text{if }\pa_{z_2}^{k}\gamma\text{ is an external leg}\end{cases}
\fe
In the case of multiple $\B,\gamma$ internal propagators flowing out, we prescribe to keep only one $\D_{z_2}$ function.
\end{itemize}
The $\B\gamma-$system also couples to the 5d Chern-Simons theory in a canonical way:
{\ie\label{bcCS}
\frac{1}{\E_1}\int_{\mathbb{C}_{z_2}} \B(\pa_{\bar{z}_2}-A_{\bar{z}_2}\star_{\E_2})\gamma
\fe
}
from which we read off the last ingredients of the perturbative computation:
\begin{itemize}
\item{} The $\B\gamma$ propagator $P_{\B\gamma}=\la\B\gamma\ra$ is a solution of  
\ie\label{bcprop1}
\pa_{\bar{z}_2}P_{\B\gamma}=\D_{z_2=0}
\fe
That is,
\ie
P_{\B\gamma}=\la\B\gamma\ra\sim\frac{d z_2}{z_2}
\fe
\item{} The normalized three-point$(\B,A_{5d},\gamma)$ vertex :
\ie\label{abccoupling}
\cI^{\B A\gamma}_{3pt}=1
\fe
\end{itemize}
Note that we are taking the lowest order vertex in the Moyal product expansion of \eqref{bcCS}, and normalize the coefficient to 1 for simplicity in the following computation. Each $\B\gamma$ propagator contributes ${\E_1}$, and each $\B A\gamma$ vertex contributes ${\E_1}^{-1}$. 

Recall that there was a gauge anomaly in the 5d CS theory in the presence of the topological line defect. Similarly, the bi-module coupled to $\B\gamma$-system provides an additional source of the 5d gauge anomaly, since $\B\gamma$ system has the nontrivial coupling \eqref{bcCS} with the 5d CS theory and is charged under the 5d gauge symmetry. For the entire 5d-2d-1d coupled system to be anomaly-free, the combined gauge anomaly should be canceled. The bulk anomaly cancellation condition beautifully fixes the action of the algebra on the bi-module.

{The simplest example involving the bi-module is akin to the first example of \secref{sec:pert}; notice the similarity between Fig \ref{fig1} and Fig \ref{fig7}}. As a result, the calculation in this section resembles that of \secref{subsec:algefeyn1}.

From the ingredients provided above, we can interpret the commutator $\left[t[2,1],b[z^1]c[z^0]\right]$ as the difference between two tree-level diagrams: 
\begin{figure}[H]
\centering
\includegraphics[width=10cm]{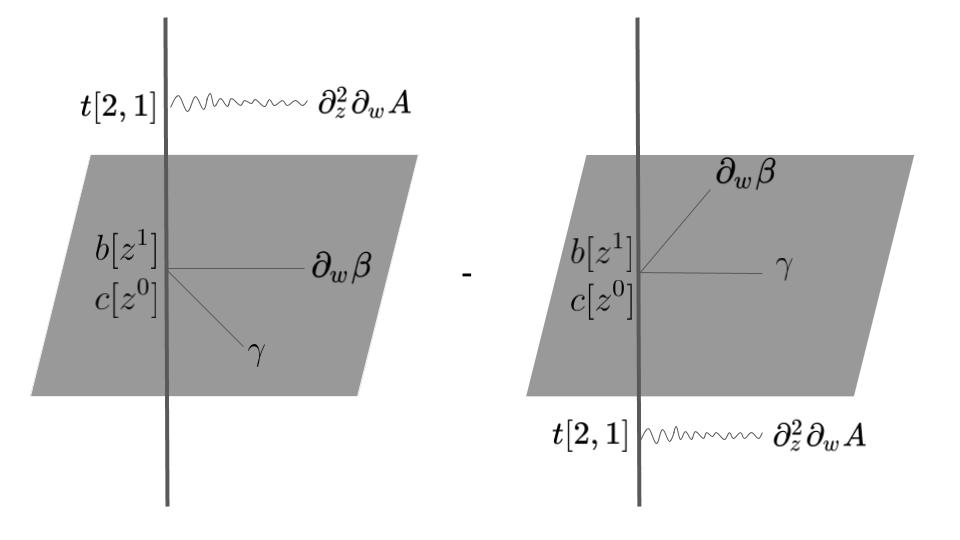}
  \vspace{-20pt}
\caption{Feynman diagrams representing the commutator $[t[2,1],b[z^1]c[z^0]]$. {The vertical straight lines are time axis, and $\B\gamma$ lives on the gray planes. $\B\gamma$ only flows out of the time axis, but not flowing along the time axis.} Note that there is no internal propagators of any sort. All types of lines are external legs, they are modes of $\B$, $\gamma$, $A$.}
\label{fig13}
 \centering
\end{figure}\noindent
As Fig \ref{fig13} does not involve any loops, the amplitude is simply
\ie\label{sec}
\left[t[2,1],b[z^1]c[z^0]\right](\pa_{z}^2\pa_{w}A)(\pa_{w}\B)\gamma
\fe
and its BRST variation is proportional to 
\ie\label{secd}
\left[t[2,1],b[z^1]c[z^0]\right](\pa_{z}^2\pa_{w}c)(\pa_{w}\B)\gamma
\fe
{At $\mathcal O(\E_1)$ level, it will cancel the anomalies coming from all possible one-loop Feynman diagrams with three external legs coupled to $\pa_z^2\pa_w A$, $\gamma$, and $\pa_w\beta$, respectively, so the only possibilities are Figure \ref{fig7} and Figure \ref{fig8}, which we will call the diagram I and the diagram II, respectively.}
\begin{figure}[H]
\centering
\includegraphics[width=11cm]{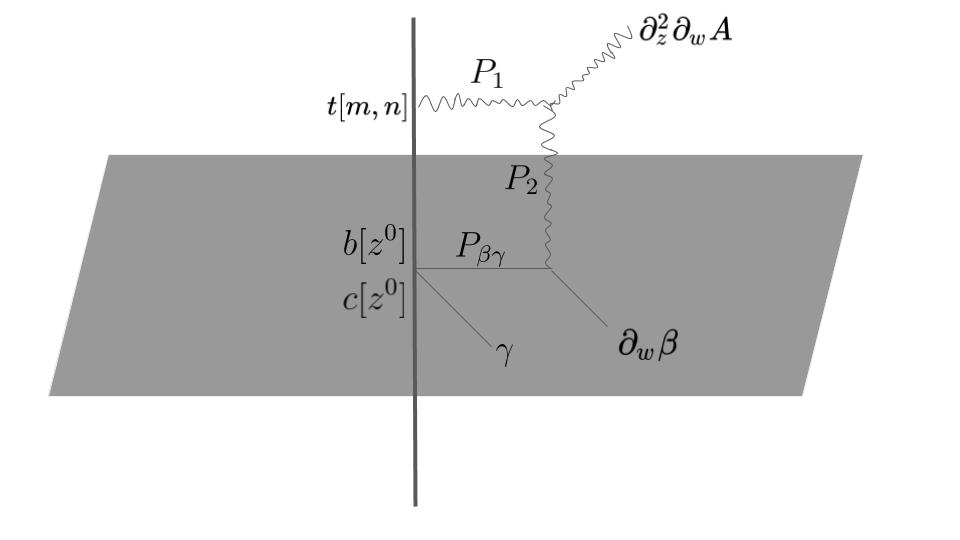}
  \vspace{-15pt}
\caption{The Feynman diagram I. The vertical straight lines are the time axis, and the gray plane is where $\B\gamma$-system is living. The internal horizontal straight lines are $\B\gamma$ propagators and the external slant straight lines are modes of $\B\gamma$. {Note that no $\B\gamma$ propagates along the time axis.} The $\B A\gamma$ three-point vertex is restricted to the $\B\gamma$-plane, but the $AAA$ three-point vertex can be anywhere in the bulk.}
\label{fig7}
 \centering
\end{figure}

\begin{figure}[H]
\centering
\includegraphics[width=10cm]{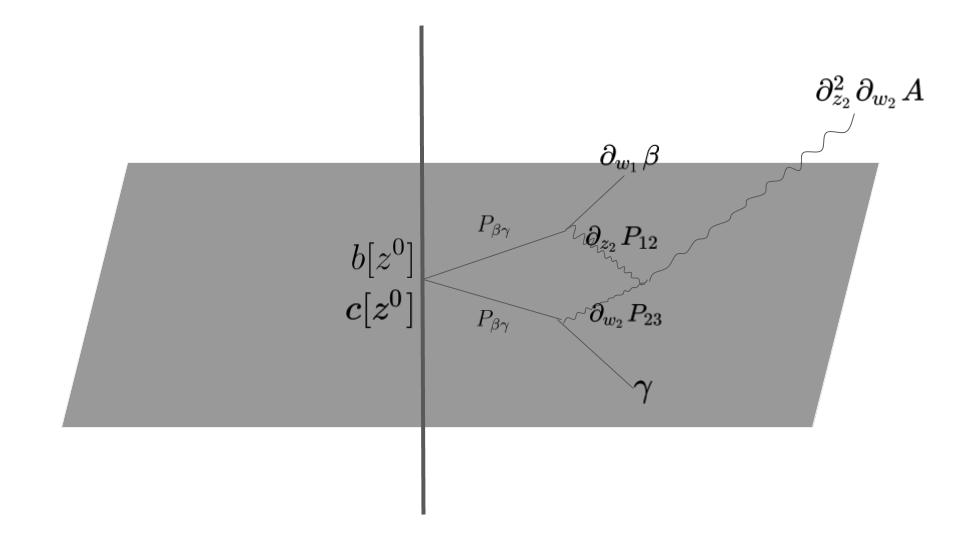}
\caption{The Feynman diagram II.}
\label{fig8}
\centering
\end{figure}
{Before we start doing concrete computations, we make a similar analysis to the ADHM algebra case, i.e., let $U(1)$ rotates the $z$ and $w$ coordinates with weight $1$, then $\beta-\gamma$ propagator has weight $0$, Chern-Simons propagator has weight $-2$ and all interaction vertices have weight zero. It follows that the Feynman diagram I has a total weight of $-m-n$, and the Feynman diagram II has a total weight of $0$. Hence, the amplitude for the Feynman diagram I is nonzero only if $m=n=0$, so in the later discussions we will impose the condition that $m=n=0$.}

\subsection{Feynman diagram I}\label{subsec:simplemod}
In this subsection, we will show that the amplitude of Fig \ref{fig7} is
\ie\label{fig6anomalyp}
(const)~{\E_1} ~\pa_z^2\pa_w A\pa_z\B~\gamma~c[z^0]b[z^0]t[0,0]
\fe
The factor $z_2^2w_2\pa_{z_2}^2\pa_{w_2}A$ is for the external leg attached to the top 3-point vertex, $v_2$. The factor corresponds to {$t[2,1]$}. For the convenience of presentation, We will drop $\pa_{z_2}^2\pa_{w_2}A$ and recover it in the final result.

Along the way, we also show the constant factor in front of \eqref{fig6anomalyp} finite only if the external legs are $\pa_{z}^2\pa_{w}A\pa_z\B\gamma$. For simplicity, we will {abbreviate} the leg factors during the computation. 
\newline

\noindent{\bf{First vertex}}\\
\newline
First, we focus on computing the integral over the first vertex:
\ie\label{v1integral2}
\int_{v_1}\pa_{z_1}P_1(v_0,v_1)\wedge(w_1dw_1)\wedge(z^2_1dz_1)\wedge\pa_{w_1} P_2(v_1,v_2)
\fe
Note that $\pa_{z_1}$ and $\pa_{w_1}$ comes from the three-point coupling at $v_1$:
\ie
\E_2A\wedge \pa_{z_1}A\wedge\pa_{w_1}A
\fe

In {\bf{Lemma 5}} in \appref{app:2-2}, we showed how to evaluate \eqref{v1integral2} and arrive at following expressions.
\ie\label{crucialsttep3b}
-\int^1_0dx\sqrt{x(1-x)}^7\int_{v_1}[dV_1]\frac{(\rvert z_1\rvert^2+x^2\rvert z_2\rvert^2)^2(\rvert w_1\rvert^2+x^2\rvert w_2\rvert^2)t_2d\bar w_2}{(\rvert z_1\rvert^2+\rvert w_1\rvert^2+t_1^2+x(1-x)(\rvert z_2\rvert^2+\rvert w_2\rvert^2+t^2_2))^7}
\fe
where $[dV_1]$ is an integral measure for $v_1$ integral. We see from \eqref{crucialsttep3b} that it was necessary to choose $c[m,n]$, $\beta_n$ to be $c[2,1]$, $\beta_1$. Otherwise, the numerator of \eqref{crucialsttep3b} would contain holomorphic or anti-holomorphic dependence on $z_1$ or $w_1$, and this makes the $z_1$ or $w_1$ integral to vanish.

 Moreover, we can drop the term proportional to $\rvert z_2\rvert^2$, since there is a delta function at the second vertex that evaluates $z_2=0$. So, \eqref{crucialsttep3b} simplifies to
\ie
-\int^1_0dx\sqrt{x(1-x)}^7\int_{v_1}[dV_1]\frac{\rvert z_1\rvert^4(\rvert w_1\rvert^2+x^2\rvert w_2\rvert^2)t_2d\bar w_2}{(\rvert z_1\rvert^2+\rvert w_1\rvert^2+t_1^2+x(1-x)(\rvert z_2\rvert^2+\rvert w_2\rvert^2+t^2_2))^7}
\fe
This is evaluated to 
\ie
\frac{c_1t_2}{d_{02}^3}+\frac{c_2t_2\rvert w_2\rvert^2}{d_{02}^5}
\fe
where $c_1$ and $c_2$ are 1-forms of $v_2$. Let us call them as $\cP^1_{02}$ and $\cP^2_{02}$ respectively.
\newline

\noindent{\bf{Second vertex}}\\
\newline
Now, compute the second vertex integral using the above computation:
\ie
&\int_{v_2}(\cP^1_{02}+\cP^2_{02})\wedge dw_2\frac{1}{w_2}(w_2)\D(z_2=0,t_2=\E)\\
=&~\E_1\int\left(\frac{c_1}{r^5}+\frac{c_2}{r^3}\right)rdrd\theta\\
=&~4\pi^4\E_1\left(\frac{1}{43200\rvert\E\rvert}+\frac{1}{57600\rvert\E\rvert^3}\right)
\fe
We can rescale $\E$ to be 1, so the integral converges. Reinstating Gamma function factors, we finally obtain
\ie
(const)=\frac{\Gamma(7)}{\Gamma(7/2)\Gamma(7/2)}4\pi^4\left(\frac{1}{43200}+\frac{1}{57600}\right)=\frac{112\pi}{3375}
\fe
Hence, the amplitude for the Feynman diagram is
\ie
(const)\E_1\E_2t[0,0]b[z^0]c[z^0](\pa_{z}^2\pa_{w}A)(\pa_{w}\B)\gamma
\fe
Its BRST variation is 
\ie\label{fird}
(const)\E_1\E_2t[0,0]b[z^0]c[z^0](\pa_{z}^2\pa_{w}c)(\pa_{w}\B)\gamma
\fe
By equating \eqref{fird} and \eqref{secd}, we reproduce from the 5d gauge theory (with $\B\gamma$-system) the calculation part of the algebra action on the bi-module, which is
\ie\label{dfdfasd}
\left[t[2,1],b[z^1]c[z^0]\right]=\E_1\E_2t[0,0]b[z^0]c[z^0]+\ldots
\fe

\subsection{Feynman diagram II}\label{diag2}
In this subsection we will reproduce the remaining $\mathcal O(\E_1)$-term in \eqref{sec7relsII}
\ie\label{boxedrelsII}
\left[t[2,1],b[z^1]c[z^0]\right]_{\E_1}=\ldots+\E_1\E_2b[z^0]c[z^0]+\ldots
\fe
by using the Feynman diagram II, see Figure \ref{fig8}.

The amplitude of the diagram is
\ie
(const)\E_2\E_1b[z^0]c[z^0]
\fe
since there are 4 internal propagators$(\E_1^4)$ and 3 internal vertices$(\E_1^{-3})$, one of which is $A\pa A\pa A$ type vertex$(\E_2)$. We will explicitly show that $(const)$ does not vanish and hence the diagram has nonzero BRST variation, which completes the RHS of \eqref{dfdfasd}.
\newline

\noindent{\bf{First vertex}$(P_{\beta\gamma}~\pa_{w_1}\beta~\pa_{z_2}P_{12})$}\\
\newline
First, we focus on computing the integral over the first vertex:
\ie\label{v1integral-1}
\int_{v_1}\frac{1}{w_1}(w_1dw_1)\D(t_1=0,z_1=0)\wedge\pa_{z_2} P_{12}(v_1,v_2)
\fe
Note that $\pa_{w_2}$ comes from the three-point coupling at $v_2$:
\ie
\E_2A\wedge \pa_{z_2}A\wedge\pa_{w_2}A
\fe
This integral evaluates to 
\ie\label{firstv5}
-\frac{2\pi (t_2d\bar z_2+\bar z_2dt_2)\bar z_2}{5\sqrt{t_2^2+\rvert z_2\rvert^2}^5}
\fe
We presented the details in {\bf{Lemma 6.}} in \appref{app:2-3}.
\newline

\noindent{\bf{Third vertex$(P_{\beta\gamma}~\gamma~\pa_{w_2}P_{23})$}}\\
\newline
Second, we focus on computing the integral over the third vertex:
\ie\label{v3integral-1}
\int_{v_3}\frac{1}{w_3}(dw_3)\D(t_3=0,z_3=0)\wedge\pa_{w_2} P(v_2,v_3)
\fe
Note that $\pa_{w_2}$ comes from the three-point coupling at $v_2$:
\ie
\E_2A\wedge \pa_{z_2}A\wedge\pa_{w_2}A
\fe
This integral evaluates to 
\ie\label{thirdv5}
-(t_2d\bar z_2-\bar z_2dt_2)\frac{2\pi}{15w^2_2}\left(\frac{2}{\sqrt{t_2^2+\rvert z_2\rvert^2}^3}-\frac{5\rvert w_2\rvert^2+2t_2^2+2\rvert z_2\rvert^2}{\sqrt{t_2^2+\rvert z_2\rvert^2+\rvert w_2\rvert^2}^5}\right)
\fe
We presented the details in {\bf{Lemma 7.}} in \appref{app:2-3}.
\newline

\noindent{\bf{Second vertex}$(\pa_{z_2}P_{12}~\pa_{z_2}^2\pa_{w_2}A~\pa_{w_2}P_{23})$}\\
\newline
Now, combine \eqref{firstv5} and \eqref{thirdv5}, and compute the second vertex integral; here $z_2^nw_2^m$ denotes the external gauge boson leg.
\ie\label{secondvertex}
&\int_{v_2}dw_2\wedge dz_2\wedge (t_2d\bar z_2-\bar z_2dt_2)\wedge(t_2d\bar z_2+\bar z_2dt_2)\bar z_2\\
&\times\frac{4\pi^2 z_2^nw_2^m}{75w^2_2\sqrt{t_2^2+\rvert z_2\rvert^2}^5}\left(\frac{2}{\sqrt{t_2^2+\rvert z_2\rvert^2}^3}-\frac{5\rvert w_2\rvert^2+2t_2^2+2\rvert z_2\rvert^2}{\sqrt{t_2^2+\rvert z_2\rvert^2+\rvert w_2\rvert^2}^5}\right)\\
=~&\int_{v_2}dw_2\wedge dz_2\wedge d\bar z_2\wedge dt_2\frac{4\pi^2t_2 \rvert z_2\rvert^2}{75w_2\sqrt{t_2^2+\rvert z_2\rvert^2}^5}\left(\frac{2}{\sqrt{t_2^2+\rvert z_2\rvert^2}^3}-\frac{5\rvert w_2\rvert^2+2t_2^2+2\rvert z_2\rvert^2}{\sqrt{t_2^2+\rvert z_2\rvert^2+\rvert w_2\rvert^2}^5}\right) 
\fe
We inserted $(n,m)=(2,1)$ for the external gauge boson leg. Then, $z_2^2$ pairs with $\bar{z}_2^2$, and $w_2$ combines with $1/w_2^2$ to yield $1/w_2$. Since we do not have $d\bar{w}_2$, the integral is a holomorphic integral. If $(n,m)$ were other values, the integral will simply vanish.

In {\bf{Lemma 8.}} in \appref{app:2-3}, we show \eqref{secondvertex} is convergent and bounded as
\ie
c_1< \eqref{secondvertex}<c_2
\fe
where $c_1$, $c_2$ are some finite constants.

Hence, the amplitude of the Feynman diagram is
\ie
(const)\E_1\E_2b[z^0]c[z^0](\pa_{z}^2\pa_{w}A)(\pa_{w}\B)\gamma
\fe
Its BRST variation is therefore nonvanishing:\footnote{We hope there is no confusion between the ghost of the 5d gauge field $\pa_{z}^2\pa_{w}c$ and the module element $c[z^0]$.}
\ie
(const)\E_1\E_2b[z^0]c[z^0](\pa_{z}^2\pa_{w}c)(\pa_{w}\B)\gamma
\fe
This completes the remaining part of the algebra-bi-module commutation relation \eqref{resultmod1}:
\ie
\left[t[2,1],b[z^1]c[z^0]\right]_{\E_1}=\E_1\E_2t[0,0]b[z^0]c[z^0]+\E_1\E_2b[z^0]c[z^0]
\fe

\section{Conclusion}
In this paper, we studied the simplest possible configurations of M2 and M5 branes in the $\Omega-$deformed and topologically twisted M-theory. In particular, we showed the operator algebra living on the M2 branes acts on the operator algebra on M5 brane, and computed the simplest commutators. As the M2 and M5 branes are embedded in $\Omega-$deformed and topologically twisted M-theory, the field theories on the branes have twisted holographic duals in the twisted supergravity. The dual side is interestingly captured by the 5d non-commutative Chern-Simons theory coupled to a topological line defect and a vertex operator algebra. By computing several Feynman diagrams and imposing BRST invariance of the coupled system, we demonstrated that the gravity dual computation can reproduce the operator algebra commutator in the field theory. Lastly, we would like to end the paper with some open questions for future research.

First of all, the derivation of the 5d Chern-Simons theory as a localization of $\Omega-$deformed and topologically twisted 11d supergravity is via the type IIA/M-theory relation. We wonder how one can derive the 5d Chern-Simons theory by a direct localization of 11d supergravity. We hope to study this point in the future.

Second, the system we are considering in this work is the simplest configuration that belongs to the more general framework \cite{Gaiotto:2019wcc}. We can introduce $M2_i$-branes on $\bR_t\times\bC_{\E_i}$ and $M5_I$-branes on $\bC\times\bC_j\times\bC_k$, where $i\in\{1,2,3\}$, $(j,k)\in\{(1,2),(2,3),(3,1)\}$, and $I=\{1,2,3\}\backslash\{j,k\}$. Using the M-theory / type IIB duality, we can map the most general configuration to ``GL-twisted type IIB'' theory \cite{Kapustin:2006pk}, where each M2-brane maps to $(1,0),(0,1),(1,1)$ 1-brane, respectively, and each M5-brane maps to D3-brane whose boundary is provided by $(1,0),(0,1),(1,1)$ 5-branes.

At the corner of the tri-valent vertex, so-called Y-algebra \cite{Gaiotto:2017euk}, which comes from D3-brane boundary degree of freedom \cite{Gaiotto:2008sa,Mikhaylov:2014aoa}, lives. This Vertex Algebra is the most general version of our toy model $\beta\gamma$ system and is labeled by three integers $N_1$, $N_2$, $N_3$, each of which is the number of D3-branes on the three corners of the trivalent graph. Therefore, in principle, one can extend our analysis related to the M5-brane into Y-algebra Vertex Algebra. The Koszul dual object of the the Vertex Algebra was called as universal bimodule $\cB^{N_1,N_2,N_3}_{\E_1,\E_2}$ in \cite{Gaiotto:2019wcc}.

Moreover, our ADHM algebra from $M2_1$-brane has its triality image at $M2_2$-brane and $M2_3$-brane. It was proposed in \cite{Gaiotto:2019wcc} that there is a coproduct structure in $M2_i$-brane algebras in the Coulomb branch algebra language\footnote{It is equally possible to describe the M2-brane algebra in terms of Coulomb branch algebra, as the ADHM theory is a self-mirror in the sense of 3d mirror symmetry \cite{Intriligator:1996ex,deBoer:1996mp}.}. Hence, one can generalize our analysis related to the M2-brane into the most general algebra, obtained by the fusion of three $M2_i$-brane algebras. This was called as universal algebra $\cA^{n_1,n_2,n_3}_{\E_1,\E_2}$ in \cite{Gaiotto:2019wcc}.

\section*{Acknowledgements}
We thank Davide Gaiotto and Kevin Costello for crucial advice and encouragement in various stages of this project, and their comments on our paper. We also thank Miroslav Rapcak for his comments on the draft. JO is grateful to Dongmin Gang, Hee-Cheol Kim, Jaewon Song for their questions during JO's seminar at APCTP. {These questions helped us to shape the outline of the paper.} 

\paragraph{Funding information}
Research of JO was supported in part by Kwanjeong Educational Foundation and by the Visiting Graduate Fellowship Program at the Perimeter Institute for Theoretical Physics and in part by the Berkeley Center of Theoretical Physics. Research at the Perimeter Institute is supported by the Government of Canada through Industry Canada and by the Province of Ontario through the Ministry of Economic Development $\&$ Innovation.

\begin{appendix}
\section{Algebra and bimodule computation}\label{app:bdry}
In this appendix we derive some of the commutation relations for the algebra $\cA_{\E_1,\E_2}$ and the bi-module $\cM_{\E_1,\E_2}$.
\subsection{Algebra}\label{app:alg}
In this subsection we will take a closer look at the algebra $\cA_{\E_1,\E_2}$. We begin with a formal definition of the truncated version of the algebra: the $\mathbb C[\epsilon_1^{\pm},\epsilon_2]$-algebra $\mathcal A^{(N)}$ is generated by $\{X^i_j,Y^i_j,I_i,J^j|1\le i,j\le N\}$ with relations
\ie
\begin{aligned}
\overleftarrow{[X,Y]^i_j}+I_jJ^i=\epsilon_2\delta^i_j\text{ , }[X^i_j,Y^k_l]=\epsilon_1\delta^i_l\delta^k_j&\text{ , }[J^j,I_i]=\epsilon_1\delta^j_i\text{ , }[X^i_j,X^k_l]=[Y^i_j,Y^k_l]=0\\
I_iJ^jS(X^nY^m)^i_j&=(I S(X^nY^m)J)
\end{aligned}
\fe,
where $\overleftarrow{f(X,Y,I,J)}$ means rearranging the expression $f(X,Y,I,J)$ in the order $I<J<X<Y$, $(\cdots)$ means fully contracting all indices, the symbol $S$ means symmetrization. Similarly we define $\overrightarrow{f(X,Y,I,J)}$ as rearranging the expression $f(X,Y,I,J)$ in the order $Y<X<I<J$. The ADHM algebra $\cA_{\E_1,\E_2}$ is the large $N$ limit of $\mathcal A^{(N)}$ where the limit is taken in the sense of the procedure in section \ref{subsec:backreac}. The first relation is the F-term relation, and the following lemma is an obvious consequence of the F-term relation:
\begin{lemma}
\ie
(IS(X^nY^m)J)=\epsilon_2(S(X^nY^m))
\fe
\end{lemma}
\noindent From now on we will use $t_{n,m}$ to denote $(S (X^nY^m))/\epsilon_1$, note that these generators are denoted by $t[n,m]$ in the rest of this paper, but here we use the subscript to make the presentation more compact. The following is clear
\begin{lemma}\label{Lemma_Basic Commutation Relation}
\ie
\begin{aligned}\relax
 [t_{0,0},t_{n,m}]=0\text{ , }&[t_{1,0},t_{n,m}]=mt_{n,m-1}\text{ , }[t_{0,1},t_{n,m}]=nt_{n-1,m}\\
[t_{2,0},t_{n,m}]=2mt_{n+1,m-1}\text{ , }&[t_{1,1},t_{n,m}]=(m-n)t_{n,m}\text{ , }[t_{0,2},t_{n,m}]=-2nt_{n-1,m+1}
\end{aligned}
\fe
\end{lemma}
This means that $t_{0,0}$ is central, the linear span of $t_{2,0},t_{1,1},t_{0,2}$ is isomorphic to $\mathfrak{sl}_2$, and the linear span of $t_{m,n}$ with $m+n=L$ is a representation of $\mathfrak{sl}_2$ of spin $L/2$.

\begin{lemma}\label{Lemma_[X,Y^n]}
\begin{align}
\overleftarrow{[X,Y^n]^i_j}&=n\epsilon_2(Y^{n-1})^i_j-\sum_{a+b=n-1}\overleftarrow{(IY^a)_j(Y^bJ)^i}\\
\overrightarrow{[X,Y^n]^i_j}&=n\epsilon_2(Y^{n-1})^i_j-\sum_{a+b=n-1}\overrightarrow{(IY^a)_j(Y^bJ)^i}
\end{align}
In particular, $\overrightarrow{[X,Y]^i_j}=\overleftarrow{[X,Y]^i_j}$.
\end{lemma}

\begin{lemma}\label{Lemma_Y^nXY^m}
\begin{align}
(Y^nXY^m)^i_j-\overleftarrow{(Y^nXY^m)^i_j}&=-\epsilon_1\sum_{a+b=n-1}(Y^{a+m})^i_j(Y^b)\\
(Y^nXY^m)^i_j-\overrightarrow{(Y^nXY^m)^i_j}&=\epsilon_1\sum_{a+b=m-1}(Y^a)(Y^{b+n})^i_j
\end{align}
\end{lemma}

\noindent Combine Lemma \ref{Lemma_[X,Y^n]} and \ref{Lemma_Y^nXY^m}, we immediately see that
\ie\label{Eqn_[X,Y^n]}
\begin{aligned}\relax
 [X,Y^n]^i_j&=\overleftarrow{[X,Y^n]^i_j}-(Y^nX)^i_j+\overleftarrow{(Y^nX)^i_j}\\
&=n\epsilon_2(Y^{n-1})^i_j-\sum_{a+b=n-1}\overleftarrow{(IY^a)_j(Y^bJ)^i}+\epsilon_1\sum_{a+b=n-1}(Y^{a})^i_j(Y^b)\\
&=n\epsilon_2(Y^{n-1})^i_j-\sum_{a+b=n-1}\overrightarrow{(IY^a)_j(Y^bJ)^i}+\epsilon_1\sum_{a+b=n-1}(Y^{a})^i_j(Y^b)
\end{aligned}
\fe

\begin{proposition}\label{Proposition_(IY^a)(Y^bJ)}
\begin{align}
\overleftarrow{(IY^a)_j(Y^bJ)^i}=\overrightarrow{(IY^a)_j(Y^bJ)^i}+\epsilon_1(Y^a)^i_j(Y^b)-\epsilon_1(Y^a)(Y^b)^i_j
\end{align}
\end{proposition}

\begin{proof}
\begin{align*}
\overleftarrow{(IY^a)_j(Y^bJ)^i}-\overrightarrow{(IY^a)_j(Y^bJ)^i}&=I_lJ^m(Y^a)^l_j(Y^b)^i_m-(Y^a)^l_j(Y^b)^i_mI_lJ^m\\
&=[I_lJ^m,(Y^a)^l_j(Y^b)^i_m]\\
&=[-\overleftarrow{[X,Y]^m_l},(Y^a)^l_j(Y^b)^i_m]\\
&=(Y^a)^l_j(Y^b)^i_m\overrightarrow{[X,Y]^m_l}-\overleftarrow{[X,Y]^m_l}(Y^a)^l_j(Y^b)^i_m\\
&=\overrightarrow{(Y^b[X,Y]Y^a)^i_j}-\overleftarrow{(Y^b[X,Y]Y^a)^i_j}
\end{align*}
where in the third line we used the F-term relation and in the fourth line we used the equation $\overrightarrow{[X,Y]^m_l}=\overleftarrow{[X,Y]^m_l}$ (cf. Lemma \ref{Lemma_[X,Y^n]}). Then the result follows from Lemma \ref{Lemma_Y^nXY^m}. 
\end{proof}

\begin{proposition}
\begin{align}
(Y^c)^i_k\overleftarrow{(IY^a)_j(Y^bJ)^k}&=\overleftarrow{(IY^a)_j(Y^{b+c}J)^i}+\epsilon_1(Y^{a+c})^i_j(Y^b)-\epsilon_1(Y^a)^i_j(Y^{b+c})\\
\overrightarrow{(IY^a)_k(Y^bJ)^i}(Y^c)^k_j&=\overrightarrow{(IY^{a+c})_j(Y^bJ)^i}+\epsilon_1(Y^a)(Y^{b+c})^i_j-\epsilon_1(Y^{a+c})(Y^b)^i_j
\end{align}
\end{proposition}

\begin{proof}
\begin{align*}
(Y^c)^i_k\overleftarrow{(IY^a)_j(Y^bJ)^k}&=(Y^c)^i_k\left(\overrightarrow{(IY^a)_j(Y^bJ)^k}+\epsilon_1(Y^a)^k_j(Y^b)-\epsilon_1(Y^a)(Y^b)^k_j\right)\\
&=\overrightarrow{(IY^a)_j(Y^{b+c}J)^i}+\epsilon_1(Y^{a+c})^k_j(Y^b)-\epsilon_1(Y^a)(Y^{b+c})^k_j\\
&=\overleftarrow{(IY^a)_j(Y^{b+c}J)^i}+\epsilon_1(Y^{a+c})^i_j(Y^b)-\epsilon_1(Y^a)^i_j(Y^{b+c})
\end{align*}
where we used Proposition \ref{Proposition_(IY^a)(Y^bJ)} to move the direction of arrows back and forth. 
\end{proof}

\begin{proposition}
\begin{align}
\frac{(XY^m)}{\epsilon_1}&=t_{1,m}+\frac{1}{m+1}\sum_{k=0}^{m-1}(k+1)(Y^k)(Y^{m-1-k})\\
\frac{(Y^mX)}{\epsilon_1}&=t_{1,m}-\frac{1}{m+1}\sum_{k=0}^{m-1}(k+1)(Y^k)(Y^{m-1-k})
\end{align}
\end{proposition}

\begin{proof}
\begin{align*}
(m+1)\frac{(XY^m)}{\epsilon_1}-(m+1)t_{1,m}&=\frac{1}{\epsilon_1}\sum_{a+b=m}([X,Y^a]Y^b)\\
&=\sum_{r+s+t=m-1}(Y^r)(Y^{s+t})
\end{align*}
Similarly for the other one.
\end{proof}

\noindent{\bf{The Key Commutation Relation}}\\

There is a $\mathrm{SL}_2$-symmetry on the algebra $\mathcal A^{(N)}$ under which $(X,Y)$ transforms as a vector. We will use the following particular transform $$\phi_{\alpha}:X\mapsto X\text{ , }Y\mapsto Y+\alpha X$$where $\alpha$ is a formal parameter. Consider 
\begin{align}
A:=\sum_{a+b=n-1}(([Y^a,X]Y^bX)+(XY^a[X,Y^b]))=\frac{\mathrm{d}}{\mathrm{d}\alpha}\phi_{\alpha}\left((XY^n)+(Y^nX)\right)-2n(XY^{n-1}X)
\end{align}
This leads to 
\begin{equation}
\begin{aligned}
3A+2n([X,Y^{n-1}X])+2n([XY^{n-1},X])&=3\frac{\mathrm{d}}{\mathrm{d}\alpha}\phi_{\alpha}\left((XY^n)+(Y^nX)\right)-2\epsilon_1[t_{3,0},t_{0,n}]\\
&=6n\epsilon_1t_{2,n-1}-2\epsilon_1[t_{3,0},t_{0,n}]
\end{aligned}
\end{equation}
It follows that 
\begin{equation}\label{Eqn_Preparation}
\begin{aligned}
[t_{3,0},t_{0,n}]&=3nt_{2,n-1}-\frac{3A}{2\epsilon_1}-\frac{n}{\epsilon_1}([X,Y^{n-1}X])-\frac{n}{\epsilon_1}([XY^{n-1},X])\\
&=3nt_{2,n-1}-\frac{3A}{2\epsilon_1}+n\sum_{a+b=n-2}\left((XY^a)(Y^b)-(Y^a)(Y^bX)\right)\\
&=3nt_{2,n-1}-\frac{3A}{2\epsilon_1}+n\sum_{a+b=n-2}[(XY^a),(Y^b)]+2n\epsilon_1\sum_{u+v+w=n-3}\frac{u+1}{u+v+2}(Y^u)(Y^v)(Y^w)\\
&=3nt_{2,n-1}-\frac{3A}{2\epsilon_1}+\epsilon_1^2\frac{n(n-1)(n-2)}{2}t_{0,n-3}+n\epsilon_1\sum_{u+v+w=n-3}(Y^u)(Y^v)(Y^w)
\end{aligned}
\end{equation}
We have the following assertion which will be proven in the end of this subsection
\begin{lemma}\label{Lemma_Computation of A}
\begin{equation}
\begin{aligned}
A=&\epsilon_2(\epsilon_1+\epsilon_2)\sum_{m=0}^{n-3}(m+1)(n-2+m)(Y^m)(Y^{n-3-m})-\epsilon_2(\epsilon_1+\epsilon_2){n\choose 3}(Y^{n-3})\\
&+\epsilon_1^2{n\choose 3}(Y^{n-3})+\frac{2n\epsilon_1^2}{3}\sum_{u+v+w=n-3}(Y^u)(Y^v)(Y^w)
\end{aligned}
\end{equation}
\end{lemma}
\noindent Plug it into equation \ref{Eqn_Preparation} and we obtain the following
\begin{proposition}[The Key Commutation Relation]\label{Proposition_The Key Commutation Relation}
Let $\sigma_2=\epsilon_1^2+\epsilon_2^2+\epsilon_1\epsilon_2$ and $\sigma_3=-\epsilon_1\epsilon_2(\epsilon_1+\epsilon_2)$, then
\begin{align}
[t_{3,0},t_{0,n}]&=3nt_{2,n-1}+\frac{3\sigma_2}{2}{n\choose 3}t_{0,n-3}+\frac{3\sigma_3}{2}\sum_{m=0}^{n-3}(m+1)(n-2+m)t_{0,m}t_{0,n-3-m}
\end{align}
\end{proposition}
\noindent Proposition \ref{Proposition_The Key Commutation Relation} together with Lemma \ref{Lemma_Basic Commutation Relation} actually determine all other commutation relations as following: first of all, we have 
\begin{align*}
[t_{3,0},t_{n,m}]&=\frac{1}{2^n m!{{n+m}\choose m}}\mathrm{ad} _{t_{2,0}}^n([t_{3,0},t_{0,n+m}])\\
&=3mt_{n+2,m-1}+\frac{3\sigma_2}{2}{m\choose 3}t_{n,m-3}\\
&+\frac{3\sigma_3}{2}\sum_{b=0}^{m-3}\sum_{a=0}^{n}(a+1)(n-a+1)\frac{{{a+b+1}\choose {a+1}}{{m+n-a-b-2}\choose{n-a+1}}}{{{n+m}\choose m}}t_{a,b}t_{n-a,m-3-b}
\end{align*}
then for $a+b=3$, $[t_{a,b},t_{n,m}]$ is obtained by applying $\mathrm{ad}_{t_{0,2}}$ to $[t_{3,0},t_{n',m'}]$. Suppose that $[t_{a,b},t_{n,m}]$ is obtained for all $a+b\le k$ and all pairs $(n,m)$, then $[t_{k+1,0},t_{n,m}]$ can be obtained by applying $\mathrm{ad}_{t_{3,0}}$ to $[t_{k-1,1},t_{n,m}]$. Hence, the general $[t_{a,b},t_{n,m}]$ is obtained by induction on $k$.
\begin{proof}[Proof of Lemma \ref{Lemma_Computation of A}]
\begin{equation}
\begin{aligned}
A=&\sum_{a+b=n-1}(([Y^a,X]Y^bX)+(XY^a[X,Y^b]))\\
=&\sum_{a+b=n-1}\left(\sum_{s+t=a-1}\overleftarrow{(IY^{s+b})_i(Y^tJ)^j}X^i_j-a\epsilon_2(Y^{n-2}X)-\epsilon_1\sum_{s+t=a-1}(Y^{t})(Y^{s+b}X)\right)\\
&+\sum_{a+b=n-1}\left(b\epsilon_2(XY^{n-2})+\epsilon_1\sum_{s+t=b-1}(XY^{s})(Y^{t+a})-\sum_{s+t=b-1}X^i_j\overleftarrow{(IY^{s})_i(Y^{t+a}J)^j}\right)\\
=&\sum_{a+b=n-1}\left(\sum_{s+t=a-1}I_kJ^l\overrightarrow{(Y^{s+b}XY^t)^k_l}-a\epsilon_2(Y^{n-2}X)-\epsilon_1\sum_{s+t=a-1}(Y^{t})(Y^{s+b}X)\right)\\
&+\sum_{a+b=n-1}\left(b\epsilon_2(XY^{n-2})+\epsilon_1\sum_{s+t=b-1}(XY^{s+a})(Y^{t})-\sum_{s+t=b-1}I_kJ^l\overleftarrow{(Y^{s}XY^{t+a})^k_l}\right)\\
=&\epsilon_2{n\choose 2}([X,Y^{n-2}])+\epsilon_1\sum_{r+s+t=n-2}\left((XY^{r+s})(Y^{s})-(Y^t)(Y^{r+s}X)\right)\\
&+\sum_{r+s+t=n-2}I_kJ^l\overleftarrow{(Y^{r}[Y^s,X]Y^{t})^k_l}-\epsilon_1\epsilon_2\sum_{r+s+t+u=n-3}(Y^{r+s+t})(Y^u)\\
&-\epsilon_1\epsilon_2\sum_{r+s+t=n-3}(r+t+2)(Y^{r+s})(Y^t)\\
=&\sum_{r+s+t+u=n-3}I_kJ^l\overleftarrow{(IY^{r+s})_l(Y^{t+u}J)^k}-\sum_{r+s+t=n-2}s\epsilon_2 I_kJ^l\overleftarrow{(Y^{r+s+t-1})^k_l}\\
&+\epsilon_1\sum_{r+s+t=n-2}\left((XY^{r+s})(Y^{t})-(Y^t)(Y^{r+s}X)\right)\\
=&\epsilon_2(\epsilon_1+\epsilon_2)\sum_{m=0}^{n-3}(m+1)(n-2+m)(Y^m)(Y^{n-3-m})-\epsilon_2(\epsilon_1+\epsilon_2){n\choose 3}(Y^{n-3})\\
&+\epsilon_1\sum_{r+s+t=n-2}\left((XY^{r+s})(Y^{t})-(Y^t)(Y^{r+s}X)\right)\\
=&\epsilon_2(\epsilon_1+\epsilon_2)\sum_{m=0}^{n-3}(m+1)(n-2+m)(Y^m)(Y^{n-3-m})-\epsilon_2(\epsilon_1+\epsilon_2){n\choose 3}(Y^{n-3})\\
&+\epsilon_1^2\sum_{r+s+t=n-2}[t_{1,r+s},(Y^t)]+2\epsilon_1^2\sum_{u+v+w=n-3}(u+1)(Y^u)(Y^v)(Y^w)\\
=&\epsilon_2(\epsilon_1+\epsilon_2)\sum_{m=0}^{n-3}(m+1)(n-2+m)(Y^m)(Y^{n-3-m})-\epsilon_2(\epsilon_1+\epsilon_2){n\choose 3}(Y^{n-3})\\
&+\epsilon_1^2{n\choose 3}(Y^{n-3})+\frac{2n\epsilon_1^2}{3}\sum_{u+v+w=n-3}(Y^u)(Y^v)(Y^w)
\end{aligned}
\end{equation}

\noindent Some explanation: from the 5th equality to the 6th equality, the essential computation is the following:
\begin{equation}
\begin{aligned}
&\sum_{r+s+t+u=n-3}I_kJ^l\overleftarrow{(IY^{r+s})_l(Y^{t+u}J)^k}=\sum_{r+s+t+u=n-3}I_kJ^lI_iJ^j(Y^{r+s})^i_l(Y^{t+u})^k_j\\
&=\sum_{r+s+t+u=n-3}I_kI_iJ^jJ^l(Y^{r+s})^i_l(Y^{t+u})^k_j+\epsilon_1\sum_{r+s+t+u=n-3}I_kJ^j\delta_i^l(Y^{r+s})^i_l(Y^{t+u})^k_j\\
&=\sum_{r+s+t+u=n-3}I_kJ^jI_iJ^l(Y^{r+s})^i_l(Y^{t+u})^k_j-\epsilon_1\sum_{r+s+t+u=n-3}I_kJ^l\delta_i^j(Y^{r+s})^i_l(Y^{t+u})^k_j\\
&+\epsilon_1\sum_{r+s+t+u=n-3}I_kJ^j\delta_i^l(Y^{r+s})^i_l(Y^{t+u})^k_j\\
&=\epsilon_2(\epsilon_1+\epsilon_2)\sum_{r+s+t+u=n-3}(Y^{r+s})(Y^{t+u})-\epsilon_2\epsilon_1{n\choose 3}(Y^{n-3})
\end{aligned}
\end{equation}
Then we define $m=r+s$, then there are $m+1$ ways of decomposing $m$ as $r+s$, similarly there are $n-3-m+1$ ways of decomposing $n-3-m$ as $t+u$, hence the result can be simplified to 
\begin{equation}
\begin{aligned}
\epsilon_2(\epsilon_1+\epsilon_2)\sum_{m=0}^{n-3}(m+1)(n-2+m)(Y^m)(Y^{n-3-m})-\epsilon_2\epsilon_1{n\choose 3}(Y^{n-3})
\end{aligned}
\end{equation}
\end{proof}

\subsection{Bi-module}\label{app:bimod}
Simplest algebra, bi-module commutator that has $\E_1$ correction in the RHS is
\ie
\left[T[2,1],b[z]c[1]\right]=&\bigg(-\frac{5}{3}\E_2T[0,1]+\E_2^2b[1]c[1]\bigg)\\
&+\E_1\bigg({-\E_2b[1]c[1]T[0,0]}+{\frac{4}{3}\E_2b[1]c[1]}\bigg)\\
&+\E_1^2\bigg(-\frac{4}{3}b[1]c[1]T[0,0]\bigg)\\
&+\E_1^3\bigg(-\frac{1}{3}b[1]c[1]b[1]c[1]\bigg)
\fe
We will prove it in this section. 

Let us expand the LHS.
\ie\label{nnl}
\left[S(X^2Y),(IY\tilde\varphi)(\varphi J)\right]=&\frac{1}{3}(XXY+XYX+YXX)\cdot(IY\tilde\varphi)(\varphi J)\\
&-\frac{1}{3}(IY\tilde\varphi)(\varphi J)\cdot(XXY+XYX+YXX)
\fe
Compute the first term:
\ie
&(XXY)\cdot(IY\tilde\varphi)(\varphi J)=X^0_1X^1_2\rvert\tilde\varphi^b\varphi_c\rvert I_aY^a_bJ^cY^2_0+ X^0_1X^1_2\rvert\tilde\varphi^b\varphi_c\tilde\varphi^2\varphi_0\rvert I_aY^a_bJ^c\\
=&\rvert\tilde\varphi^b\varphi_c\rvert I_aX^0_1(\E_1\D^1_b\D^a_2+Y^a_bX^1_2)J^cY^2_0+\E_1X^0_1\rvert\tilde\varphi^b(\D^1_c\varphi_0+\D^1_0\varphi_c)\rvert I_aY^a_bJ^c\\
=&\E_1\rvert\tilde\varphi^b\varphi_c\rvert I_2X^0_bJ^cY^2_0+\E_1\rvert\tilde\varphi^b\varphi_c\rvert I_a(\E_1\D^0_b\D^a_1+Y^a_bX^0_1)X^1_2J^cY^2_0+\E_1\rvert\tilde\varphi^b\varphi_0\rvert I_aX^0_cY^a_bJ^c\\
&+\E_1\rvert\tilde\varphi^b\varphi_c\rvert I_a(X)Y^a_bJ^c\\
=&\E_1(-\E_1)(IYJ)+\E_1\rvert\tilde\varphi^0\varphi_c\rvert I_1J^cX^1_2Y^2_0+\underline{(IY\tilde\varphi)(\varphi J)(X^2Y)}+(-\E_1)\E_1(IYJ)\\
&+\E_1\rvert\tilde\varphi^b\varphi_c\rvert I_a(\E_1\D^a_b+Y^a_b(X))J^c\\
=&-\E_1^2\E_2(Y)+\E_1(IXY\tilde\varphi)(\varphi J)+\underline{(IY\tilde\varphi)(\varphi J)\cdot(XXY)}-\E_1^2\E_2(Y)\\
&+\E_1^2(I\tilde\varphi)(\varphi J)+\E_1(IY\tilde\varphi)(\varphi J)(X)\\
=&-2\E_1^2\E_2(Y)+\E_1(IXY\tilde\varphi)(\varphi J)+\underline{(IY\tilde\varphi)(\varphi J)\cdot(XXY)}+\E_1^2(I\tilde\varphi)(\varphi J)\\
&+\E_1(IY\tilde\varphi)(\varphi J)(X)
\fe
Therefore,
\ie
\left[(XXY),(IY\tilde\varphi)(\varphi J)\right]=&-2\E_1^2\E_2(Y)+\E_1(IXY\tilde\varphi)(\varphi J)+\E_1^2(I\tilde\varphi)(\varphi J)\\
&+\E_1(IY\tilde\varphi)(\varphi J)(X)
\fe
Next,
\ie
&(XYX)\cdot(IY\tilde\varphi)(\varphi J)=X^0_1Y^1_2\rvert\tilde\varphi^b\varphi_c\rvert I_a(\E_1\D^2_b\D^a_0+Y^a_bX^2_0)J^c\\
=&\E_1\rvert\tilde\varphi^2\varphi_c\rvert I_0X^0_1Y^1_2J^c+\E_1\rvert\tilde\varphi^2\varphi_c\tilde\varphi^1\varphi_2\rvert I_0X^0_1J^c+\rvert\tilde\varphi^b\varphi_c\rvert I_aX^0_1Y^1_2Y^a_bX^2_0J^c\\
&+\rvert\tilde\varphi^b\varphi_c\tilde\varphi^1\varphi_2\rvert I_aX^0_1Y^a_bX^2_0J^c\\
=&\E_1(IXY\tilde\varphi)(\varphi J)+\E_1(-\E_1)((\tilde\varphi\varphi)(IJ)+(I\tilde\varphi)(\varphi J))\\
&+\rvert\tilde\varphi^b\varphi_c\rvert I_a(\E_1\D^0_b\D^a_1+Y^a_bX^0_1)J^cY^1_2X^2_0+(-\E_1)(\rvert\tilde\varphi^b\varphi_2\rvert I_aY^a_bX^2_0J^0+\rvert\tilde\varphi^b\varphi_c\rvert I_aY^a_bJ^c(X))\\
=&\E_1(IXY\tilde\varphi)(\varphi J)-\E_1^2(\tilde\varphi\varphi)(IJ)-\E_1^2(I\tilde\varphi)(\varphi J)+\E_1\rvert\tilde\varphi^0\varphi_c\rvert I_1J^cY^1_2X^2_0\\
&+\underline{(IY\tilde\varphi)(\varphi J)(XYX)}-\E_1\rvert\tilde\varphi^b\varphi_2\rvert I_a(-\E_1\D^a_0\D^2_b+X^2_0Y^a_b)J^0-\E_1(IY\tilde\varphi)(\varphi J)(X)\\
=&\E_1(IXY\tilde\varphi)(\varphi J)-\E_1^2(\tilde\varphi\varphi)(IJ)-\E_1^2(I\tilde\varphi)(\varphi J)+\E_1\rvert\tilde\varphi^0\varphi_c\rvert I_1J^c(-\E_1N\D^1_0+X^2_0Y^1_2)\\
&+\underline{(IY\tilde\varphi)(\varphi J)(XYX)}+\E_1^2(\tilde\varphi\varphi)(IJ)-\E_1(-\E_1)(IYJ)\\
=&\E_1(IXY\tilde\varphi)(\varphi J)-\E_1^2(I\tilde\varphi)(\varphi J)-\E_1^2N(I\tilde\varphi)(\varphi J)-\E_1^2(IYJ)+\E_1^2(IYJ)\\
&+\underline{(IY\tilde\varphi)(\varphi J)(XYX)}\\
=&\E_1(IXY\tilde\varphi)(\varphi J)-\E_1^2(I\tilde\varphi)(\varphi J)-\E_1^2N(I\tilde\varphi)(\varphi J)+\underline{(IY\tilde\varphi)(\varphi J)(XYX)}
\fe
Therefore,
\ie
\left[(XYX),(IY\tilde\varphi)(\varphi J)\right]=\E_1(IXY\tilde\varphi)(\varphi J)-\E_1^2(I\tilde\varphi)(\varphi J)-\E_1^2N(I\tilde\varphi)(\varphi J)
\fe
Next,
\ie
&(YXX)\cdot(IY\tilde\varphi)(\varphi J)=Y^0_1\rvert\tilde\varphi^b\varphi_c\rvert I_aX^1_2(\E_1\D^2_b\D^a_0+Y^a_bX^2_0)J^c\\
=&\E_1Y^0_1\rvert\tilde\varphi^2\varphi_c\rvert I_0X^1_2J^c+Y^0_1\rvert\tilde\varphi^b\varphi_c\rvert I_a(\E_1\D^1_b\D^a_2+Y^a_bX^1_2)X^2_0J^c\\
=&\E_1(-\E_1)(IYJ)+\E_1Y^0_1\rvert\tilde\varphi^1\varphi_c\rvert I_aX^a_0J^c+\rvert\tilde\varphi^b\varphi_c\tilde\varphi^0\varphi_1\rvert I_aY^a_bX^1_2X^2_0J^c+\underline{(IY\tilde\varphi)(\varphi J)(YXX)}\\
=&-\E_1^2\E_2(Y)+\E_1(IXY\tilde\varphi)(\varphi J)+\E_1(-N\E_1)(I\tilde\varphi)(\varphi J)\\
&+\E_1\rvert\tilde\varphi^1\varphi_c\varphi^0\varphi_1\rvert I_aX^a_0J^c+\rvert\tilde\varphi^b\varphi_c\varphi^0\varphi_1\rvert I_a(-\E_1\D^a_2\D^1_b+X^1_2Y^a_b)X^2_0J^c+\underline{(IY\tilde\varphi)(\varphi J)(YXX)} \\
=&-\E_1^2\E_2(Y)+\E_1(IXY\tilde\varphi)(\varphi J)-N\E_1^2(I\tilde\varphi)(\varphi J)+\E_1(-\E_1)(\tilde\varphi\varphi)(I\tilde\varphi)(\varphi J)\\
&+\E_1(-\E_1)(\tilde\varphi\varphi)(IJ)-\E_1\rvert\tilde\varphi^1\varphi_c\tilde\varphi^0\varphi_1\rvert I_2X^2_0J^c\\
&+(-\E_1)(\rvert\tilde\varphi^b\varphi_c\rvert I_aY^a_bJ^c(X)+\rvert\tilde\varphi^0\varphi_c\rvert I_aY^a_2X^2_0J^c)+\underline{(IY\tilde\varphi)(\varphi J)(YXX)}\\
=&-\E_1^2\E_2(Y)+\E_1(IXY\tilde\varphi)(\varphi J)-N\E_1^2(I\tilde\varphi)(\varphi J)-\E_1^2(\tilde\varphi\varphi)(I\tilde\varphi)(\varphi J)-\E_1^2(\tilde\varphi\varphi)(IJ)\\
&-\E_1(-\E_1)(I\tilde\varphi)(\varphi J)-\E_1(-\E_1)(\tilde\varphi\varphi)(IJ)-\E_1(IY\tilde\varphi)(\varphi J)(X)\\
&-\E_1(-\E_1N)(I\tilde\varphi)(\varphi J)-\E_1(-\E_1)(IYJ)+\underline{(IY\tilde\varphi)(\varphi J)(YXX)}\\
=&\E_1(IXY\tilde\varphi)(\varphi J)-\E_1(IY\tilde\varphi)(\varphi J)(X)+\E_1^2(I\tilde\varphi)(\varphi J)-\E_1^2(\tilde\varphi\varphi)(I\tilde\varphi)(\varphi J)\\
&+\underline{(IY\tilde\varphi)(\varphi J)(YXX)}
\fe
Therefore,
\ie
\left[(YXX),(IY\tilde\varphi)(\varphi J)\right]=&\E_1(IXY\tilde\varphi)(\varphi J)-\E_1(IY\tilde\varphi)(\varphi J)(X)+\E_1^2(I\tilde\varphi)(\varphi J)\\
&-\E_1^2(\tilde\varphi\varphi)(I\tilde\varphi)(\varphi J)
\fe
Collecting the above, we have
\ie\label{finalfi1}
\bigg[S(X^2Y),&(IY\tilde\varphi)(\varphi J)\bigg]=\frac{1}{3}\bigg(-2\E_1^2\E_2(Y)+\E_1(IXY\tilde\varphi)(\varphi J)+\E_1^2(I\tilde\varphi)(\varphi J)\\
&+\E_1(IY\tilde\varphi)(\varphi J)(X)+\E_1(IXY\tilde\varphi)(\varphi J)-\E_1^2(I\tilde\varphi)(\varphi J)-\E_1^2N(I\tilde\varphi)(\varphi J)\\
&+\E_1(IXY\tilde\varphi)(\varphi J)-\E_1(IY\tilde\varphi)(\varphi J)(X)+\E_1^2(I\tilde\varphi)(\varphi J)-\E_1^2(\tilde\varphi\varphi)(I\tilde\varphi)(\varphi J)\bigg)\\
&=\E_1(IXY\tilde\varphi)(\varphi J)-\frac{2}{3}\E_1^2\E_2(Y)-\frac{1}{3}\E_1^2N(I\tilde\varphi)(\varphi J)-\frac{1}{3}\E_1^2(\tilde\varphi\varphi)(I\tilde\varphi)(\varphi J)\\
&+\frac{1}{3}\E_1^2(I\tilde\varphi)(\varphi J)
\fe
We are not done yet, since $(IXY\tilde\varphi)(\varphi J)$ is reducible by the F-term relation.
\ie
\E_1\rvert\tilde\varphi^0\varphi_c\rvert I_1J^cX^1_2Y^2_0=&\E_1\rvert\tilde\varphi^0\varphi_c\rvert I_1J^c(X^2_0Y^1_2-(I_0J^1-\E_2\D^1_0))\\
=&\E_1(-\E_1)(IYJ)-\E_1\rvert\tilde\varphi^0\varphi_c\rvert(J^cI_1-\E_1\D_1^c)I_0J^1+\E_1\E_2(I\tilde\varphi)(\varphi J)\\
=&-\E_1^2(IYJ)-\E_1\rvert\tilde\varphi^0\varphi_c\rvert (I_0J^c+\E_1\D^c_0)I_1J^1+\E_1^2(I\tilde\varphi)(\varphi J)\\
&+\E_1\E_2(I\tilde\varphi)(\varphi J)\\
=&-\E_1^2(IYJ)-\E_1(I\tilde\varphi)(\varphi J)(IJ)-\E_1^2(\tilde\varphi\varphi)(IJ)+\E_1^2(I\tilde\varphi)(\varphi J)\\
&+\E_1\E_2(I\tilde\varphi)(\varphi J)
\fe
Plugging this into \eqref{finalfi1}, we get
\ie
\left[S(X^2Y),(IY\tilde\varphi)(\varphi J)\right]&=(-\E_1^2(IYJ)-\E_1(I\tilde\varphi)(\varphi J)(IJ)-\E_1^2(\tilde\varphi\varphi)(IJ)+\E_1^2(I\tilde\varphi)(\varphi J)\\
&+\E_1\E_2(I\tilde\varphi)(\varphi J))-\frac{2}{3}\E_1^2\E_2(Y)-\frac{1}{3}\E_1^2(\tilde\varphi\varphi)(I\tilde\varphi)(\varphi J)\\
&-\frac{1}{3}\E_1^2N(I\tilde\varphi)(\varphi J)+\frac{1}{3}\E_1^2(I\tilde\varphi)(\varphi J)
\fe
After normalization, by multiplying $\frac{\E_2}{\E_1^3}$ both sides, and using the identity \footnote{The identity can be derived using the F-term relation:
\ie
\tilde\varphi^i\bigg([X,Y]^j_i+I_iJ^j-\E_2\D^j_i\bigg)\varphi_j&=0\\
(Y)-(Y)+(I\tilde\varphi)(\varphi J)-\E_2(\tilde\varphi\varphi)&=0\\
(I\tilde\varphi)(\varphi J)&=\E_2(\tilde\varphi\varphi)
\fe}
\ie
(\tilde\varphi\varphi)\E_2=(I\tilde\varphi)(\varphi J)
\fe
we have
\ie
\left[T[2,1],b[z]c[1]\right]=&\bigg(-\frac{5}{3}\E_2T[0,1]+\E_2^2b[1]c[1]\bigg)\\
&+\E_1\bigg({-\E_2b[1]c[1]T[0,0]}+{\frac{4}{3}\E_2b[1]c[1]}\bigg)\\
&+\E_1^2\bigg(-\frac{4}{3}b[1]c[1]T[0,0]\bigg)\\
&+\E_1^3\bigg(-\frac{1}{3}b[1]c[1]b[1]c[1]\bigg)
\fe

\section{Intermediate steps in Feynman diagram computation}\label{app:2}
\subsection{Intermediate steps in section 4.2}\label{app:2-1}
{\bf{Lemma 1.}}
\newline
We will compute the following integral.
\ie
\E_1\E_2^2\int_{v_1} dw_1\wedge dz_1\wedge\pa_{z_1}P_1(v_0,v_1)\wedge\pa_{z_2}\pa_{w_1} P_2(v_1,v_2)(z_1^2w_1\pa_{z_1}^2\pa_{w_1}A)
\fe
Computing the partial derivatives, we can re-write it as
\ie
\E_1\E_2^2\left(\frac{\bar z_1}{d_{01}^{2}}\frac{\bar{w}_1}{d^4_{12}}(w_1z_1\bar z_2)\right)\left[P(v_0,v_1)\wedge dw_1\wedge z_1dz_1\wedge P(v_1,v_2)\right]
\fe
{Note that we ignore all constant factors here.} We see that
\ie
P(v_0,v_1)\wedge P(v_1,v_2)=\frac{d\bar z_1d\bar w_1 dt_1}{d_{01}^5d_{12}^5}&(\bar z_{01}\bar w_{12}dt_2-\bar z_{01}t_{12}d\bar w_2+\bar w_{01}t_{12}d\bar z_2\\
&-\bar w_{01}\bar z_{12}dt_2+t_{01}\bar z_{12}d\bar w_2-t_{01}\bar w_{12}d\bar z_{2})
\fe
Including $\wedge dw_1\wedge(z_1dz_1)\wedge$, we can simplify it:
\ie
&P(v_0,v_1)\wedge P(v_1,v_2)\wedge(w_1dw_1)\wedge(z_1dz_1)=d\bar z_1dz_1dw_1d\bar w_1dt_1\left(\rvert z_1\rvert^2\rvert w_1\rvert^2\bar z_2\right)\times\\
&\bigg[\pa_{\bar z_0}\bigg(\frac{\bar z_{01}\bar w_{12}dt_2-\bar z_{01}t_{12}d\bar w_2+\bar w_{01}t_{12}d\bar z_2-\bar w_{01}\bar z_{12}dt_2+t_{01}\bar z_{12}d\bar w_2-t_{01}\bar w_{12}d\bar z_{12}}{d_{01}^{5}d_{12}^{9}}\bigg)\\
&-\frac{\pa_{\bar z_0}(\bar z_{01}\bar w_{12}dt_2-\bar z_{01}t_{12}d\bar w_2+\bar w_{01}t_{12}d\bar z_2-\bar w_{01}\bar z_{12}dt_2+t_{01}\bar z_{12}d\bar w_2-t_{01}\bar w_{12}d\bar z_{12})}{d_{01}^{5}d_{12}^{9}}\bigg]
\fe
By integration by parts, the the integral over $t_1$, $z_1$, $\bar z_1$, $w_1$, $\bar w_1$ of all terms in the first two lines vanishes.

Therefore, we are left with
\ie\label{crucialsttepappp}
-\int_{v_1}dt_1dz_1d\bar z_1dw_1d\bar w_1\frac{\rvert z_1\rvert^2\rvert w_1\rvert^2\bar z_2(\bar w_{12}dt_2-t_{12}d\bar w_2)}{d_{01}^5d_{12}^9}
\fe
\newline
{\bf{Lemma 2.}}\\
\newline
We can use Feynman integral technique to convert \eqref{crucialsttepappp} to the following:
\ie
&\int_{v_1} \int^1_0 dx\frac{\Gamma(7)}{\Gamma(5/2)\Gamma(9/2)}\frac{\sqrt{x^3(1-x)^7}\rvert z_1\rvert^2\rvert w_1\rvert^2\bar z_2(\bar w_{12}dt_2-t_{12}d\bar w_2)}{((1-x)(\rvert z_1\rvert^2+\rvert w_1\rvert^2+t_1^2)+x(\rvert z_{12}\rvert^2+\rvert w_{12}\rvert^2+t^2_{12}))^7}\\
&=\int_{v_1}\int^1_0 dx\frac{(\Gamma \text{ factors})\sqrt{x^3(1-x)^7}\rvert z_1\rvert^2\rvert w_1\rvert^2\bar z_2(\bar w_{12}dt_2-t_{12}d\bar w_2)}{(\rvert z_1-xz_2\rvert^2+\rvert w_1-xw_2\rvert^2+(t_1-xt_2)^2+x(1-x)(\rvert z_2\rvert^2+\rvert w_2\rvert^2+t^2_2))^7}
\fe
Shift the integral variables as
\ie
z_1\ria z_1+xz_2,\quad w_1\ria w_1+xw_2,\quad t_1\ria t_1+xt_2
\fe
Then the above becomes
\ie
\int_{v_1}\int^1_0dx\frac{\Gamma(7)}{\Gamma(5/2)\Gamma(9/2)}&\frac{\sqrt{x^3(1-x)^7}\rvert z_1+xz_2\rvert^2\rvert w_1+xw_2\rvert^2\bar z_2}{(\rvert z_1\rvert^2+\rvert w_1\rvert^2+t_1^2+x(1-x)(\rvert z_2\rvert^2+\rvert w_2\rvert^2+t^2_2))^7}\\
&\times((\bar w_1+(x-1)\bar w_2)dt_2-(t_1+(x-1)t_2)d\bar w_2)
\fe
Drop terms with odd number of $t_1$ and terms that has holomorphic or anti-holomorphic dependence on $z_1$ or $w_1$:
\ie
\int_{v_1}\int^1_0dx\frac{\Gamma(7)}{\Gamma(5/2)\Gamma(9/2)}\frac{\sqrt{x^3(1-x)^9}(\rvert z_1\rvert^2+x^2\rvert z_2\rvert^2)(\rvert w_1\rvert^2+x^2\rvert w_2\rvert^2)\bar z_2(\bar w_2dt_2-t_2d\bar w_2)}{(\rvert z_1\rvert^2+\rvert w_1\rvert^2+t_1^2+x(1-x)(\rvert z_2\rvert^2+\rvert w_2\rvert^2+t^2_2))^7}
\fe
After doing the $v_1$ integral using Mathematica with the integral measure \\
$dt_1dz_1d\bar z_1dz_2d\bar z_2$, we get
\ie\label{appequ}
\bar z_2(\bar w_2dt_2-t_2d\bar w_2)\left(\frac{c_1}{d_{02}^5}+\frac{c_2w_2^2}{d_{02}^7}+\frac{c_3z_2^2}{d_{02}^7}+\frac{c_4z_2^2w_2^2}{d_{02}^9}\right)
\fe
\newline
{\bf{Lemma 3.}}
\newline
We will compute the integral over the second vertex.
\ie
&\int_{v_2}\cP\wedge\pa_{w_2}P_{3}(v_2,v_3)\wedge dz_2\wedge dw_2 (z_2w_2^2\pa_{z_2}\pa_{w_2}^2A)\\
=&\int_{v_2}\cP\wedge\frac{\bar w_2(\bar z_{23}d\bar w_{2}dt_{2}-\bar w_{23}d\bar z_{2}dt_{2}+t_{23}d\bar z_{2}d\bar w_{2})}{d_{23}^7}\wedge dw_2\wedge dz_2
\fe

Now, compute the integrand:
\ie\label{crucialsttep2app}
&\frac{\bar z_2(\bar w_2dt_2-t_2d\bar w_2)\bar w_2(\bar z_{23}d\bar w_{2}dt_{2}-\bar w_{23}d\bar z_{2}dt_{2}+t_{23}d\bar z_{2}d\bar w_{2})}{d^5_{02}d^7_{23}}\wedge dw_2\wedge dz_2\\
=&\frac{\rvert z_2\rvert^2\rvert w_2\rvert^4(t_2-t_3-t_2)}{d^5_{02}d^7_{23}}dt_2d\bar z_2d\bar w_2 dw_2 dz_2\\
=&-\frac{\rvert z_2\rvert^2\rvert w_2\rvert^4t_3}{d^5_{02}d^7_{23}}dt_2d\bar z_2d\bar w_2 dw_2 dz_2\quad\text{substitute }t_3=\E,~\text{then,}\\
=&-\frac{\rvert z_2\rvert^2\rvert w_2\rvert^4\E}{d^5_{02}d^7_{23}}dt_2d\bar z_2d\bar w_2 dw_2 dz_2
\fe
We can rescale $\E\ria1$, without loss of generality, then it becomes
\ie
-\frac{\rvert z_2\rvert^2\rvert w_2\rvert^4}{d^5_{02}d^7_{23}}dt_2d\bar z_2d\bar w_2 dw_2 dz_2
\fe
\newline
{\bf{Lemma 4.}}
\newline
Now, it remains to evaluate the delta function at the third vertex. In other words, substitute
\ie
w_3\ria0,\quad z_3\ria0,\quad t_3\ria\E=1
\fe
and use Feynman technique to convert the above integral into
\ie
&-\frac{\Gamma(6)}{\Gamma(5/2)\Gamma(7/2)}\int^1_0dx\int_{v_2}\frac{\sqrt{x^3(1-x)^5}\rvert z_2\rvert^2\rvert w_2\rvert^4}{(x(z_{2}^2+w_{2}^2+(t_2-1)^2)+(1-x)(z_2^2+w_2^2+t_2^2))^6}\\
=&-\frac{\Gamma(6)}{\Gamma(5/2)\Gamma(7/2)}\int^1_0dx\int_{v_2}\frac{\sqrt{x^3(1-x)^5}\rvert z_2\rvert^2\rvert w_2\rvert^4}{(z_2^2+w_2^2+(t_2-x)^2+x(1-x))^6}\\
=&-\frac{\Gamma(6)}{\Gamma(5/2)\Gamma(7/2)}\int^1_0dx\int_{v_2}\frac{\sqrt{x^3(1-x)^5}\rvert z_2\rvert^2\rvert w_2\rvert^4}{(z_2^2+w_2^2+t_2^2+x(1-x))^6}
\fe
In the second equality, we shift $t_2$ to $t_2+x$.

After doing $v_2$ integral, it reduces to
\ie
\frac{\Gamma(6)}{\Gamma(5/2)\Gamma(7/2)}\frac{\pi}{2880}\int^1_0dxx(1-x)^2=\frac{\Gamma(6)}{\Gamma(5/2)\Gamma(7/2)}\frac{\pi}{2880}
\fe
Finally, re-introduce all omitted constants:
\ie
(First Term)=\frac{\Gamma(6)}{\Gamma(5/2)\Gamma(7/2)}\frac{\Gamma(7)}{\Gamma(5/2)\Gamma(9/2)}(2\pi)^2(2\pi)^2\frac{\pi}{2880}
\fe
Similarly, we can compute all the others without any divergence.
\ie\label{secthirfour}
\text{(Second Term)}&=\frac{\Gamma(6)}{\Gamma(5/2)\Gamma(7/2)}\frac{\Gamma(7)}{\Gamma(5/2)\Gamma(9/2)}(2\pi)^2(2\pi)^2\frac{\pi}{5760}\\
\text{(Third Term)}&=\frac{\Gamma(6)}{\Gamma(5/2)\Gamma(7/2)}\frac{\Gamma(7)}{\Gamma(5/2)\Gamma(9/2)}(2\pi)^2(2\pi)^2\frac{\pi}{8640}\\
\text{(Fourth Term)}&=\frac{\Gamma(6)}{\Gamma(5/2)\Gamma(7/2)}\frac{\Gamma(7)}{\Gamma(5/2)\Gamma(9/2)}(2\pi)^2(2\pi)^2\frac{\pi}{20160}
\fe
Hence, every terms in \eqref{appequ} are integrated into finite terms.
\subsection{Intermediate steps in section 5.2}\label{app:2-2}
{\bf{Lemma 5.}}
\newline
We want to evaluate the following integral.
\ie\label{v1integralabp}
\int_{v_1}\pa_{z_1}P_1(v_0,v_1)\wedge(w_1dw_1)\wedge(z^2_1dz_1)\wedge\pa_{w_1} P_2(v_1,v_2)
\fe

Substituting the expressions for the propagators, we get
\ie
\int_{v_1}\frac{\rvert z_1\rvert^2z_1w_1(\bar w_1-\bar w_2)}{d^7_{01}d^7_{12}}&(\bar z_{01}\bar w_{12}dt_2-\bar z_{01}t_{12}d\bar w_2+\bar w_{01}t_{12}d\bar z_2-\bar w_{01}\bar z_{12}dt_2\\
&+t_{01}\bar z_{12}d\bar w_2-t_{01}\bar w_{12}d\bar z_{2})d\bar z_1d\bar w_1dt_1dz_1dw_1
\fe
We already know that the terms proportional to $\bar{w}_2$ will vanish in the second vertex integral, so drop them. Evaluating the delta function at $v_0$, the above simplifies to
\ie
\int_{v_1}\frac{\rvert z_1\rvert^2z_1\rvert w_1\rvert^2}{d^7_{01}d^7_{12}}(&-\bar z_{1}\bar w_{12}dt_2+\bar z_{1}t_{12}d\bar w_2-\bar w_{1}t_{12}d\bar z_2+\bar w_{1}\bar z_{12}dt_2\\
&-t_1\bar z_{12}d\bar w_2+t_{1}\bar w_{12}d\bar z_{2})d\bar z_1d\bar w_1dt_1dz_1dw_1
\fe
Note that the integrand with the odd number of $t_1$ vanishes, so 
\ie
\int_{v_1}\frac{\rvert z_1\rvert^2z_1\rvert w_1\rvert^2}{d^7_{01}d^7_{12}}(-\bar z_{1}\bar w_{12}dt_2-\bar z_{1}t_{2}d\bar w_2+\bar w_{1}t_{2}d\bar z_2+\bar w_{1}\bar z_{12}dt_2)d\bar z_1d\bar w_1dt_1dz_1dw_1
\fe
Now, apply Feynman technique and omit the Gamma functions, to be recovered at the end.
\ie
&\int^1_0dx\sqrt{x(1-x)}^7\int_{v_1}\frac{\rvert z_1\rvert^2\rvert w_1\rvert^2z_1(-\bar z_{1}\bar w_{12}dt_2-\bar z_{1}t_{2}d\bar w_2+\bar w_{1}t_{2}d\bar z_2+\bar w_{1}\bar z_{12}dt_2)}{(x(\rvert z_1\rvert^2+\rvert w_1\rvert^2+\rvert t_1\rvert^2)+(1-x)(\rvert z_{12}\rvert^2+\rvert w_{12}\rvert^2+\rvert t_{12}\rvert^2))^7}\\
=&\int^1_0dx\sqrt{x(1-x)}^7\int_{v_1}\frac{\rvert z_1\rvert^2\rvert w_1\rvert^2z_1(-\bar z_{1}\bar w_{12}dt_2-\bar z_{1}t_{2}d\bar w_2+\bar w_{1}t_{2}d\bar z_2+\bar w_{1}\bar z_{12}dt_2)}{(\rvert z_1-xz_2\rvert^2+\rvert w_1-xw_2\rvert^2+(t_1-xt_2)^2+x(1-x)(\rvert z_2\rvert^2+\rvert w_2\rvert^2+t^2_2))^7}
\fe
Shift the integral variables as
\ie
z_1\ria z_1+xz_2,\quad w_1\ria w_1+xw_2,\quad t_1\ria t_1+xt_2
\fe
Then the above becomes
\ie
\int^1_0dx\sqrt{x(1-x)}^7&\int_{v_1}dz_1d\bar z_1dw_1d\bar w_1dt_1(\rvert z_1\rvert^2+x^2\rvert z_2\rvert^2)(\rvert w_1\rvert^2+x^2\rvert w_2\rvert^2)(z_1+x z_2)\\
&\bigg(\frac{-(\bar z_{1}+x\bar z_2)(\bar w_{1}+(x-1)\bar w_2)dt_2-(\bar z_{1}+x\bar z_2)t_{2}d\bar w_2}{(\rvert z_1\rvert^2+\rvert w_1\rvert^2+t_1^2+x(1-x)(\rvert z_2\rvert^2+\rvert w_2\rvert^2+t^2_2))^7}\\
&+\frac{(\bar w_{1}+x\bar w_2)t_{2}d\bar z_2+(\bar w_{1}+x\bar w_2)(\bar z_{1}+(x-1)\bar z_2)dt_2}{(\rvert z_1\rvert^2+\rvert w_1\rvert^2+t_1^2+x(1-x)(\rvert z_2\rvert^2+\rvert w_2\rvert^2+t^2_2))^7}\bigg)
\fe
The terms with (anti)holomorphic dependence on complex coordinates drop:
\ie
\int^1_0dx\sqrt{x(1-x)}^7&\int_{v_1}dz_1d\bar z_1dw_1d\bar w_1dt_1(\rvert z_1\rvert^2+x^2\rvert z_2\rvert^2)(\rvert w_1\rvert^2+x^2\rvert w_2\rvert^2)\\
&\bigg(\frac{-\rvert z_1\rvert^2t_2d\bar w_2+x\rvert z_1\rvert^2\bar w_2dt_2-x^2\rvert z_2\rvert^2(x-1)\bar w_2dt_2}{(\rvert z_1\rvert^2+\rvert w_1\rvert^2+t_1^2+x(1-x)(\rvert z_2\rvert^2+\rvert w_2\rvert^2+t^2_2))^7}\\
&+\frac{-x^2\rvert z_2\rvert^2t_2d\bar w_2+x^2z_2\bar w_2t_2d\bar z_2+x^2\rvert z_2\rvert^2\bar w_2(x-1)dt_2}{(\rvert z_1\rvert^2+\rvert w_1\rvert^2+t_1^2+x(1-x)(\rvert z_2\rvert^2+\rvert w_2\rvert^2+t^2_2))^7}\bigg)
\fe
We can be prescient again; using the fact that the second vertex is tagged with a delta function $\D(z_2=0,t_2=\E)\propto dz_2d\bar z_2dt_2$, we can drop most of the terms.
\ie\label{crucialsttep3}
& -\int^1_0dx\sqrt{x(1-x)}^7\int_{v_1}[dV_1]\frac{(\rvert z_1\rvert^2+x^2\rvert z_2\rvert^2)(\rvert w_1\rvert^2+x^2\rvert w_2\rvert^2)(-\rvert z_1\rvert^2-x^2\rvert z_2\rvert^2)t_2d\bar w_2}{(\rvert z_1\rvert^2+\rvert w_1\rvert^2+t_1^2+x(1-x)(\rvert z_2\rvert^2+\rvert w_2\rvert^2+t^2_2))^7}\\
&=-\int^1_0dx\sqrt{x(1-x)}^7\int_{v_1}[dV_1]\frac{(\rvert z_1\rvert^2+x^2\rvert z_2\rvert^2)^2(\rvert w_1\rvert^2+x^2\rvert w_2\rvert^2)t_2d\bar w_2}{(\rvert z_1\rvert^2+\rvert w_1\rvert^2+t_1^2+x(1-x)(\rvert z_2\rvert^2+\rvert w_2\rvert^2+t^2_2))^7}
\fe
where $[dV_1]$ is an integral measure for $v_1$ integral. 
\subsection{Intermediate steps in section 5.3}\label{app:2-3}
{\bf{Lemma 6.}}
\newline
We will evaluate the following integral.
\ie\label{v1integral-1app}
\int_{v_1}\frac{1}{w_1}(w_1dw_1)\D(t_1=0,z_1=0)\wedge\pa_{z_2} P_{12}(v_1,v_2)
\fe

Substituting the expressions for propagators, we get
\ie
&\int_{v_1}\frac{\bar z_1-\bar z_2}{d^7_{12}}(\bar z_{12}d\bar w_{12}dt_{12}-\bar w_{12}d\bar z_{12}dt_{12}+t_{12}d\bar z_{12}d\bar w_{12})dw_1\D(t_1=z_1=0)\\
=&\int_{v_1}\frac{\bar z_1-\bar z_2}{d^7_{12}}(\bar z_2d\bar w_1dt_2+t_{2}d\bar z_{2}d\bar w_1)dw_1\D(t_1=z_1=0)\\
=&~(t_2d\bar z_2+\bar z_2dt_2)\int_{v_1}\frac{\bar z_1-\bar z_2}{\sqrt{t_{12}^2+\rvert z_{12}\rvert^2+\rvert w_{12}\rvert ^2}^7}d\bar w_1dw_1\D(t_1=z_1=0)\\
=&~(t_2d\bar z_2+\bar z_2dt_2)\int dw_1d\bar w_1\frac{-\bar z_2}{\sqrt{t_2^2+\rvert z_2\rvert^2+\rvert w_1-w_2\rvert^2}^7}\\
=&~-(t_2d\bar z_2+\bar z_2dt_2)\int rdrd\theta\frac{\bar z_2}{\sqrt{t_2^2+\rvert z_2\rvert^2+r^2}^7}=-\frac{2\pi (t_2d\bar z_2+\bar z_2dt_2)\bar z_2}{5\sqrt{t_2^2+\rvert z_2\rvert^2}^5}
\fe
where the first equality comes from the fact that $\D(t_1=z_1=0)\propto dt_1dz_1d\bar z_1$.
\newline
{\bf{Lemma 7.}}
\newline
We will evaluate the following integral.
\ie\label{v3integral-1app}
\int_{v_3}\frac{1}{w_3}(dw_3)\D(t_3=0,z_3=0)\wedge\pa_{w_2} P(v_2,v_3)
\fe
Substituting the expressions for propagators, we get
\ie
&\int_{v_3}\frac{\bar w_2-\bar w_3}{w_3d^7_{23}}(\bar z_{23}d\bar w_{23}dt_{23}-\bar w_{23}d\bar z_{23}dt_{23}+t_{23}d\bar z_{23}d\bar w_{23})dw_3\D(t_3=z_3=0)\\
=&\int_{v_3}\frac{\bar w_2-\bar w_3}{w_3d^7_{23}}(-\bar z_{2}d\bar w_3dt_2+t_{2}d\bar z_{2}d\bar w_3)dw_3\D(t_3=z_3=0)\\
=&(t_2d\bar z_2-\bar z_2dt_2)\int_{v_3}\frac{\bar w_2-\bar w_3}{w_3\sqrt{t_{23}^2+\rvert z_{23}\rvert^2+\rvert w_{23}\rvert ^2}^7}d\bar w_3dw_3\D(t_3=z_3=0)\\
=&(t_2d\bar z_2-\bar z_2dt_2)\int dw_3d\bar w_3\frac{(\bar w_2-\bar w_3)/w_3}{\sqrt{t_2^2+\rvert z_2\rvert^2+\rvert w_2-w_3\rvert^2}^7}\\
=&(t_2d\bar z_2-\bar z_2dt_2)\int dw_3d\bar w_3\frac{-\bar w_3/(w_3+w_2)}{\sqrt{t_2^2+\rvert z_2\rvert^2+\rvert w_3\rvert^2}^7}\\
=&(t_2d\bar z_2-\bar z_2dt_2)\int _{|w_3|\le |w_2|} dw_3d\bar w_3\frac{-\bar w_3\left(1-\frac{w_3}{w_2}+\frac{1}{2!}\frac{w_3^2}{w_2^2}-\ldots\right)}{w_2\sqrt{t_2^2+\rvert z_2\rvert^2+\rvert w_3\rvert^2}^7}\\
&+(t_2d\bar z_2-\bar z_2dt_2)\int _{|w_3|\ge |w_2|} dw_3d\bar w_3\frac{-\bar w_3\left(1-\frac{w_2}{w_3}+\frac{1}{2!}\frac{w_2^2}{w_3^2}-\ldots\right)}{w_3\sqrt{t_2^2+\rvert z_2\rvert^2+\rvert w_3\rvert^2}^7}\\
=&(t_2d\bar z_2-\bar z_2dt_2)\int  _{|w_3|\le |w_2|}dw_3d\bar w_3\left(0+\frac{-\rvert w_3\rvert^2}{w^2_2\sqrt{t_2^2+\rvert z_2\rvert^2+\rvert w_3\rvert^2}^7}+0+0+\ldots\right)\\
=&(t_2d\bar z_2-\bar z_2dt_2)\int _{0}^{|w_2|} rdrd\theta\frac{-r^2}{w^2_2\sqrt{t_2^2+\rvert z_2\rvert^2+r^2}^7}\\
=&-(t_2d\bar z_2-\bar z_2dt_2)\frac{2\pi}{15w^2_2}\left(\frac{2}{\sqrt{t_2^2+\rvert z_2\rvert^2}^3}-\frac{5\rvert w_2\rvert^2+2t_2^2+2\rvert z_2\rvert^2}{\sqrt{t_2^2+\rvert z_2\rvert^2+\rvert w_2\rvert^2}^5}\right)
\fe
\newline
{\bf{Lemma 8.}}
\newline
 We will evaluate 
\ie\label{secondvertexap}
\int_{v_2}dw_2\wedge dz_2\wedge d\bar z_2\wedge dt_2\frac{4\pi^2t_2 \rvert z_2\rvert^2}{75w_2\sqrt{t_2^2+\rvert z_2\rvert^2}^5}\left(\frac{2}{\sqrt{t_2^2+\rvert z_2\rvert^2}^3}-\frac{5\rvert w_2\rvert^2+2t_2^2+2\rvert z_2\rvert^2}{\sqrt{t_2^2+\rvert z_2\rvert^2+\rvert w_2\rvert^2}^5}\right).
\fe
Assuming the $w_2$ integral domain is a contour surrounding the origin of $w_2$ plane or a path that can be deformed into the contour, we may use the residue theorem for the first term of \eqref{secondvertexap}. After doing $w_2$ integral we have
\ie
\int_{\E}^\infty dt_2\int_{\bC_{z_2}}d^2z_2\frac{4\pi^2t_2 \rvert z_2\rvert^2}{75\sqrt{t_2^2+\rvert z_2\rvert^2}^5}\frac{2}{\sqrt{t_2^2+\rvert z_2\rvert^2}^3}=\frac{2\pi^3}{225\E^2}
\fe
Combining with the other diagram with the second vertex in the $t\in[-\infty,-\E
]$, we get
\ie\label{appboundd}
\frac{2\pi^3}{225\E^2}-\bigg(-\frac{2\pi^3}{225\E^2}\bigg)=\frac{4\pi^3}{225\E^2}
\fe
Re-scaling $\E\ria1$, this is finite.

For the second term of \eqref{secondvertexap}, let us choose the contour to be a constant radius circle so that $r(\theta)=R$. We need to use an unconventional version of the residue theorem, as the integrand is not a holomorphic function, depending on $\rvert w_2\rvert^2$. Let $w_2=Re^{i\theta}$, then for a given integrand $f(w_2,\bar w_2)$, we have
\ie
I=\int^{2\pi}_0d(Re^{i\theta})f(Re^{i\theta},Re^{-i\theta})
\fe
Then, $w_2$ integral is evaluated as
\ie
-\int^{2\pi}_0 \frac{d(Re^{i\theta})}{Re^{i\theta}}\frac{4\pi^2t_2 \rvert z_2\rvert^2}{75\sqrt{t_2^2+\rvert z_2\rvert^2}^5}\frac{5R^2+2t_2^2+2\rvert z_2\rvert^2}{\sqrt{t_2^2+\rvert z_2\rvert^2+R^2}^5} =-\frac{8\pi^3it_2 \rvert z_2\rvert^2}{75\sqrt{t_2^2+\rvert z_2\rvert^2}^5}\frac{5R^2+2t_2^2+2\rvert z_2\rvert^2}{\sqrt{t_2^2+\rvert z_2\rvert^2+R^2}^5}
\fe
Before evaluating $z_2$ integral, it is better to work without $R$. using the following inequality is useful to facilitate an easier integral:
\ie
0<\frac{8\pi^3it_2 \rvert z_2\rvert^2}{75\sqrt{t_2^2+\rvert z_2\rvert^2}^5}\left(\frac{5R^2+2t_2^2+2\rvert z_2\rvert^2}{\sqrt{t_2^2+\rvert z_2\rvert^2+R^2}^5}\right)<\frac{(8\pi^3it_2 \rvert z_2\rvert^2)(2t_2^2+2\rvert z_2\rvert^2)}{75(t_2^2+\rvert z_2\rvert^2)^5}
\fe
Here we used $R\in Real^+$. The left bound is obtained by $R\ria\infty$, and the right bound is obtained by $R\ria0$. We only care the convergence of the integral. Therefore, let us proceed with the inequalities.
\ie
-\frac{4\pi}{192}\frac{8\pi^3i}{75}\frac{1}{\E^3}<-\int^\infty_\E dt_2\int_{\bC_{z_2}}d^2z_2\frac{8\pi^3it_2 \rvert z_2\rvert^2}{75\sqrt{t_2^2+\rvert z_2\rvert^2}^5}\left(\frac{5R^2+2t_2^2+2\rvert z_2\rvert^2}{\sqrt{t_2^2+\rvert z_2\rvert^2+R^2}^5}\right)<0
\fe
After rescaling $\E\ria1$, we have a finite answer. Combining with the other diagram with the second vertex in the $t\in[-\infty,-\E
]$, we get the left bound as
\ie
-\frac{4\pi}{192}\frac{8\pi^3i}{75}-\bigg(\frac{4\pi}{192}\frac{8\pi^3i}{75}\bigg)=-\frac{\pi^4i}{225\E^3}
\fe
After rescaling $\E_1\ria1$, this is also finite.

Hence, combining with \eqref{appboundd} we get the bound
\ie
\frac{4\pi^3}{225\E^2}-\frac{\pi^4i}{225\E^3}< \eqref{secondvertexap}<\frac{4\pi^3}{225\E^2}
\fe
\end{appendix}

\bibliography{SciPost_Example_BiBTeX_File.bib}

\providecommand{\href}[2]{#2}\begingroup\raggedright\begin{thebibliography}{10}
\bibitem{Costello:2016mgj} 
  K.~Costello and S.~Li,
  ``Twisted supergravity and its quantization,''
  arXiv:1606.00365 [hep-th].
\bibitem{Witten:1988ze} 
  E.~Witten,
  ``Topological Quantum Field Theory,''
  Commun.\ Math.\ Phys.\  \textbf{117}, 353 (1988)
doi:10.1007/BF01223371
\bibitem{Witten:1988xj} 
  E.~Witten,
  ``Topological Sigma Models,''
  Commun.\ Math.\ Phys.\  \textbf{118}, 411 (1988)
doi:10.1007/BF01466725
\bibitem{Costello:2016nkh} 
  K.~Costello,
  ``M-theory in the Omega-background and 5-dimensional non-commutative gauge theory,''
  arXiv:1610.04144 [hep-th].

\bibitem{Nekrasov:2002qd} 
  N.~A.~Nekrasov,
  ``Seiberg-Witten prepotential from instanton counting,''
  Adv.\ Theor.\ Math.\ Phys.\  {\bf 7}, no. 5, 831 (2003)
doi:10.4310/ATMP.2003.v7.n5.a4
[arXiv:hep-th/0206161 [hep-th]].
\bibitem{Alday:2009aq} 
  L.~F.~Alday, D.~Gaiotto and Y.~Tachikawa,
  ``Liouville Correlation Functions from Four-dimensional Gauge Theories,''
  Lett.\ Math.\ Phys.\  {\bf 91}, 167 (2010)
doi:10.1007/s11005-010-0369-5
[arXiv:0906.3219 [hep-th]].
\bibitem{Nekrasov:2009rc} 
  N.~A.~Nekrasov and S.~L.~Shatashvili,
  ``Quantization of Integrable Systems and Four Dimensional Gauge Theories,''
 doi:10.1142/9789814304634\_0015
[arXiv:0908.4052 [hep-th]].
\bibitem{Nekrasov:2010ka} 
  N.~Nekrasov and E.~Witten,
  ``The Omega Deformation, Branes, Integrability, and Liouville Theory,''
  JHEP {\bf 1009}, 092 (2010)
doi:10.1007/JHEP09(2010)092
[arXiv:1002.0888 [hep-th]].
\bibitem{Yagi:2014toa} 
  J.~Yagi,
  ``$\Omega$-deformation and quantization,''
  JHEP {\bf 1408}, 112 (2014)
doi:10.1007/JHEP08(2014)112
[arXiv:1405.6714 [hep-th]].
\bibitem{Nekrasov:2015wsu} 
  N.~Nekrasov,
  ``BPS/CFT correspondence: non-perturbative Dyson-Schwinger equations and qq-characters,''
  JHEP {\bf 1603}, 181(2016)
doi:10.1007/JHEP03(2016)181
[arXiv:1512.05388 [hep-th]].

\bibitem{Oh:2019bgz} 
  J.~Oh and J.~Yagi,
  ``Chiral algebras from $\Omega$-deformation,''
  JHEP {\bf 1908}, 143 (2019)
doi:10.1007/JHEP08(2019)143
[arXiv:1903.11123 [hep-th]].
\bibitem{Jeong:2019pzg} 
  S.~Jeong,
  ``SCFT/VOA correspondence via $\Omega$-deformation,''
  JHEP {\bf 1910}, 171 (2019)
doi:10.1007/JHEP10(2019)171
[arXiv:1904.00927 [hep-th]].
\bibitem{Beem:2018fng}
C.~Beem, D.~Ben-Zvi, M.~Bullimore, T.~Dimofte and A.~Neitzke,
``Secondary products in supersymmetric field theory,''
Annales Henri Poincare \textbf{21}, no.4, 1235-1310 (2020)
doi:10.1007/s00023-020-00888-3
[arXiv:1809.00009 [hep-th]].
\bibitem{Oh:2019mcg} 
  J.~Oh and J.~Yagi,
  ``Poisson vertex algebras in supersymmetric field theories,''
doi:10.1007/s11005-020-01290-0
[arXiv:1908.05791 [hep-th]].

\bibitem{Costello:2018zrm} 
  K.~Costello and D.~Gaiotto,
  ``Twisted Holography,''
  arXiv:1812.09257 [hep-th].
  
\bibitem{Costello:2012cy} 
  K.~J.~Costello and S.~Li,
  ``Quantum BCOV theory on Calabi-Yau manifolds and the higher genus B-model,''
  arXiv:1201.4501 [math.QA].
\bibitem{Costello:2015xsa} 
  K.~Costello and S.~Li,
  ``Quantization of open-closed BCOV theory, I,''
  arXiv:1505.06703 [hep-th].


\bibitem{Costello:2017fbo} 
  K.~Costello,
  ``Holography and Koszul duality: the example of the $M2$ brane,''
  arXiv:1705.02500 [hep-th].

\bibitem{Gaiotto:2019wcc} 
  D.~Gaiotto and J.~Oh,
  ``Aspects of $\Omega$-deformed M-theory,''
  arXiv:1907.06495 [hep-th].

\bibitem{Bullimore:2015lsa} 
  M.~Bullimore, T.~Dimofte and D.~Gaiotto,
  ``The Coulomb Branch of 3d ${\mathcal{N}= 4}$ Theories,''
  Commun.\ Math.\ Phys.\  {\bf 354}, no. 2, 671 (2017)
doi:10.1007/s00220-017-2903-0
[arXiv:1503.04817 [hep-th]].
\bibitem{Braverman:2016wma} 
  A.~Braverman, M.~Finkelberg and H.~Nakajima,
  ``Towards a mathematical definition of Coulomb branches of $3$-dimensional $\mathcal{N} = 4$ gauge theories, II,''
  Adv.\ Theor.\ Math.\ Phys.\  {\bf 22}, 1071 (2018)
doi:10.4310/ATMP.2018.v22.n5.a1
[arXiv:1601.03586 [math.RT]].

\bibitem{Bullimore:2016nji} 
  M.~Bullimore, T.~Dimofte, D.~Gaiotto and J.~Hilburn,
  ``Boundaries, Mirror Symmetry, and Symplectic Duality in 3d $\mathcal{N}=4$ Gauge Theory,''
  JHEP {\bf 1610}, 108 (2016)
doi:10.1007/JHEP10(2016)108
[arXiv:1603.08382 [hep-th]].
\bibitem{Bullimore:2016hdc}
M.~Bullimore, T.~Dimofte, D.~Gaiotto, J.~Hilburn and H.~C.~Kim,
``Vortices and Vermas,''
Adv. Theor. Math. Phys. \textbf{22}, 803-917 (2018)
doi:10.4310/ATMP.2018.v22.n4.a1
[arXiv:1609.04406 [hep-th]].

\bibitem{Tsymbaliuk:2014fvq} 
  A.~Tsymbaliuk,
  ``The affine Yangian of $\mathfrak{gl}_1$ revisited,''
  Adv.\ Math.\  {\bf 304}, 583 (2017)
doi:10.1016/j.aim.2016.08.041
[arXiv:1404.5240 [math.RT]].
\bibitem{Prochazka:2015deb} 
  T.~Prochazka,
  ``$ \mathcal{W} $ -symmetry, topological vertex and affine Yangian,''
  JHEP {\bf 1610}, 077 (2016)
doi:10.1007/JHEP10(2016)077
[arXiv:1512.07178 [hep-th]].
  \bibitem{Kodera:2016faj} 
  R.~Kodera and H.~Nakajima,
  ``Quantized Coulomb branches of Jordan quiver gauge theories and cyclotomic rational Cherednik algebras,''
  Proc.\ Symp.\ Pure Math.\  {\bf 98}, 49 (2018)
[arXiv:1608.00875 [math.RT]].
\bibitem{Gaberdiel:2017dbk} 
  M.~R.~Gaberdiel, R.~Gopakumar, W.~Li and C.~Peng,
  ``Higher Spins and Yangian Symmetries,''
  JHEP {\bf 1704}, 152 (2017)
doi:10.1007/JHEP04(2017)152
[arXiv:1702.05100 [hep-th]].

\bibitem{Gaiotto:2017euk} 
  D.~Gaiotto and M.~Rapcak,
  ``Vertex Algebras at the Corner,''
  JHEP {\bf 1901}, 160 (2019)
doi:10.1007/JHEP01(2019)160
[arXiv:1703.00982 [hep-th]].
  
\bibitem{Gaiotto:2020vqj}
D.~Gaiotto and J.~Abajian,
``Twisted M2 brane holography and sphere correlation functions,''
[arXiv:2004.13810 [hep-th]].

\bibitem{Costello:2017dso} 
  K.~Costello, E.~Witten and M.~Yamazaki,
  ``Gauge Theory and Integrability, I,''
doi:10.4310/ICCM.2018.v6.n1.a6
[arXiv:1709.09993 [hep-th]].

\bibitem{Maldacena:1997re} 
  J.~M.~Maldacena,
  ``The Large N limit of superconformal field theories and supergravity,''
  Int.\ J.\ Theor.\ Phys.\  {\bf 38}, 1113-1133 (1999)
doi:10.1023/A:1026654312961
[arXiv:hep-th/9711200 [hep-th]].

\bibitem{Gubser:1998bc} 
  S.~S.~Gubser, I.~R.~Klebanov and A.~M.~Polyakov,
  ``Gauge theory correlators from noncritical string theory,''
  Phys.\ Lett.\ B {\bf 428}, 105 (1998)
doi:10.1016/S0370-2693(98)00377-3
[arXiv:hep-th/9802109 [hep-th]].
\bibitem{Witten:1998qj} 
  E.~Witten,
  ``Anti-de Sitter space and holography,''
  Adv.\ Theor.\ Math.\ Phys.\  {\bf 2}, 253 (1998)
doi:10.4310/ATMP.1998.v2.n2.a2
[arXiv:hep-th/9802150 [hep-th]].

\bibitem{Bonetti:2016nma} 
  F.~Bonetti and L.~Rastelli,
  ``Supersymmetric localization in AdS$_{5}$ and the protected chiral algebra,''
  JHEP {\bf 1808}, 098 (2018)
doi:10.1007/JHEP08(2018)098
[arXiv:1612.06514 [hep-th]].

\bibitem{Mezei:2017kmw} 
  M.~Mezei, S.~S.~Pufu and Y.~Wang,
  ``A 2d/1d Holographic Duality,''
  arXiv:1703.08749 [hep-th].
\bibitem{Beem:2013sza} 
  C.~Beem, M.~Lemos, P.~Liendo, W.~Peelaers, L.~Rastelli and B.~C.~van Rees,
  ``Infinite Chiral Symmetry in Four Dimensions,''
  Commun.\ Math.\ Phys.\  {\bf 336}, no. 3, 1359-1433 (2015)
doi:10.1007/s00220-014-2272-x
[arXiv:1312.5344 [hep-th]].
  
\bibitem{Gaiotto:2019mmf} 
  D.~Gaiotto and T.~Okazaki,
  ``Sphere correlation functions and Verma modules,''
JHEP \textbf{02}, 133 (2020)
doi:10.1007/JHEP02(2020)133
[arXiv:1911.11126 [hep-th]].
  
\bibitem{Ishtiaque:2018str} 
  N.~Ishtiaque, S.~Faroogh Moosavian and Y.~Zhou,
 ``Topological Holography: The Example of The D2-D4 Brane System,''
SciPost Phys. \textbf{9}, no.2, 017 (2020)
doi:10.21468/SciPostPhys.9.2.017
[arXiv:1809.00372 [hep-th]].
\bibitem{Costello:2020jbh} 
  K.~Costello and N.~M.~Paquette,
  ``Twisted Supergravity and Koszul Duality: A case study in AdS$_3$,''
  arXiv:2001.02177 [hep-th].
\bibitem{Rozansky:1996bq} 
  L.~Rozansky and E.~Witten,
  ``HyperKahler geometry and invariants of three manifolds,''
  Selecta Math.\  {\bf 3}, 401 (1997)
doi:10.1007/s000290050016
[arXiv:hep-th/9612216 [hep-th]].

\bibitem{Nekrasov:1996} 
  N. ~Nekrasov, 
  ``Four-dimensional holomorphic theories". ProQuest LLC, Ann Arbor, MI, (1996),
p. 174, Thesis (Ph.D.)?Princeton University, ISBN: 978-0591-07477-2. 
\bibitem{Johansen:1995}
A.~Johansen,
``Twisting of $N=1$ SUSY gauge theories and heterotic topological theories,''
Int. J. Mod. Phys. A \textbf{10}, 4325-4358 (1995)
doi:10.1142/S0217751X9500200X
[arXiv:hep-th/9403017 [hep-th]].

\bibitem{Costello:2011np} 
  K.~J.~Costello,
  ``Notes on supersymmetric and holomorphic field theories in dimensions 2 and 4,''
\bibitem{Saberi:2019ghy} 
  I.~Saberi and B.~R.~Williams,
  ``Twisted characters and holomorphic symmetries,''
  arXiv:1906.04221 [math-ph].
\bibitem{Costello:2013zra}
K.~Costello,
``Supersymmetric gauge theory and the Yangian,''
[arXiv:1303.2632 [hep-th]].
\bibitem{Bershadsky:1993cx} 
  M.~Bershadsky, S.~Cecotti, H.~Ooguri and C.~Vafa,
  ``Kodaira-Spencer theory of gravity and exact results for quantum string amplitudes,''
  Commun.\ Math.\ Phys.\  {\bf 165}, 311 (1994)
doi:10.1007/BF02099774
[arXiv:hep-th/9309140 [hep-th]].

\bibitem{Kapustin:2006hi} 
  A.~Kapustin,
  ``Holomorphic reduction of N=2 gauge theories, Wilson-'t Hooft operators, and S-duality,''
  hep-th/0612119.
\bibitem{Costello:2018txb} 
  K.~Costello and J.~Yagi,
  ``Unification of integrability in supersymmetric gauge theories,''
  arXiv:1810.01970 [hep-th].

\bibitem{Dedushenko:2016jxl} 
  M.~Dedushenko, S.~S.~Pufu and R.~Yacoby,
``A one-dimensional theory for Higgs branch operators,''
JHEP \textbf{03}, 138 (2018)
doi:10.1007/JHEP03(2018)138
[arXiv:1610.00740 [hep-th]].

\bibitem{Beem:2016cbd} 
  C.~Beem, W.~Peelaers and L.~Rastelli,
  ``Deformation quantization and superconformal symmetry in three dimensions,''
  Commun.\ Math.\ Phys.\  {\bf 354}, no. 1, 345 (2017)
doi:10.1007/s00220-017-2845-6
[arXiv:1601.05378 [hep-th]].
\bibitem{Gaiotto:2015zna} 
  D.~Gaiotto, G.~W.~Moore and E.~Witten,
  ``An Introduction To The Web-Based Formalism,''
  arXiv:1506.04086 [hep-th].
\bibitem{Gaiotto:2015aoa} 
  D.~Gaiotto, G.~W.~Moore and E.~Witten,
  ``Algebra of the Infrared: String Field Theoretic Structures in Massive ${\cal N}=(2,2)$ Field Theory In Two Dimensions,''
  arXiv:1506.04087 [hep-th].
  
\bibitem{Kapustin:2006pk} 
  A.~Kapustin and E.~Witten,
  ``Electric-Magnetic Duality And The Geometric Langlands Program,''
  Commun.\ Num.\ Theor.\ Phys.\  {\bf 1}, 1 (2007)
doi:10.4310/CNTP.2007.v1.n1.a1
[arXiv:hep-th/0604151 [hep-th]].
  
\bibitem{Gaiotto:2008sa} 
  D.~Gaiotto and E.~Witten,
  ``Supersymmetric Boundary Conditions in N=4 Super Yang-Mills Theory,''
  J.\ Statist.\ Phys.\  {\bf 135}, 789 (2009)
doi:10.1007/s10955-009-9687-3
[arXiv:0804.2902 [hep-th]].

\bibitem{Mikhaylov:2014aoa} 
  V.~Mikhaylov and E.~Witten,
  ``Branes And Supergroups,''
  Commun.\ Math.\ Phys.\  {\bf 340}, no. 2, 699 (2015)
doi:10.1007/s00220-015-2449-y
[arXiv:1410.1175 [hep-th]].

\bibitem{Intriligator:1996ex} 
  K.~A.~Intriligator and N.~Seiberg,
  ``Mirror symmetry in three-dimensional gauge theories,''
  Phys.\ Lett.\ B {\bf 387}, 513 (1996)
doi:10.1016/0370-2693(96)01088-X
[arXiv:hep-th/9607207 [hep-th]].
\bibitem{deBoer:1996mp} 
  J.~de Boer, K.~Hori, H.~Ooguri and Y.~Oz,
  ``Mirror symmetry in three-dimensional gauge theories, quivers and D-branes,''
  Nucl.\ Phys.\ B {\bf 493}, 101 (1997)
doi:10.1016/S0550-3213(97)00125-9
[arXiv:hep-th/9611063 [hep-th]].

\end{thebibliography}\endgroup

\nolinenumbers

\end{document}